\documentclass{article}
\usepackage{graphicx} 
\usepackage{xcolor} 
\usepackage{CJKutf8} 
\usepackage{float}
\usepackage{amssymb} 
\usepackage{amsmath} 
\usepackage{amsthm}
\usepackage{bm}
\usepackage{geometry} 
\usepackage{hyperref} 
\usepackage{url} 
\usepackage{verbatim} 
\usepackage{authblk} 
\usepackage{enumitem} 
\usepackage{caption}
\usepackage{subcaption}
\newcommand{\subfigref}[2]{\ref{#1}~(\subref{#2})}
\usepackage{multirow}
\usepackage[round]{natbib}
\let\cite\citep

\DeclareMathOperator{\Var}{Var}
\DeclareMathOperator{\Cov}{Cov}

\newtheorem{theorem}{Theorem}

\newtheorem{lemma}[theorem]{Lemma}

\hypersetup{colorlinks=true,linkcolor=blue,filecolor=blue,urlcolor=blue,citecolor=cyan}
\title{Addressing zero-inflated and mis-measured functional predictors in scalar-on-function regression model} 
\author[1]{Heyang Ji}
\author[2]{Lan Xue}
\author[3]{Ufuk Beyaztas}
\author[1]{Roger S. Zoh}
\author[4]{Jeff Goldsmith}
\author[5]{Mark Benden}
\author[1]{Carmen D. Tekwe}
\affil[1]{Department of Epidemiology and Biostatistics, Indiana University School of Public Health, Bloomington, Indiana, USA}
\affil[2]{Department of Statistics, Oregon State University College of Science, Corvallis, Oregon, USA}
\affil[3]{Department of Statistics, Marmara University Faculty of Sciences
, Istanbul, Turkiye}
\affil[4]{Department of Biostatistics, Columbia University Mailman School of Public Health, New York, New York, USA}
\affil[5]{Department of Environmental \& Occupational Health, Texas A\&M University School of Public Health, College Station, Texas, USA}

\date{}

\begin{document}
\begin{CJK}{UTF8}{gbsn} 
\setlength{\parindent}{0pt} 

\maketitle


\begin{abstract}
Wearable devices are often used in clinical and epidemiological studies to monitor physical activity behavior and its influence on health outcomes. These devices are worn over multiple days to record activity patterns, such as step counts recorded at the minute level, resulting in multi-level, longitudinal, high-dimensional, or functional data. When monitoring patterns of step counts over multiple days, devices may record excess zeros during periods of sedentary behavior or non-wear times. Additionally, it has been demonstrated that the accuracy of wearable devices in monitoring true physical activity patterns depends on the intensity of the activities and wear times. While work on adjusting for biases due to measurement errors in functional data is a growing field, relatively less work has been done to study the occurrence of excess zeros along with measurement errors and their combined influence on estimation and inference in multi-level scalar-on-function regression models. We propose semi-continuous modeling approaches to adjust for biases due to zero inflation and measurement errors in scalar-on-function regression models. We provide theoretical justifications for our proposed methods and, through extensive simulations, we demonstrated their finite sample properties. Finally, the developed methods are applied to a school-based intervention study examining the association between school day physical activity with age- and sex-adjusted body mass index among elementary school-aged children.
\end{abstract}

\section{Introduction}\label{sec:intro}

With the rapid advancement of digital health technologies \cite{friend2023wearable}, wearable devices have increasingly been used for continuous monitoring and collection of biobehavioral data. Examples include physical activity \cite{teixeira2021wearable, coughlin2016use, strain2020wearable, wang2022effectiveness, phillips2018wearable}, blood glucose levels \cite{rodbard2016continuous, klonoff2005continuous, mastrototaro2000minimed, klonoff2017continuous, juvenile2008continuous}, sleep \cite{de2019wearable,muaremi2013towards, scott2020systematic, tobin2021challenges, sargent2018well}, and ambulatory blood pressure monitoring \cite{pickering2006ambulatory, o2018ambulatory, turner2015ambulatory, huang2021ambulatory, pickering1996recommendations}. These measures are subsequently employed in epidemiological studies to evaluate their impact on health outcomes \cite{mcdonough2021health, phillips2018wearable, huang2021ambulatory}. Physical activity data gathered from Actigraph devices is an illustrative example\cite{troiano2023evolution}.

The ActiGraph GTX3+, utilized in the National Health and Nutrition Examination Survey (NHANES), is a wearable device that captures participants' triaxial data on device acceleration at sampling frequencies from 30Hz to 100Hz \cite{NHANES1}, with activity summaries like step counts reported in epochs with a duration of 60 seconds \cite{ActiGraph-GT3X}. Wearable devices typically use shorter epoch durations to improve the accuracy of estimating intervals of intense physical activity \cite{ayabe2013epoch}. Shortened epochs are beneficial for real-time tracking, while extended epochs may distort brief yet critical fluctuations. For instance, intense exercise conducted over a few minutes, when recorded with a protracted epoch (such as one hour), might lead to only a slight increase in the average heart rate, failing to accurately depict the actual exertion period. These continuously gathered data, collected multiple times per hour over several days, produce multi-level, high-dimensional or functional datasets with intricate error configurations and patterns and functional data analysis is commonly used to analyze such datasets. \cite{ramsay2009introduction}.

Functional data analysis is a statistical approach frequently used to analyze high-dimensional data that appear as functions or images \cite{gertheiss2024functional,wang2016functional}. 
When a variable is densely sampled across time or space within a certain domain, it can be treated as a function and is usually assumed to be square-integrable \cite{bosq2000linear,ramsay2009introduction}. 
Conventionally, densely collected data are designated as functional data, in contrast with sparsely sampled data, which may be described as longitudinal data depending on the context and modeling strategy \cite{wang2016functional}. As noted, continuous monitoring yields high-dimensional data where one or more variables are measured at numerous time points. 
This leads to correlations among measurements from the same subject at different times, often with a complex structure and a resulting joint distribution that is challenging to model. Capturing such complexity typically requires statistical models with a large number of parameters, i.e., high degrees of freedom. One effective strategy to address this issue is to represent the data as realizations of functions, leveraging the flexibility of infinite-dimensional function spaces such as the $L^2$ space \cite{ferraty2006nonparametric}. Functional data analysis provides a principled framework for analyzing such data, where high-dimensional observations are treated as smooth functions over time or space. In practice, these functions are approximated by linear combinations of a finite set of basis functions \cite{morris2015functional, ramsay2005functional}.
These functional representations can then be incorporated into regression frameworks, including generalized functional linear models, to explore associations between functional predictors and health outcomes.

Generalized linear regression models with scalar responses and predictors have been extended to generalized functional linear regression models, where predictors can be either functional or scalar, and responses may similarly be scalar or functional \cite{muller2005generalized, reiss2017methods}. In these models, truncated complete basis such as the Fourier basis, functional principal components from the $L^2$ space, or linearly independent functions like B-splines are used to approximate the functional components 
\cite{reiss2017methods, crainiceanu2024functional}. In classical linear regression, the effect of each scalar predictor is represented by a scalar coefficient; in contrast, functional regression represents the effect of a functional predictor by a functional coefficient. To incorporate functional predictor into a linear model, the inner product in the $L^2$ space between the parameter function and the predictor function is used in place of the product of a scalar coefficient and a scalar predictor.

Monitoring behavioral patterns such as physical activity can lead to observations that are susceptible to measurement error \cite{rothney2008validity, jeffries2014physical, kozey2010accelerometer}. In functional data, observed values are often contaminated by random noise (i.e., measurement error) 
\cite{wang2016functional}. Measurement error contamination in functional data has been widely recognized in the literature \cite{yao2005functional, hall2006properties}
\cite{bunn2018current, feehan2018accuracy, atkinson1998statistical, bassett2012calibration, ferrari2007role}. 
Multiple sources contribute to measurement error in physical activity data monitored by accelerometers. By definition, physical activity is bodily movement produced by skeletal muscles that results in energy expenditure \cite{caspersen1985physical, world2019global}. Thus, data collected by wearable devices represent an estimation of physical activity based on detected movements. Various factors including monitor placement, individual characteristics (e.g., body weight), suboptimal dynamic range and sampling frequency, external disturbances (e.g., vehicular accelerations), activity type (e.g., walking, cycling, weight lifting), and differences in proprietary algorithms can affect the data output and introduce measurement error \cite{plasqui2013daily, plasqui2007physical, butte2012assessing, robertson2011utility, corder2007accelerometers, rowlands2004validation}.

Measurement errors may occur due to unobserved or unrecorded values or inaccuracies in measurement. Typically, in many studies, these errors are treated as independent both between and within individuals, often incorporated as an additive error term in regression models \cite{carroll1995measurement, fuller2009measurement}. Even when these errors have zero expectations, they can still introduce bias into parameter estimates. Naive averaging or smoothing approaches often fail to fully correct such bias and introduce attenuation in the estimated the effect of the error-prone variable. Therefore, more sophisticated methods are required to appropriately account for measurement error. Various functional data analysis methods have been developed to address measurement error \cite{crambes2009smoothing, zhang2023partially, tekwe2022estimation}, each tailored to specific assumptions about the error characteristics. For instance, some methods presuppose the errors are additive, independent of the true value, and exhibit short-range temporal correlation \cite{cardot2013confidence}. In certain scenarios, bias arising from measurement error is corrected by conceptualizing the error as an additive stationary Gaussian process with a zero mean, akin to white noise, on top of the genuine signal \cite{hall2006on}. A variety of methods have been developed to address measurement error in functional data, including smoothing-based techniques \cite{ullah2013applications, cardot2013confidence, florens2015instrumental} and instrumental variable approaches \cite{tekwe2022estimation, jadhav2022function, zoh2024bayesian}, each relying on different assumptions about the nature of the error. However, not all measurement errors can be accurately characterized as white noise or similar simple zero-mean additive noise.
Functional data may exhibit even more complex error structures. For example, measurement error problem can mix with zero inflation problem where the measurement errors are not additive to the true values, instead, the observed values are generated via a zero-inflated surrogate or observed measures. Zero inflation indicates an unexpectedly high frequency of zero values in the dataset \cite{lambert1992zero, ridout1998models, tu2006zero}. In count data, this means that zeros occur more frequently than predicted by standard models (e.g., Poisson or negative binomial) \cite{tait2012modelling,campbell2021consequences,shao2023zero}. Continuous data can also be zero-inflated, often leading to violations of model assumptions and computational difficulties \cite{tu2006zero}. In some cases, zero values correspond to missing data, while in others they represent true zeros. For example, when wearable devices record zero activity during periods of inactivity or sedentary behavior \cite{jeffries2014physical, troiano2008physical, matthews2008amount}. 
To address excess zeros—particularly in scalar or non-functional data—researchers commonly employ mixture or semicontinuous models \cite{lee2010analyze, wang2024generalized, li2005longitudinal, mills2013adjusting, tooze2002analysis}.  

While numerous methods have been developed to address bias from measurement error in zero-inflated scalar data \cite{li2005longitudinal, mills2013adjusting, tooze2002analysis, wang2024generalized}, much less work has been done in the context of functional data. In particular, methods that simultaneously account for both zero inflation and measurement error in functional predictors remain scarce.
In this manuscript, we propose a mixture model to handle zero-inflated functional data with measurement error and develop parameter estimation methods that correct for the bias introduced by such errors. We demonstrate the robustness of our approach through simulation studies. The proposed method can extract meaningful insights from complex health and behavioral datasets. We apply our methods to data collected from children in a Texas school district to assess the relationships that school-based physical activity have with age- and sex- adjusted BMI. This dataset serves as a valuable real-world case study for evaluating how our method can uncover hidden patterns, assess intervention outcomes, and improve predictive modeling in health-related educational research. 

This manuscript is structured as follows. Section~\ref{sec:method} introduces a multi-level generalized functional scalar-on-function regression model capable of handling zero-inflated functional predictors with measurement errors. Here, we detail the statistical model, the assumptions, estimation techniques, and the theoretical aspects of our approach. Section~\ref{sec:simulation} describes a simulation study designed to assess our method's effectiveness compared to other methods. Section~\ref{sec:application} demonstrates the application of our approach using data from obtained from elementary school-aged children from a Texas school district, examining the link between physical activity and body mass index (BMI) changes, adjusting for relevant demographics. Finally, section~\ref{sec:C&D} highlights our key findings, while the discussion explores potential improvements and future directions for extending the methodology to address similar challenges.

\section{Generalized scalar-on-function regression model with zero inflated functional predictors prone to measurement errors}
\label{sec:method}


We now define our proposed model along with the relevant assumptions. 
Let $\{Y_i, X_i(t), Z_i\}$ denote the data tuple for individual $i$ ($i=1,\dots,n$), where, $Y_i$ is a scalar-valued response, $X_i(t)\in L^2(\Omega)$ is a functional predictor defined for $t\in\Omega$, and  $Z_i\in\mathbb{R}^p$ is a vector of scalar predictors, with each element representing a distinct predictor. 
The functional variable $X_i(t)$ is latent and is observed indirectly via replicate measurements $W_{ij}(t)$ for $j=1,\dots,J_i$. 
Here, \( Y_i \) represents a health outcome, \( W_{ij}(t) \) represents the observed physical activity variable, \( X_i(t) \) is the long-term expectation of \( W_{ij}(t) \), and \( Z_i \) includes additional subject-level covariates such as age and sex.
We propose the following models that frame the relationship of these varaibles:
\begin{align}
    &Y_i \sim \mathrm{EF}(\mu_i,\phi), \nonumber \\
    &\mu_i = \mathbb{E}\{Y_i\mid X_i(t),Z_i\} = g^{-1}\Bigg\{\int_\Omega \beta(t)X_i(t)dt + (1,Z_i^\top)\gamma\Bigg\}, \label{eq:s-f model}\\
    &W_{ij}(t) = I\Bigl[G_{ij}(t) < \Phi^{-1}\{p_i(t)\}\Bigr]\cdot\Biggl\{\frac{X_i(t)}{p_i(t)} + U_{ij}(t)\Biggr\}, \label{eq:ME model}
\end{align}
where $\mu_i$ denotes the conditional mean of $Y_i$ given $X_i(t)$ and $Z_i$,
$g(\cdot)$ is a strictly monotonic link function, $p_i(t)$ satisfies $p_i(t)\in(0,1)$ for all $i$, 
$Y_i$ follows a distribution from the exponential family, with density function of the form
$f(y_i; \theta_i, \phi) = \exp\left\{ \frac{y_i \theta_i - b(\theta_i)}{a(\phi)} + c(y_i, \phi) \right\}$, $\theta_i$ is the canonical parameter, $\mu_i = b^\top (\theta_i)$ is the mean, $\phi$ is the dispersion parameter, and $a(\cdot)$, $b(\cdot)$, and $c(\cdot,\cdot)$ are known functions. We use the notation $Y_i \sim \mathrm{EF}(\mu_i, \phi)$ as a shorthand for this family.
The $p_i(t)$ serves as a activation probability, representing the probability of $W_{ij}(t)\neq 0$, and $\beta(t)\in L^2(\Omega)$ is the coefficient function associated with $X_i(t)$.
The scaling of $X_i(t)$ by the activation probability $p_i(t)$ in {Model~\eqref{eq:ME model}} is a key feature of our model, as it ensures that the observed proxy $W_{ij}(t)$ is an unbiased measurement of the true latent function $X_i(t)$ in expectation.
For the outcome distribution, when $Y_i$ is continuous we assume a Gaussian distribution as $\mathrm{EF}$ and adopt the identity link, i.e., $g(\mu)=\mu$. For binary outcomes, we assume a Bernoulli distribution as $\mathrm{EF}$ and use the logit link, i.e., $g(\mu)=\ln\left(\frac{\mu}{1-\mu}\right)$.
In contrast, $Y_i$ and $Z_i$ are assumed to be measured without error. 
Finally, $\{G_{ij}(t), t\in\Omega\}$ and $\{U_{ij}(t), t\in\Omega\}$ are assumed to be latent, independent Gaussian processes.

\subsection{Model assumptions}

Our objective is to estimate the parameters $\beta(t)$ and $\gamma$ in Model~\eqref{eq:s-f model}. 
Prior to estimation, to ensure that the proposed model adequately reflects the underlying data structure and supports valid statistical inference, we impose a set of assumptions. These assumptions provide realistic distributional conditions for the observed variables and are intended to facilitate the identifiability and estimability of model parameters.
\begin{enumerate}[label=\textbf{(A\arabic*)}, align=left]
    \item \label{As.G} For $G_{ij}(t)$,  we have
        $\mathbb{E}\{G_{ij}(t)\} = 0$ and $\Var\{G_{ij}(t)\} = 1$, for all $t\in\Omega$ and $ i,j$.
    $G_{ij}(t)$ are independent across individual index $i$ but can be correlated across replication index $j$. 
    $G_{ij}(t)$ have unspecified correlation structure across $t$. 
    $G_{ij}(t)$ is independent to $X_i(t)$ and $U_{ij}(t)$. 
    \item \label{As.4} $\mathbb{E}\{U_{ij}(t)\} = 0$  for all $t\in\Omega$ and $\Var\{U_{ij}(t)\} = \sigma_{u}^2(t)$.
    $U_{ij}(t)$ have unspecified correlation structure across $t$. 
    \item \label{As.6} The predictor process $\{X_i(t), t\in\Omega\}$ is Gaussian with
        $\mathbb{E}\{X_i(t)\} = \mu_x(t)$ and $\Var\{X_i(t)\} = \sigma_x^2(t)$.
    $X_i(t)$ have unspecified correlation structure across $t$. 
    \item For $p_i(t)$, we assume it is determined by the following model: 
    \begin{equation}
        \text{logit}\{p_i(t)\} = (1, Z_i^\top) \theta(t). \label{eq:Z2p}
    \end{equation} \label{As.7}
\end{enumerate}

Assumption~\ref{As.G} ensures that the binary term $I\Bigl[G_{ij}(t) < \Phi^{-1}\{p_i(t)\}\Bigr]$ follows distribution $\text{Be}\{p_i(t)\}$ for each $i$ and $t$, where $\textit{Be}$ stands for the Bernoulli distribution. The term $p_i(t)$ represents the probability of activation for individual $i$ at time $t$. Assumption~\ref{As.G} represents a relaxation of the stronger condition that $\{G_{ij}(t), t \in \Omega\}$ are independent across individuals and measurement occasions. While our methods are designed to work when this independence condition is met, they are also applicable to certain exceptions to this independence assumption.
Assumptions~\ref{As.4}-\ref{As.6} establish that for non-zero values, $W_{ij}(t)$ should have an expected value equating to $X_i(t)/p_i(t)$. Consequently, according to assumption~\ref{As.G}, $\mathbb{E}[W_{ij}(t)] = X_i(t)$. We impose assumption~\ref{As.7} due to noticeable variations in the proportion of zero readings in physical activity data among different groups (e.g., varying age groups) \cite{nobre2017multinomial}.
These observations suggest that the probability of observing a nonzero activity level, $\Pr\{W_{ij}(t) \ne 0\} = p_i(t)$, may not be constant across individuals. Instead, it may depend on subject-specific characteristics such as age, sex, or socioeconomic status. Therefore, we model the log-odds of a nonzero observation using a time-varying coefficient model, where scalar predictors $Z_i$ are allowed to influence $p_i(t)$ through the function $\theta(t)$.

\subsection{Predicting the unobserved predictor}\label{subsec:predict X}

With the model and necessary assumptions in place, we now proceed to the estimation of the model parameters. The goal is to develop an estimation procedure that leverages the structure of the model while respecting the imposed assumptions.
In an ideal scenario where all predictors in {Model~\eqref{eq:s-f model}} are fully observed, 
one could directly extend methods for solving scalar-on-scalar regression model to scalar-on-function regression model using basis expansion. 
However, the functional predictor, $X_i(t)$, is latent and is only observed indirectly through its proxy, $W_{ij}(t)$. This leads to a model with multiple high-dimensional latent variables, resulting in a likelihood function with complex integrals and treating $X_i(t)$ as a parameter making the maximum likelihood estimation computationally intensive. A practical alternative is to substitute $X_i(t)$ with a predictor $\widehat{X}_i(t)$ constructed from the observed measures or proxies, $W_{ij}(t)$. One important challenge is that measurements from the same individual at different values of $t$ may be correlated. 
Jointly predicting $X_i(t)$ across all time points would require specifying the full joint distribution of $X_i(t)$ and $W_{ij}(t)$, which involves assumptions about time-based correlation structures. If these assumptions do not accurately reflect the true process, the prediction may be adversely affected. On the other hand, predicting $X_i(t)$ independently at each time point avoids the need for these correlation structure assumptions and can significantly reduce the computational load.
A naive approach would be to use the sample mean, $\overline{W}_{i\cdot}(t)=\frac{1}{J_i}\sum_{j=1}^{J_i}W_{ij}(t)$, 
exploiting the property $\mathbb{E}\{W_{ij}(t)\}=X_i(t)$. 
However, when plugging this estimator as the substitution of $X_i(t)$ in to the Model~\eqref{eq:s-f model}, it does not fully leverage the information contained in $W_{ij}(t)$ and may result in both bias and inefficiency. 
Alternatively, one might use a single observation (e.g., $\widehat{X}_i(t)=W_{i1}(t)$, the measurement from the first day), but this approach is similarly suboptimal.
In this article, to predict the latent predictor $X_i(t)$, we propose two distinct but related pointwise strategies. The first leverages a mixed-effects model representation of the non-zero data, while the second employs a regression calibration framework to utilize the full conditional distribution of $X_i(t)$. They are designed to adjust for the biases associated with estimation. While both are developed under the same model assumptions and are designed for the same application context, they differ in the technical approach used for recovering $X_i(t)$.

\subsubsection{Pointwise two-stage mixed effect model approach}\label{MM approach}

Mixed effect models have been previously demonstrated as an approach for adjusting for measurement error \cite{luan2023scalable, huo2025comparison, xie2018generalized}. 
For two-staged-based methods, such as regression calibration \cite{spiegelman1997regression, hardin2003regression, wang1997regression}, $\mathbb{E}\{X_i(t)\mid W_{ij}(t), j = 1,\dots J_i\}$ are often used as a prediction for $X_i(t)$ under certain assumptions and conditions. 
Here, we propose a pointwise mixed effect model–based method to obtain the prediction of $X_i(t)$ from its noisy and zero-inflated proxy $W_{ij}(t)$ via the idea of conditional expectation. While our method focuses on correcting for measurement error in zero-inflated functional predictors, it is partially inspired by the modular inferential framework proposed by \citet{cui2022fast}, which fits pointwise mixed models followed by functional smoothing and joint inference.
We will explain why this method works and demonstrate that it employs the conditional expectation of $X$ given $W$ as the substitution. 

From our statistical model we deduce that, for a fixed time point \(t\), the observed variable \(W_{ij}(t)\) is distributed as a mixture of a Gaussian distribution and a degenerate distribution:
\begin{equation*}
    W_{ij}(t) \sim 
    \begin{cases}
        \mathcal{N}\Bigl\{\dfrac{X_i(t)}{p_i(t)},\,\sigma_u^2(t)\Bigr\} & \text{with probability } p_i(t),\\[1mm]
        0 & \text{with probability } 1-p_i(t).
    \end{cases}
\end{equation*}
Let $W_{ij}^*(t)$ denote the non-zero observations of $W_{ij}(t)$. Based on our model assumptions, we can decompose $p_i(t)\,W_{ij}^*(t)$ as
\begin{equation}
    p_i(t)\,W_{ij}^*(t) = \mu_x(t) + \Bigl\{X_i(t)-\mu_x(t)\Bigr\} + p_i(t)\,U_{ij}(t),\label{eq:ME-decom},
\end{equation}
where for any fixed $t$, $\mu_x(t)$ is a constant, and both $\{X_i(t)-\mu_x(t)\}$ and $p_i(t)\,U_{ij}(t)$ are Gaussian with zero mean and are independent. Consequently, equation~\eqref{eq:ME-decom} corresponds to the following mixed effect model:
\begin{equation}
    \mathcal{W}_{ij} = b_0 + b_i + \varepsilon_{ij},\label{eq:mixed M}
\end{equation}
subject to the distributional assumptions
\[
b_i \sim \mathcal{N}(0,d^2), \quad \Cov(b_i,b_k) = 0 \quad \text{for all }\, i\neq k,
\]
\[
\varepsilon_{ij} \sim \mathcal{N}(0,\sigma_i^2), \quad \Cov(\varepsilon_{ij},\varepsilon_{lm}) = 0 \quad \text{for } (i,j)\neq (l,m),
\]
\[
\Cov(b_i,\varepsilon_{lm}) = 0 \quad \text{for all }\, i, l, m.
\]
Here \(\mathcal{W}_{ij}=p_i(t)\,W_{ij}^{*}(t)\) is the response;  
the fixed intercept \(b_0\) represents \(\mu_x(t)\);  
the individual-specific random effect \(b_i\) represents \(X_i(t)-\mu_x(t)\); and  
the error term \(\varepsilon_{ij}\) represents \(p_i(t)\,U_{ij}(t)\).

Assuming the $p_i(t)$ is known, we fit the {Model~\eqref{eq:mixed M}} separately for each $t$. Let \(\widehat b_0\) denote the estimated $b_0$ and \(\widehat b_i\) denote the prediction of $b_i$.  
The resulting predictor of \(X_i(t)\) is then
\[
\widehat{X}_i(t) = \widehat{b}_0 + \widehat{b}_i.
\]
For $p_i(t)$, we can use its estimate to replace it. And we will propose the method to estimate $p_i(t)$. 

For mixed effect models, the best prediction of the random effect in terms of mean squared error is given by the conditional expectation of the random effect given the observed data \cite{chow2018advanced}. 
Therefore, this prediction represents the conditional expectation of $X(t)$ given the non-zero values of $W(t)$.

Note that the Model~\eqref{eq:mixed M} does not satisfy the Gauss-Markov conditions due to heteroskedasticity, thus the variance of the error term $\varepsilon_{ij}$ is not constant across observations and each individual may have a different variance. However, robust inference for fixed effects is well-established, and this heteroskedasticity does not compromise the quality of the estimates \cite{chow2018advanced}. Furthermore, simulation studies have demonstrated that the estimation of fixed effects and the prediction of random effects in mixed effect models are robust to violations of the underlying distributional assumptions \cite{mcculloch2011prediction, schielzeth2020robustness}. Consequently, by fitting the Model~\eqref{eq:mixed M} using estimation methods developed under the Gauss-Markov conditions, we can generate reliable predictions for $X_i(t)$.

\textbf{Estimation of $p_i(t)$.} \phantomsection\label{estimate p}
The $\mathcal{W}_{ij}$ in {Model~\eqref{eq:mixed M}} is defined as $p_i(t) W_{ij}^*(t)$ in {equation~\eqref{eq:ME-decom}}. 
Recall that $p_i(t)$ denotes the activation probability at time $t$ for subject $i$, representing the probability that the true latent process $X_i(t)$ is nonzero. Since $p_i(t)$ is latent, we replace $p_i(t)$ with its estimate $\widehat{p}_i(t)$.
Based on {Model~\eqref{eq:ME model}} and {Model~\eqref{eq:Z2p}}, we have
$\text{logit}[\Pr\{W_{ij}(t) \neq 0\}] = (1, Z_i^\top) \theta(t)$.
Therefore, we obtain $\widehat{p}_i(t)$ by fitting the logistic regression model
\begin{equation}
    \text{logit}\Bigl(\Pr [I\{W_{ij}(t) \neq 0]\}\Bigr) = (1, Z_i^\top) \theta(t), \label{eq:model fit p}
\end{equation}
for each $t$ separately.

In the assumption~\ref{As.G}, we assume that $G_{ij}(t)$ can be correlated across $j$. If we assume $G_{ij}(t)$ are independent across $j$ instead, we can use maximum likelihood methods to fit the model. Otherwise, we fit the logistic regression model by solving the generalized estimating equations (GEE).
Alternatively, we can use $\widehat{p}_i(t) = \frac{\sum_{j} I\{W_{ij}(t) \neq 0\}}{J_i}$ as the estimate of $p_i(t)$.
This approach is mathematically equivalent to excluding the predictor $Z_i$ and keeping only the intercept in the linear model described by {equation~\eqref{eq:model fit p}}. The advantage of the proportional approach is that it is simpler, distribution-free and remains valid even without assumption~\ref{As.7}, but the trade off is that it may be less efficient if covariates in $Z_i$ truly predict the activation probability.

\subsubsection{Pointwise regression calibration approach}\label{RC approach}

The second approach adopts a regression calibration \cite{fuller2009measurement} framework, which is a commonly used method to adjust for measurement error in predictors when fitting regression models. The key idea is to replace the unobserved true predictor with its conditional expectation given the observed, error-prone measurement. Specifically, if \( W \) denotes the observed predictor contaminated by error and \( X \) represents the true predictor, regression calibration approximates \( X \) by \( \mathbb{E}(X \mid W) \) and uses this estimate in the regression analysis. This approach effectively reduces the bias introduced by measurement error, particularly in linear models. In nonlinear models, such as logistic regression, regression calibration serves as an approximate correction and may still result in minor biases if the magnitude of measurement error is substantial \cite{carroll2006measurement}.

We propose a regression calibration estimation method that directly employs the conditional expectation of \(X_i(t)\) given \(W_{ij}(t)\) as the substitution for \(X_i(t)\). 
In our previously proposed mixed-effects model–based approach, we estimate \(p_i(t)\) based on whether \(W_{ij}(t) = 0\), and use only the non-zero observations of \(W_{ij}(t)\) to fit Model~\eqref{eq:mixed M} for predicting \(\widehat{X}_i(t)\). 
Thus, the information from zero and non-zero values of \(W_{ij}(t)\) is leveraged separately. 
In contrast, the regression calibration approach utilizes both zero and non-zero observations of \(W_{ij}(t)\) jointly by modeling the full conditional distribution of \(X_i(t)\) given \(W_{ij}(t)\), offering a more unified and efficient strategy for measurement error correction. 
\citet{wang2024generalized} proposed a regression calibration method to address the issue of error-prone zero-inflated predictor in generalized linear model. 
This work extend their approach to functional data context and adapts it to accommodate the proposed relationship between the latent and observed variables. According to our assumptions, for each $t$, the marginal distribution of \(X_i(t)\) is 
$X_i(t) \sim \mathcal{N}\{\mu_x(t),\sigma_x^2(t)\}$,
and we have
\begin{equation}
    \mathbb{E}\{X_i(t)\mid \widetilde{W}_i(t)\} 
    = \frac{\displaystyle \int_\chi x \prod_j \left( \Bigl[p_i(t)\,\varphi\{W_{ij}(t)\mid x/p_i(t),\sigma_u^2(t)\}\Bigr]^{\eta_{ij}(t)} \{1-p_i(t)\}^{1-\eta_{ij}(t)} \right) \varphi\{x;\mu_x(t),\sigma_x^2(t)\}\,dx}
    {\displaystyle \int_\chi \prod_j \left( \Bigl[p_i(t)\,\varphi\{W_{ij}(t)\mid x/p_i(t),\sigma_u^2(t)\}\Bigr]^{\eta_{ij}(t)} \{1-p_i(t)\}^{1-\eta_{ij}(t)} \right) \varphi\{x;\mu_x(t),\sigma_x^2(t)\}\,dx}, \label{eq:RC}
\end{equation}
where \(\widetilde{W}_i(t) = [W_{i1}(t),\dots,W_{iJ_i}(t)]^\top\), \(\eta_{ij}(t)=I\{W_{ij}(t)\neq0\}\), and \(\varphi(x;\mu,\sigma^2)\) denotes the probability density function of \(\mathcal{N}(\mu,\sigma^2)\).
When we compute $\mathbb{E}\{X_i(t)\mid \widetilde{W}_i(t)\}$ using equation~\eqref{eq:RC}, we replace the unknown parameters with their estimates. For $p_i(t)$, we use the estimating method discussed in pointwise two-stage mixed effect model approach part (Section~\ref{estimate p}). For \(\mu_x(t)\), \(\sigma_x^2(t)\), and \(\sigma_u^2(t)\), we propose the following estimators:
$$
    \widehat{\mu}_x(t) = \frac{\sum_i\sum_j W_{ij}(t)}{\sum_i J_i}, \ \ 
    \widehat{\sigma}_x^2(t) = \frac{\sum_i \left[\widehat{p}_i(t)\cdot\Bigl\{\overline{W}_{i\cdot}^*(t) - \overline{W}^*(t)\Bigr\}\right]^2}{n^*(t)-1},
$$
$$
    \widehat{\sigma}_u^2(t) = \frac{\sum_{i:\,r_i(t)>1} \left(\left[\sum_j I\{W_{ij}(t)\neq0\}\cdot \Bigl\{W_{ij}(t) - \overline{W}_{i\cdot}^*(t)\Bigr\}^2\right]/\{r_i(t)-1\}\right)}{n_1^*(t)},
$$
where
$$
    r_i(t) = \sum_j I\{W_{ij}(t)\neq0\}, \ \ 
    n^*(t) = \sum_i I\{r_i(t)\neq0\}, \ \ 
    n_1^*(t) = \sum_i I\{r_i(t)>1\},
$$
$$
    \overline{W}_{i\cdot}^*(t) = \frac{\sum_j W_{ij}(t)\, I\{W_{ij}(t)\neq0\}}{r_i(t)}, \ \ 
    \overline{W}^*(t) = \frac{\sum_{i: \, r_i(t)\neq0} \overline{W}_{i\cdot}^*(t)}{n^*(t)}.
$$

While the integrals in the numerator and denominator appear complex, they are univariate and can be efficiently computed using Monte Carlo methods.

\subsection{Scalar-on-function regression}\label{sec:s-f regression}

Following the prediction step of $X(t)$ described in section~\ref{subsec:predict X}, we address the computational challenges of functional regression by approximating $X(t)$ using a finite number of basis functions, thereby reducing the infinite-dimensional problem to a finite-dimensional one.
For a generalized linear regression model that incorporates a functional predictor along with $m$ scalar predictors, the link between the expected value of the response variable and the predictors for an individual $i$ is expressed by 
$$g[ \mathbb{E}\{Y_i|X_i(t),Z_i\}] = \int_\Omega \beta(t)X_i(t)dt + (1,Z_i^\top) \gamma,$$
where $\Omega$ denotes the temporal domain. For a complete basis of $L^2(\Omega)$ function space, denoted as $\{\rho_k\}_{k=1}^\infty$, there exists a sequence of coefficients $\{c_k\}_{k=1}^\infty$ such that 
$\int_\Omega \beta(t)X_i(t)dt = \int_{\Omega} \left\{\sum_{k=1}^\infty c_{k} \rho_k(t)\right\} X_{i}(t) dt = \sum_{k=1}^\infty c_{k} \int_{\Omega} \rho_k(t) X_{i}(t) dt$.
We can perform truncation by selecting a finite subset of the basis and then use $\sum_{k=1}^\infty c_{k} \int_{\Omega} \rho_k(t) X_{i}(t) dt$ to approximate $\int_\Omega \beta(t)X_i(t)dt$. Then the scalar-on-function regression model above turns into
$$g[ \mathbb{E}\{Y_i|X_i(t),Z_i\}] = \sum_{k=1}^K c_{k} \int_{\Omega} \rho_k(t) X_{i}(t) dt + (1,Z_i^\top) \gamma.$$ The term $\int_{\Omega} \rho_k(t) X_{i}(t) dt$ is then regarded as the predictors for the regression ($k=1,\dots,K$), effectively transforming the problem into a scalar-on-scalar generalized linear regression problem.
In practice, we don't necessarily need the basis functions to be orthogonal. We can employ other types of basis functions such as B-splines, provided they are not linearly dependent.
We use 
$\frac{\|\Omega\|}{|\mathcal{T}|}\sum_{t\in \mathcal{T}} \rho_k(t) X_{i}(t)$
to numerically compute the value of the integral
$\int_{\Omega} \rho_k(t) X_{i}(t) dt$
where $\|\cdot\|$ denote the Lebesgue measure, $\mathcal{T}$ is the finite set of time points where $X_{i}(t)$ is measured, 
and $|\mathcal{T}|$ represents the cardinality of $\mathcal{T}$. 

To provide theoretical support for the proposed regression framework, we analyze the asymptotic behavior of the proposed bias correction estimator in a scalar-on-function linear regression setting. 
Our main finding is stated in Theorem~\ref{theorem1}.
Prior to proving Theorem~\ref{theorem1} we first need to establish a preliminary result regarding the consistency of the least-squares estimator for the scalar-on-function linear regression model.

Consider the scalar-on-function linear regression model
\begin{equation}
  Y_i = \int_{\Omega} X_i(t) \beta(t)\,dt + \varepsilon_i, \quad i = 1, \dots, n,
  \label{eq:s-f model in theorem}
\end{equation}
where $\Omega$ is a compact interval in $\mathbb{R}$, $\|\beta\|_{L^2(\Omega)} < \infty$, 
and $\varepsilon_i$ are independent errors with $\mathbb{E}[\varepsilon_i] = 0$ and $\Var(\varepsilon_i) = \sigma^2 < \infty$.  

Let \(\{\rho_{k}\}_{k=1}^{\infty}\) be a complete orthonormal basis of \(L^{2}(\Omega)\).
For each subject \(i\) and index \(k\) define \(x_{ik}:=\langle X_{i},\rho_{k}\rangle =\int_{\Omega}X_{i}(t)\rho_{k}(t)\,dt\).
Define \(b_{k}:=\langle\beta,\rho_{k}\rangle\) and \(\bm{\theta}_{K}:=(b_{1},\dots,b_{K})^{\top}\).
Let $\widetilde{\bm{\theta}}_K := (\widetilde{b}_1, \dots, \widetilde{b}_K)^\top$ denote the ordinary least square estimator of $\bm{\theta}_K$, obtained by regressing $Y_i$ on $x_{i1}, \dots, x_{iK}$ with the corresponding function estimate defined as \(\widetilde{\beta}_{K}(t) := \sum_{k=1}^{K} \widetilde{b}_{k} \rho_{k}(t)\).

\begin{lemma}\label{th:lemma1}
Let $\{(X_i,Y_i)\}_{i=1}^n$ be an i.i.d.\ sample generated from model~\eqref{eq:s-f model in theorem}.  
Estimate the regression function using the first $K=K_n$ basis functions of $\{\rho_k\}_{k=1}^{\infty}$, where the truncation level $K_n$ is allowed to depend on~$n$.  
Assume that
\begin{enumerate}[label=\textbf{(B\arabic*)},align=left,itemsep=-3pt,topsep=-8pt]
  \item $K_n\to\infty$ and $K_n/n\to 0$ as $n\to\infty$;
  \item for every finite $K$, the covariance matrix $\Sigma_K := \Cov(x_{i1},\dots,x_{iK})$ is positive definite;
  \item $\mathbb{E}\!\bigl[\lVert X_i\rVert_{L^2(\Omega)}^{2}\bigr]<\infty$.
\end{enumerate}
Then the least-squares estimator $\widetilde{\beta}_{K_n}$ is $L^{2}$-consistent for $\beta$ on $\Omega$, that is,
\[
  \lVert\widetilde{\beta}_{K_n}-\beta\rVert_{L^2(\Omega)} \;\xrightarrow{P}\; 0,
  \qquad n\to\infty.
\]
(Undefined notation is summarized in Section~\ref{sec:notations}.)
\end{lemma}

Lemma~\ref{th:lemma1} is a well-known classical result establishing the consistency of the least-squares estimator for the scalar-on-function linear regression model. Its proof is therefore omitted.

We introduce an auxiliary process $W_{ij}(t)$, observed on a finite grid
$\mathcal{T}_{m}\subset\Omega$, which acts as a functional proxy for the
latent predictor $X_i(t)$ in model~\eqref{eq:s-f model in theorem}.
The proxy mechanism is governed by the following assumptions:
\begin{enumerate}[label=\textbf{(C\arabic*)},align=left,itemsep=-3pt,topsep=-8pt]
  \item\label{h1} For every fixed $t\in\Omega$, the conditional variables
        $\{W_{ij}(t)\mid X_i(t)=x\}\sim
        F_{W\mid X,t}(\,\cdot\mid x,t)$ are i.i.d.\ over
        $j=1,\dots,J_i$.
  \item\label{h2} The collections
        $\bigl\{X_i,Y_i,\{W_{ij}\}_{j}\bigr\}_{i=1}^n$
        are i.i.d.\ over $i = 1,\dots,n$.
  \item\label{h3} (Identifiability)
        For any fixed $t\in\Omega$ and $x_1\neq x_2$,
        $F_{W\mid X=x_1,t}\neq F_{W\mid X=x_2,t}$.
  \item\label{h4} Conditional independence:
        $Y_i \perp\!\!\!\perp W_i | X_i$.
\end{enumerate}

Let \(\{\mathcal{T}_{m}\}_{m\ge1}\) be a sequence of finite grids \(\mathcal{T}_{m}=\{t_{m,1},\dots,t_{m,m}\}\subset\Omega\). 
Let \(t_{(0)} < t_{(1)} < \dots < t_{(m+1)}\) denote the increasing arrangement of the grid points in \(\mathcal{T}_m\) together with the endpoints of \(\Omega\). Define the mesh width as \(\Delta_{m} := \max_{1 \le j\le m+1}\{t_{(j)} - t_{(j-1)}\}\). 
For every \(i\) and \(t\in\mathcal{T}_{m}\) write \({W}_{i}(t):=[W_{i1}(t),\dots,W_{iJ_{i}}(t)]^{\top}\).
Define 
\(\widehat{X}_{i}(t) := \mathbb{E}\{\,X_{i}(t)\mid{W}_{i}(t)\}\). 
For each $i$ and $k$, define estimated basis scores as \(\hat{x}_{ik}:=\frac{\|\Omega\|}{|\mathcal{T}_{m}|}\sum_{t\in\mathcal{T}_{m}}\widehat{X}_{i}(t)\,\rho_{k}(t)\).
Let $\widehat{\bm{\theta}}_K := (\widehat{b}_1,\dots,\widehat{b}_K)^\top$ denote the least-squares estimator of $\bm{\theta}_K$, obtained by regressing $Y_i$ on $\hat{x}_{i1},\dots,\hat{x}_{iK}$ with the corresponding function estimate defined as \(\widehat{\beta}_{K}(t) := \sum_{k=1}^{K} \widehat{b}_{k} \rho_{k}(t)\). We call \(\widehat{\beta}_{K}(t)\) a plug-in estimator because it is obtain from plugging \(\widehat{X}_i(t)\) as the substitution of \(X_i(t)\) into the least-squares estimator \(\widetilde{\beta}_{K}(t)\). 

\begin{theorem}\label{theorem1}
Consider model~\eqref{eq:s-f model in theorem} under
assumptions~\ref{h1}–\ref{h4} together with the conditions of
Lemma~\ref{th:lemma1}.
Let
\[
  \bigl\{Y_i,\; W_{ij}(t):t\in\mathcal{T}_{m},\; j=1,\dots,J_i\bigr\}_{i=1}^n
\]
be the observed sample, and estimate $\beta$ using the first
$K_n$ basis functions.
If, in addition,
\begin{enumerate}[label=\textbf{(D\arabic*)},align=left,itemsep=-3pt,topsep=-6pt]
  \item $\Delta_{m}\;\to\;0 \quad(m\to\infty)$;
  \item $\inf_{i}J_i\;\to\;\infty$;
  \item for every \(n\) and \(K\), both design matrices
        \[
          \widehat{\bm X}_{K}:=\bigl(\hat{x}_{ik}\bigr)_{n\times K},
          \qquad
          \bm X_{K}:=\bigl(x_{ik}\bigr)_{n\times K}
        \]
        have full column rank.
\end{enumerate}
Then the plug-in estimator $\widehat{\beta}_{K_n}$ is $L^{2}$-consistent on
$\Omega$; that is,
\[
  \lVert\widehat{\beta}_{K_n}-\beta\rVert_{L^2(\Omega)}
  \;\xrightarrow{P}\;0
  \qquad\text{as}\quad
  n\to\infty,\;
  \inf_{i}J_i\to\infty,\;
  m\to\infty.
\]
(Undefined notation is summarized in Section~\ref{sec:notations}.)
\end{theorem}

The proof of Theorem~\ref{theorem1} is given in Section~\ref{sec:proofs}.

In the Section~\ref{subsec:predict X}, 
we explained that both of our proposed estimators for $X_i(t)$ are built on the surrogate predictor $\mathbb{E}[X \mid W]$.  
Theorem~\ref{theorem1} further establishes that the resulting plug-in estimator is consistent for the scalar-on-function linear regression model.
For the broader class of generalized linear models (GLMs), estimation is typically carried out via maximum likelihood, with the choice of response distribution and link function determining the specific procedure.  In these settings the same consistency result continues to hold, and the proofs follow similar arguments along those presented here.  Due to limited space, the specifics of these proofs are not provided here.

\section{Simulation study}\label{sec:simulation}

To evaluate the finite sample performance of our proposed methods, we performed several simulation studies under various controlled settings. Specifically, we examined the accuracy of parameter estimation and the robustness of our methods under various levels of measurement error and different data-generating conditions, such as sample size, predictor correlation, and the prevalence of zero values. 
We compared our proposed mixed model based- and regression calibration-based approaches with other methods, including the (i) Naive average method which uses $\overline{W}_{i\cdot}(t)$ as the substitution of $X_i(t)$; (ii) Non-zero-inflated mixed model based method. This is also a mixed-effects model–based method. However, unlike the proposed approach in this article, it applies mixed-effects modeling without accounting for zero inflation. 
To distinguish it from the mixed model–based method proposed in Section~\ref{subsec:predict X}, we refer to this comparison approach as the "non-zero-inflated mixed model–based" / "non-ZI mixed model–based" / "non-ZI MM" method, while the mixed model–based method in Section~\ref{subsec:predict X} is referred to as the "zero-inflated mixed model–based" / "ZI mixed model–based" / "ZI MM" method, or simply the "mixed model–based" / "MM" method.
\citet{luan2023scalable} previously explored this method.
Our methods were also compared to the naive one day method which uses the first replicate of each individual, $W_{i1}(t)$, as the substitution of $X_i(t)$. We note that while we used the first replicate, the replicate for another day may also be employed for this approach. Additionally, we compared our methods with the benchmark approach which uses the true value of $X(t)$, highlighting the effectiveness of our scalar-on-function regression algorithm.

\parbox{\linewidth}{
We fit two distinct methods within the framework of our proposed mixed effects model and regression calibration approaches:

\textbf{(i) Pointwise approach:} Apply the model specified in equation~\eqref{eq:model fit p} to each $t$ in order to determine $\theta(t)$, and subsequently use the estimated value $\widehat\theta(t)$ in conjunction with $Z_i$ to compute $p_i(t)$. 

\textbf{(ii) Smoothed pointwise approach:} Similar to the pointwise method, fit the model in equation~\eqref{eq:model fit p} for each $t$ to find $\theta(t)$, but smooth $\widehat\theta(t)$ over $t$ before combining it with $Z_i$ to derive $p_i(t)$.
}
The pointwise approach offers flexibility and can capture local variations in $\theta(t)$, but it may yield noisy and unstable estimates. In contrast, the smoothed pointwise approach reduces variance and improves interpretability by borrowing strength across nearby time points, though it may introduce bias if $\theta(t)$ has abrupt changes. The choice depends on the expected smoothness of $\theta(t)$ and the trade-off between bias and variance.

\subsection{Simulation Settings}

We now detail the data-generating mechanisms used in our simulation study.
For evaluating the varying effects of the sample size $n$, we used $n = $ $50$, $100$, $200$, $500$, $1000$
The number of measurement replicates for each subject, $J_i$, is uniformly set to 7. for all $i$. 
The time domain $\Omega$ is set as $[0,1]$, within which we collect functional predictor measurements at 24 evenly spaced time points. 
Although in practice accelerometer based functional data are mostly minute-level, we used 24 equally spaced time points to represent a daily trajectory in our simulation study. This coarse grid was chosen for computational efficiency and clarity of illustration. Since the proposed method operates pointwise over time, its performance is not sensitive to the temporal resolution, provided that the underlying functional pattern is adequately captured.

We generate two error-free scalar predictors, denoted as $Z_c$ and $Z_b$. The predictor $Z_c$ follows a normal distribution $\mathcal{N}_n(0,\sigma_{c}^2 I_n)$, while $Z_b$ follows a binomial distribution $\text{Binom}(n,p_{b})$.
We let $\sigma_c^2 = 1$ and $p_b = 0.6$.
For $\gamma$ in Model~\eqref{eq:s-f model}, the chosen parameters are $(5, 0.2, 0.4)^\top$, where the first value represents the intercept, and the second and third values denote the slopes for the scalar predictor $Z_c$ and $Z_b$. 

Given the two scalar predictors $Z_c$ and $Z_b$, the vector $\theta(t)$ is composed of three components: $\theta(t) = [\theta_0(t), \theta_c(t), \theta_b(t)]^\top$. Each component, namely $\theta_0(t)$, $\theta_c(t)$, and $\theta_b(t)$, is derived using a stationary Gaussian process $\vartheta(t)$ with a mean of 0 and a variance of 1. The correlation is modeled by the function $\kappa(s,t) = \exp\{-\frac{(s-t)^2}{0.45}\}$. The value assigned to each $\theta(t)$ component is in the form $a+ d\cdot\Phi(\vartheta(t))$, where $\Phi^{-1}$ represents the inverse cumulative distribution function of a standard normal distribution. The generation of each component of $\theta(t)$ follows this independent yet identical approach, leading to an interval for $\theta(t)$'s components of $(a, a+d)$ across all $t \in \Omega$. For the simulation, we vary the zero value proportion in $W$ through adjustments in $a$ and $d$. Specifically, $d$ is set to 0.2, 0.8, and 1 for $\theta_0(t)$, $\theta_c(t)$, and $\theta_b(t)$, respectively. The offset $a$ is set to -0.4 and -0.5 for $\theta_c(t)$ and $\theta_b(t)$, respectively, with multiple settings of 0.3, 0.6, 0.8, 1 for $\theta_0(t)$. 
Once the $Z$ and $\theta(t)$ are generated, the $p_i(t)$ can be calculated. 
The resulting anticipated zero-value proportions in $W$ are 0.403, 0.335, 0.294, and 0.255.

In the model, we assume that  $X_i(t)$ are i.i.d Gaussian process over $i$. We did not specify the form of the mean and covariance function for $X_i(t)$. In the simulation study, for each individual $i$, we simulate $X_i(t)$ as a Gaussian process with a mean given by $\mathbb{E} \{X_i(t)\} = 4 + \sin{(1+2.8\pi t)}$ and a covariance described by $\Cov(X_s(t),X_i(s)) = \{1 + 0.1 \cos(-1+2.8\pi t)\}\cdot\exp{\{-\frac{25(s-t)^2}{2}\}}$. In this covariance function, the term $1 + 0.1 \cos(-1 + 2.8\pi t)$ defines a time-varying variance, while the Gaussian kernel $\exp\{-\frac{25(s - t)^2}{2}\}$ specifies a smooth correlation structure. This construction introduces mild nonstationarity while preserving local dependence, reflecting realistic features of functional data.
The functional parameter $\beta(t)$ is defined as $\beta(t) = \sin{(2\pi t)}$. We generate $G_{ij}(t)$ according to the equation: 
\begin{equation} G_{ij}(t) = q_g\, G_{0.i}(t) + \sqrt{1-q_g^2}\, G_{1.ij}(t),\label{eq:sim G} 
\end{equation} 
where $\{G_{0.i}(t),t\in\Omega\}$ and $\{G_{1.ij}(t),t\in\Omega\}$ are independent, stationary Gaussian processes for every $i$ and $j$, with properties: \( \mathbb{E}\{G_{0.i}(t)\} = \mathbb{E}\{G_{1.ij}(t)\} = 0 \quad \text{and} \quad \Var\{G_{0.i}(t)\} = \Var\{G_{1.ij}(t)\} = 1\), for all \(t\in\Omega \); where $0 \leq q_g < 1$. Both $G_{0.i}(t)$ and $G_{1.ij}(t)$ are modeled and simulated as stationary Gaussian processes with a mean of zero, a variance of one, and a correlation function given by $\kappa(s,t) = \exp\{-{50(s-t)^2}\}$. Various values for $q_g$ such as 0, 0.2, and 0.4 were used. The term $q_g$ controls the proportion of shared structure in the latent variable $G_{ij}(t)$, determining the correlation between different $j$ for the same $i$, and thus shaping the dependency structure of $I\{W_{ij}(t)\neq 0\}$. A larger $q_g$ leads to stronger correlations in whether $W_{ij}(t) = 0$ across $j$, while $q_g = 0$ implies that these events are independent.

The $U_{ij}(t)$ in the model is a zero-mean Gaussian process, thus, so it suffices to specify its covariance function for simulation. The covariance function can be expressed in the form $\sigma(s)\sigma(t)\kappa(s,t)$, where $\sigma(\cdot)$ denotes the pointwise standard deviation function and $\kappa(s,t)$ is a correlation function that captures the dependence structure between time points $s$ and $t$. 
In our simulation, we implement a constant function, $\sigma(t) = \sigma_u$, for $U_{ij}(t)$ and used several constant values for $\sigma_u$, including 1, 2, and 3. 
Several types of correlation function $\kappa$ are considered:
\begin{equation*}
\begin{aligned}
\text{(i)}\quad & \textbf{Squared exponential:} 
&& \kappa(s,t) = \exp\left\{-\frac{(t-s)^2}{2 \rho_u ^2}\right\}, \\
\text{(ii)}\quad & \textbf{Spatial power:} 
&& \kappa(s,t) = \rho_u^{|t-s|}, \quad 0 < \rho_u < 1, \\
\text{(iii)}\quad & \textbf{Two-value:} 
&& \kappa(s,t) = 
\begin{cases}
1      & \text{if } s = t \\
\rho_u & \text{else}
\end{cases}, \quad 0 < \rho_u < 1.
\end{aligned}
\end{equation*}
The corresponding covariance structures are autoregressive of order 1 (AR(1)) for the spatial power function and compound symmetry (also known as exchangeable correlation) for the two-value function.
For each correlation function, we used different values for the parameter $\rho$, which controls the strength and decay rate of the correlation over time $t$ (or location $s$ and $t$: for the squared exponential correlation, $\rho = 0.1, 0.15, 0.2$; for the spatial power correlation, $\rho = 0.2, 0.4, 0.6$; and for the two-value correlation, $\rho = 0.2, 0.4, 0.6$.

We calculate the $W_{ij}(t)$ once the $X_i(t)$, $U_{ij}(t)$, $G_{ij}(t)$, and $p_i(t)$ are generated. After selecting the link function $g(\cdot)$, the type of distribution $\mathrm{EF}$, and the dispersion parameter $\phi$, the response variable can also be simulated. The specified link function $g(\cdot)$ allows for the computation of $\mu_i$, which is the conditional expectation of $Y_i$ given $X_i(t)$ and $Z_i$. We simulate $Y_i$ with the specified values of $\phi$ and $\mathrm{EF}$. 

\parbox{\linewidth}{
In the simulation study, we generate two different types of response variable $Y$: 

\textbf{(i) Continuous}: The distribution $\mathrm{EF}$ is (unidimensional) Gaussian distribution.
    For this case, the link function $g(\cdot)$ used is identity function. 
    Then we have $$Y_i = \int_\Omega \beta(t)X_i(t)dt + (1,Z_i^\top) \gamma + \varepsilon_i$$
    where $(\varepsilon_1,\dots,\varepsilon_n)^\top\sim\mathcal{N}_n(0,\sigma_0^2 I_n)$.
    The $\phi$ is corresponding to the $\sigma_0^2$. 
    In this case, we use 0.02 as the value of $\sigma_0^2$. 

\textbf{(ii) Binary}: The distribution $\mathrm{EF}$ is Bernoulli distribution.
    For this case, the link function $g(\cdot)$ used is logit function. 
    Then we have $$\Pr(Y_i=1) = \int_\Omega \beta(t)X_i(t)dt + (1,Z_i^\top) \gamma.$$
    There is no $\phi$ that we need to specify its value. 
}

In our simulation studies, several factors were varied including: the sample size (default: 100); the proportion of zero values in $W$ (default: 0.335); the scale of the deviation of the measurement error term $U$, denoted by $\sigma_u$ (default: 1); and the correlation structure of $U$, including both the type of correlation function and the parameter $\rho_u$ within it (default: squared exponential correlation function with $\rho_u = 0.2$). We also varied the types of response variable and the value of $q_g$ in equation~\eqref{eq:sim G} (default: 0).
We investigated the effects of these factors on estimation result. 
Although the probability of zero values in $W(t)$ cannot be directly specified by a single parameter, it is implicitly controlled through the choice of the function $\theta(t)$ in Model~\eqref{eq:Z2p}.

In our simulation studies, we focus on the estimation of the functional parameter $\beta(t)$, as it is directly associated with the error-prone functional predictor $X(t)$ and is of primary inferential interest in our modeling framework.
The bias and variance of the estimate of functional parameter $\beta$ are defined as follows: 
\begin{align}    
    &\text{Bias:}\qquad\left[\int_\Omega \left\{\mathbb{E} \widehat\beta(t)-\beta(t)\right\}^2 dt\right]^{\frac{1}{2}} \label{eq:def bias2}\\
    &\text{Variance:}\qquad\int_\Omega E \left[\widehat\beta(t)-\mathbb{E}\{\widehat\beta(t)\} \right]^2 dt \label{eq:def var}
\end{align}
where $\widehat\beta(t)$ is the estimate of $\beta(t)$. 
We use the defined bias and variance to evaluate the goodness of the estimate for $\beta$. 


\subsection{Simulation Results}\label{simulation result}

Figure~\ref{fig:simulation} shows the average $\widehat\beta(t)$ of different estimation methods versus the true $\beta(t)$ under the default simulation setting: sample size $n = 100$, zero proportion in $W(t) \approx 0.335$, measurement error scale $\sigma_u = 1$, $q_g=0$, and squared exponential correlation function for $U(t)$ with $\rho_u = 0.2$.

\begin{figure}[H]
    \centering
    
    \begin{subfigure}[t]{0.48\textwidth}
        \centering
        \includegraphics[width=\textwidth]{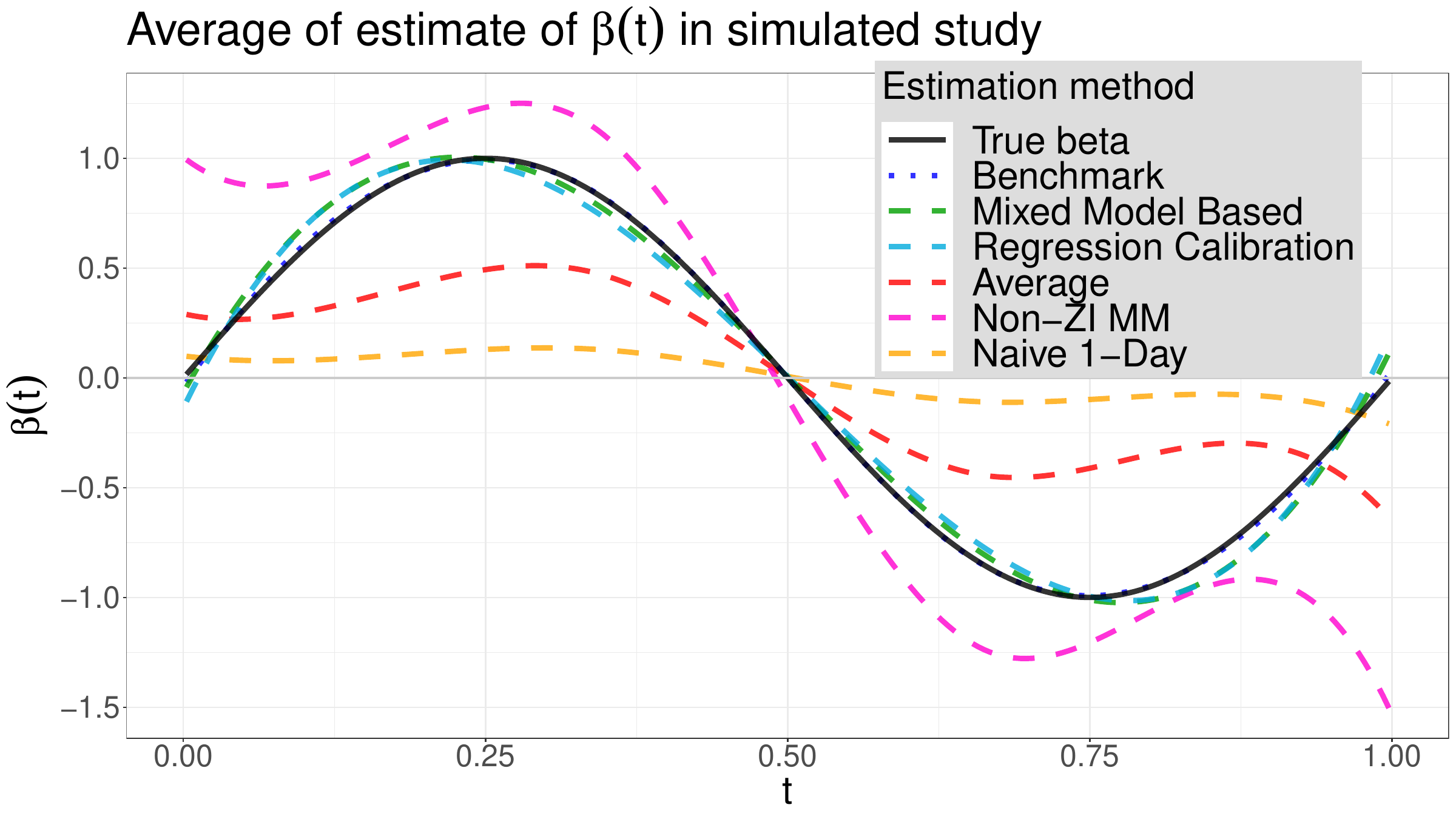}
        \caption{Continuous Outcome ($\mathrm{EF}$: Gaussian distribution)}
        \end{subfigure}
        \hfill
    \begin{subfigure}[t]{0.48\textwidth}
        \centering
        \includegraphics[width=\textwidth]{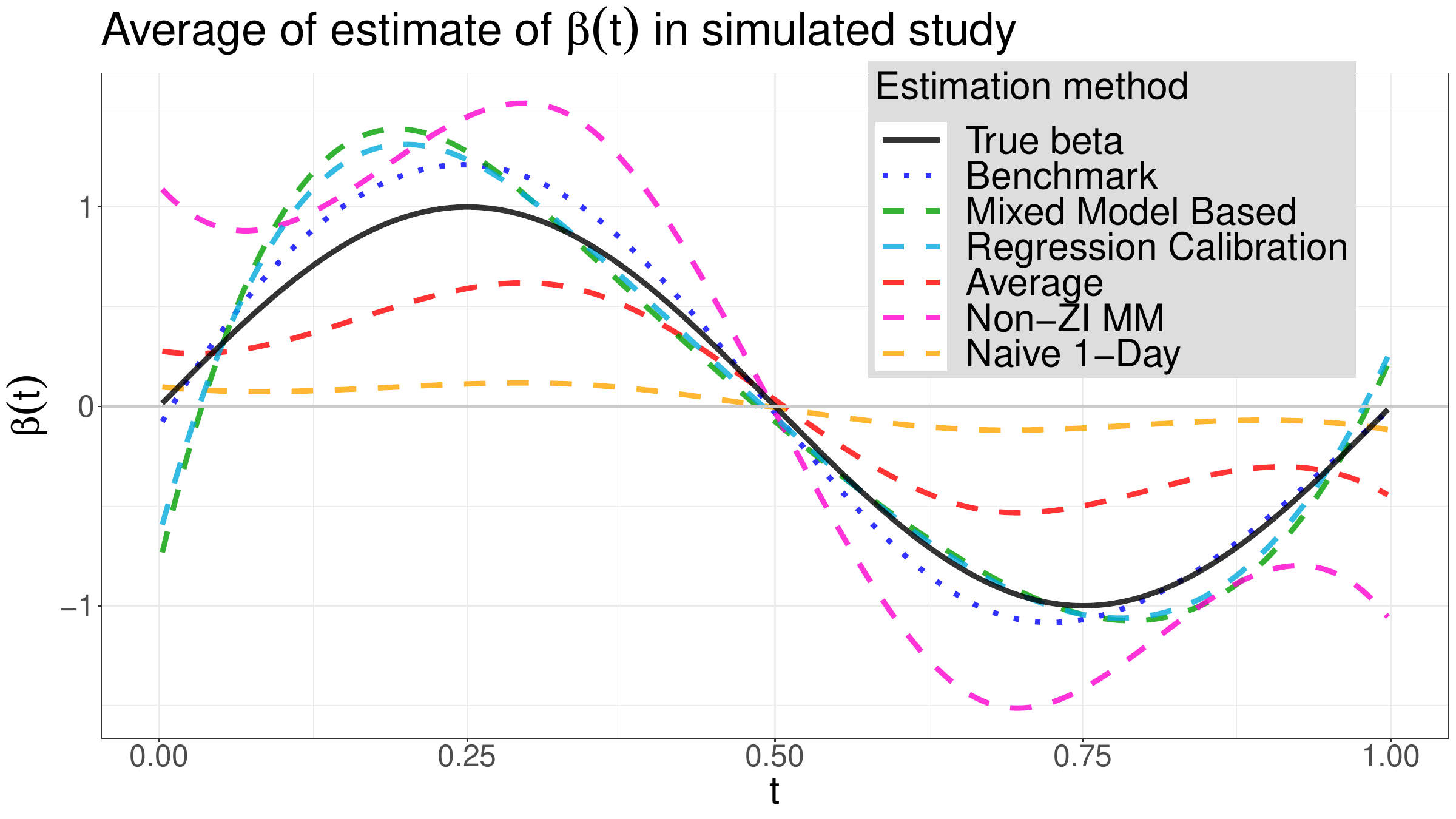}
        \caption{Binary Outcome ($\mathrm{EF}$: Bernoulli distribution)}
    \end{subfigure}
    
    \captionsetup{justification=raggedright}
    \caption[plot of simulation results for one setting]{
    True $\beta(t)$ (black solid) and the average estimated $\widehat\beta(t)$ (colored lines) under different methods in simulation study. (\textbf{Non-ZI MM} = non-zero-inflated mixed model based method.)
    \textbf{Common simulation settings (continuous outcome):} 
    {conditional distribution of $Y|X(t),Z$:} Gaussian distribution (with identity link); sample size: $n=100$; 
    {estimation method of $\theta(t)$:} pointwise; {proportion of zero values:} $\mathbb{E}\{\Pr(W_{ij}(t)=0)\}\approx0.335$; 
    {Covariance function of $U(t)$:} squared-exponential with $\rho_u=0.2$, $\sigma_u = 1$; $q_g=0$. 
    \textbf{Common simulation settings (binary outcome):} 
    {conditional distribution of $Y|X(t),Z$:} Bernoulli distribution (with logit link); other settings identical.
    }
    \label{fig:simulation}
\end{figure}


With all other settings held constant based on the values used for Fig~.\ref{fig:simulation}, while changing the sample size $N$, 
we obtained the results in the Table~\ref{tab:bias-var-combined}.
Table~\subfigref{tab:bias-var-combined}{tab:bias-cont} and Table~\subfigref{tab:bias-var-combined}{tab:bias-bin} demonstrates the average squared bias of $\widehat\beta(t)$ for different estimation methods defined in equation~\eqref{eq:def bias2}. 
Table~\subfigref{tab:bias-var-combined}{tab:var-cont} and Table~\subfigref{tab:bias-var-combined}{tab:var-bin} shows the average variance of $\widehat\beta(t)$ for different estimation methods defined in equation~\eqref{eq:def var}. 
The Fig~.\ref{fig:simulation} and Table~\ref{tab:bias-var-combined} present a subset of our simulation study results under a common model setting but with varying sample sizes. Additional results under various settings are provided in the Section~\ref{sec:more simulation results}, further demonstrating the robustness of the proposed methods across a wide range of scenarios. These scenarios encompass logistic regression models for the outcome, varying scales and error structures in the measured predictor, the use or omission of smoothing in estimating $\theta(t)$, and differing proportions of zero values in $W$.

Our simulation results indicate that the naive average approach significantly attenuates the estimation of $\beta(t)$. The "naive 1 day" method results in even more substantial attenuation. Additionally, the "no zero inflation" method yields a biased estimate for $\beta$. To address this, our bias correction techniques, including the mixed model-based and the regression calibration methods, can mitigate the extent of bias in the estimation of $\beta(t)$. However, the "non-ZI MM" method, which is the mixed-model-based bias correction approach without accounting for zero inflation, produces highly unstable estimates (with high variance) when the sample size is small.
The trade-off associated with our proposed bias correction methods is an increase in estimation variance, which corresponds to reduced efficiency. When comparing our two proposed bias correction methods, the mixed model approach yields less bias than the regression calibration approach when sample sizes are large, but more bias when sample sizes are small. Conversely, the regression calibration method results in lower variances for the estimated coefficient functions. Furthermore, in terms of computational speed, the mixed model method outperforms the regression calibration approach.

\begin{table}[H]
\centering
\small
\captionsetup{justification=raggedright}
\caption{Simulation results for six estimation methods across different sample sizes $N$.
Top two panels: continuous outcome; bottom two panels: binary outcome. Each pair shows average squared bias (first) and average variance (second).
\textbf{Common simulation settings (continuous outcome):} 
{conditional distribution of $Y|X(t),Z$:} Gaussian distribution (with identity link); 
{estimation method of $\theta(t)$:} pointwise; {proportion of zero values:} $\mathbb{E}\{\Pr(W_{ij}(t)=0)\}\approx 0.335$; 
{Covariance function of $U(t)$:} squared-exponential with $\rho_u=0.2$, $\sigma_u = 1$; $q_g=0$. 
\textbf{Common simulation settings (binary outcome):} 
{conditional distribution of $Y|X(t),Z$:} Bernoulli distribution (with logit link); other settings identical.
\textbf{Benchmark} = benchmark method; \textbf{Average} = naive average method; 
\textbf{MM} = mixed model based method proposed in Section~\ref{subsec:predict X}; \textbf{RC} = regression calibration method proposed in Section~\ref{subsec:predict X}; 
\textbf{Non-ZI MM} = non-zero-inflated mixed model based method; \textbf{1 day} = naive one-day method.
}
\label{tab:bias-var-combined}

\begin{subtable}{\linewidth}
  \centering
  \caption{Continuous $Y$: Average Squared Bias}\label{tab:bias-cont}
  \begin{tabular}{r@{\;}|rrrrrr}
    \hline
    \textbf{N} & Benchmark & MM & RC & Average & Non-ZI MM & 1 day \\
    \hline
    50 & 0.0001213 & 0.003945 & 0.005530 & 0.1497 & 0.2995 & 0.3897 \\ 
   100 & 0.0000903 & 0.003446 & 0.005038 & 0.1445 & 0.1689 & 0.3822 \\ 
   200 & 0.0001204 & 0.003259 & 0.005246 & 0.1489 & 0.1318 & 0.3840 \\ 
   500 & 0.0001262 & 0.002772 & 0.004994 & 0.1468 & 0.1190 & 0.3842 \\ 
  1000 & 0.0001110 & 0.002619 & 0.004417 & 0.1494 & 0.1225 & 0.3846 \\ 
    \hline
  \end{tabular}
\end{subtable}

\vspace{0.8em}

\begin{subtable}{\linewidth}
  \centering
  \caption{Continuous $Y$: Average Variance}\label{tab:var-cont}
  \begin{tabular}{r@{\;}|rrrrrr}
    \hline
    \textbf{N} & Benchmark & MM & RC & Average & Non-ZI MM & 1 day \\
    \hline
    50 & 0.0023843 & 0.135181 & 0.120203 & 0.0586 & 2.6930 & 0.0190 \\ 
   100 & 0.0010200 & 0.061922 & 0.049800 & 0.0269 & 0.3297 & 0.0085 \\ 
   200 & 0.0005190 & 0.034514 & 0.026972 & 0.0155 & 0.1311 & 0.0047 \\ 
   500 & 0.0002088 & 0.014465 & 0.010488 & 0.0069 & 0.0479 & 0.0018 \\ 
  1000 & 0.0000995 & 0.008700 & 0.005136 & 0.0046 & 0.0246 & 0.0010 \\ 
    \hline
  \end{tabular}
\end{subtable}

\bigskip

\begin{subtable}{\linewidth}
  \centering
  \caption{Binary $Y$: Average Squared Bias}\label{tab:bias-bin}
  \begin{tabular}{r@{\;}|rrrrrr}
    \hline
    \textbf{N} & Benchmark & MM & RC & Average & Non-ZI MM & 1 day \\
    \hline
    50 & 0.0001213 & 0.003914 & 0.005610 & 0.1497 & 0.2995 & 0.3897 \\ 
   100 & 0.0000903 & 0.003400 & 0.005045 & 0.1445 & 0.1689 & 0.3822 \\ 
   200 & 0.0001204 & 0.003118 & 0.005186 & 0.1489 & 0.1318 & 0.3840 \\ 
   500 & 0.0001262 & 0.002646 & 0.004883 & 0.1468 & 0.1190 & 0.3842 \\ 
  1000 & 0.0001110 & 0.002508 & 0.004310 & 0.1494 & 0.1225 & 0.3846 \\ 
    \hline
  \end{tabular}
\end{subtable}

\vspace{0.8em}

\begin{subtable}{\linewidth}
  \centering
  \caption{Binary $Y$: Average Variance}\label{tab:var-bin}
  \begin{tabular}{r@{\;}|rrrrrr}
    \hline
    \textbf{N} & Benchmark & MM & RC & Average & Non-ZI MM & 1 day \\
    \hline
    50 & 0.0023843 & 0.135726 & 0.119692 & 0.0586 & 2.6930 & 0.0190 \\ 
   100 & 0.0010200 & 0.062859 & 0.050002 & 0.0269 & 0.3297 & 0.0085 \\ 
   200 & 0.0005190 & 0.034582 & 0.026807 & 0.0155 & 0.1311 & 0.0047 \\ 
   500 & 0.0002088 & 0.014532 & 0.010509 & 0.0069 & 0.0479 & 0.0018 \\ 
  1000 & 0.0000995 & 0.008773 & 0.005209 & 0.0046 & 0.0246 & 0.0010 \\ 
    \hline
  \end{tabular}
\end{subtable}

\end{table}

\section{Application to School-Based Intervention Data from Texas, US}\label{sec:application}

Having demonstrated the performance of the proposed method through simulation studies, we now apply it to a real dataset to illustrate its practical utility. Our objectives are (i) to evaluate how the method performs under realistic conditions, (ii) to examine the consistency between our proposed bias correction approaches, (iii) to determine whether they yield results that differ significantly from those produced by alternative methods. The dataset used in this application was collected from a school-based intervention study conducted in Texas, United States, which investigated the impact of stand-biased desks on children’s health outcomes. Here, we use this dataset to illustrate the application of our statistical method.

Childhood obesity remains a pressing public health issue in the United States. Between 1999–2000 and 2021–2023, its prevalence among children and adolescents aged 2–19 years rose from 13.9\% to 21.1\%, while severe obesity doubled from 3.6\% to 7.0\% \cite{hales2020quickstats}. This trend is concerning due to its association with persistent obesity into adulthood \cite{simmonds2016predicting} and reduced life expectancy \cite{dietz1998childhood, mossberg198940}. Traditional classroom environments, where students remain seated for long periods, contribute to sedentary behavior and may exacerbate the problem.
To address this, researchers at Texas A\&M University, led by Dr. Mark Benden, launched a series of studies beginning in 2009 to evaluate the impact of stand-biased desks in elementary schools \cite{benden2011impact}. These studies examined whether modifying classroom environments could promote physical activity, increase energy expenditure, and reduce the body mass index (BMI) among students. Results showed that students using stand-biased desks had significantly higher energy expenditure and step counts during class compared to those in traditional seated classrooms \cite{benden2011impact, benden2014evaluation}. 
A follow-up study further reported a significant reduction in BMI percentiles among students who used stand-biased desks for two consecutive years relative to those who remained in seated classrooms \cite{wendel2016stand}. These findings suggest that redesigning classroom settings to incorporate stand-biased desks may serve as an effective, school-based intervention for reducing sedentary time and addressing childhood obesity.

Leveraging data from the Texas elementary school intervention studies, we apply our proposed analytical method to evaluate its effectiveness in extracting meaningful insights from real-world settings. The dataset includes key health and behavioral metrics, such as changes in BMI, energy expenditure, and physical activity levels measured via wearable accelerometers.
A notable characteristic of the step count data is its zero-inflation: a substantial proportion of observations are zeros, with the proportion varying markedly across both time points and individuals. This pattern aligns with a key assumption of our model, that the probability of observing a zero depends on both time and individual-specific factors—as illustrated in Figure~\ref{fig:PrZeros}.
\begin{figure}[htbp]  
  \centering 
  \includegraphics[width=0.6\textwidth]{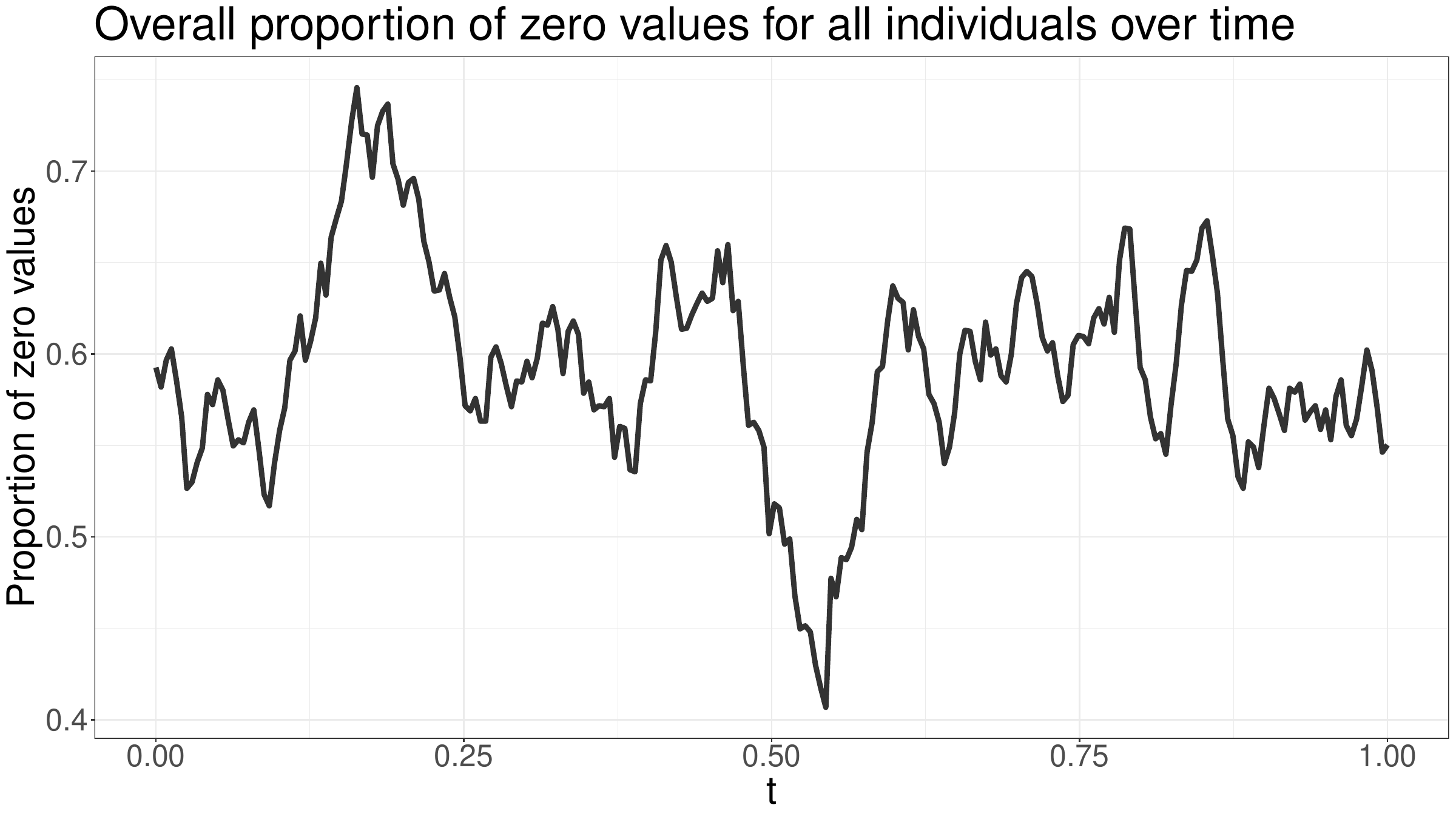}   
  \caption{The proportion of zero values of individuals in the step count data collected from elementary school children by Benden et al.}  
  \label{fig:PrZeros} 
\end{figure}

We focus on two primary outcomes: the change in BMI over the study period (\emph{delta\_BMI}, a continuous variable), and an indicator of whether a participant’s BMI decreased during the study (\emph{BMI\_reduced}, a binary outcome). For each participant, these outcomes were derived from anthropometric measurements taken at the beginning and end of their intervention period.
This setting offers an ideal case study for evaluating the impact of classroom-based interventions, such as stand-biased desks, on child health. By applying our modeling framework, we aim to identify patterns in the data, assess the intervention’s effectiveness, and demonstrate the practical applicability of our approach for addressing public health questions in real-world educational environments. The following analysis showcases how our method processes, models, and interprets the data.

In this study, participants wore wearable devices for three to seven days. We focused on minute-level activity data recorded between 10:00 AM and 2:00 PM on each day the device was worn. For analysis, only data within this four-hour window were retained. To facilitate functional data analysis, we linearly rescaled the time variable such that 10:00 AM corresponds to 0 and 2:00 PM to 1, mapping all measurements to the unit interval $[0, 1]$.

We applied the proposed bias correction methods, along with the comparison methods from the simulation study, to estimate the effect of step count on health outcomes among elementary school children, adjusting for age, sex, and ethnicity.
These methods were previously evaluated through simulation to assess their relative performance.
The full models are specified as follows:
\begin{align*}
\text{delta\_BMI} &= \int_{\Omega} \beta(t) \, \text{step count}(t) \, dt + \gamma_0 + \gamma_1 \, \text{age} + \gamma_2 \, \text{gender} + \gamma_\text{ethnicity} + \varepsilon, \\
\text{logit}\Pr(\text{BMI\_reduced}) &= \int_{\Omega} \beta(t) \, \text{step count}(t) \, dt + \gamma_0 + \gamma_1 \, \text{age} + \gamma_2 \, \text{gender} + \gamma_\text{ethnicity} + \varepsilon.
\end{align*}
We used bootstrap sampling to construct pointwise confidence intervals for the estimated coefficient function $\beta(t)$.

Figures~\ref{fig:ra.1} and~\ref{fig:ra.2} present the main application results.
Figures~\subfigref{fig:ra.1}{subfig:ra_pe} and~\subfigref{fig:ra.2}{subfig:ra_pe_bY} display the estimated coefficient functions $\beta(t)$, representing the effect of step count on \emph{delta\_BMI} and \emph{BMI\_reduced}, respectively. Figures~\subfigref{fig:ra.1}{subfig:ra_avg}-\subfigref{fig:ra.1}{subfig:ra_1d} and~\subfigref{fig:ra.2}{subfig:ra_avg_bY}-\subfigref{fig:ra.2}{subfig:ra_1d_bY} show the point estimates and confidence intervals of $\widehat\beta(t)$ obtained from the various methods described in the simulation study for these two outcomes.
The results in Figures~\ref{fig:ra.1} and~\ref{fig:ra.2} demonstrate strong agreement between our two proposed bias correction methods, while also revealing substantial differences compared to the other methods. In particular, the results highlight clear distinctions between methods that account for zero inflation and those that do not, underscoring the importance of explicitly modeling zero inflation in this context.
Across all methods, the estimated coefficient functions $\widehat{\beta}(t)$ for step count displayed certain time intervals where the pointwise bootstrap confidence intervals did not include zero, indicating statistically significant time-specific associations. 
The locations of these significant intervals were broadly similar across methods, while the proposed bias correction methods yielded effect estimates with smaller absolute magnitudes than the others—consistent with the expected attenuation from correcting for measurement error and zero inflation. 
For the scalar predictors (age, sex, and ethnicity), the estimated effects were small in magnitude and their confidence intervals included zero for all methods, indicating no statistically significant associations with either outcome. 
This lack of significance is likely due to the short follow-up period of no more than two weeks for each participant, during which substantial changes in BMI were unlikely to occur, as well as the relatively homogeneous distributions of these demographic variables in the sample, which limit variability for detecting associations. 
Nevertheless, it is important to include these predictors in the model, as they are established demographic and clinical covariates that may confound or modify the association of interest. 
Adjusting for them ensures that the estimated effect of step count is not biased by these background factors and can also improve statistical efficiency, even if their individual effects are not statistically significant in this dataset \citep{etminan2021adjust, morris2022planning}.
Although the time-varying proportion of zero values is plotted in Figure~\ref{fig:PrZeros} for reference, its trend does not exhibit pointwise correspondence with the estimated $\beta(t)$. This is expected, as $\beta(t)$ is estimated as a single smooth function, where the value at each time point depends not only on local information but also on data across the entire time domain. Therefore, the differences in $\beta(t)$ reflect changes in the global modeling structure, rather than being solely driven by local zero proportions. 
Overall, these findings demonstrate the effectiveness of our proposed methods in real-world applications. We conclude with a summary of key results and a discussion of their broader implications.

\begin{figure}[H]
\centering

\begin{subfigure}[t]{0.47\textwidth}
    \centering
    \includegraphics[width=\textwidth]{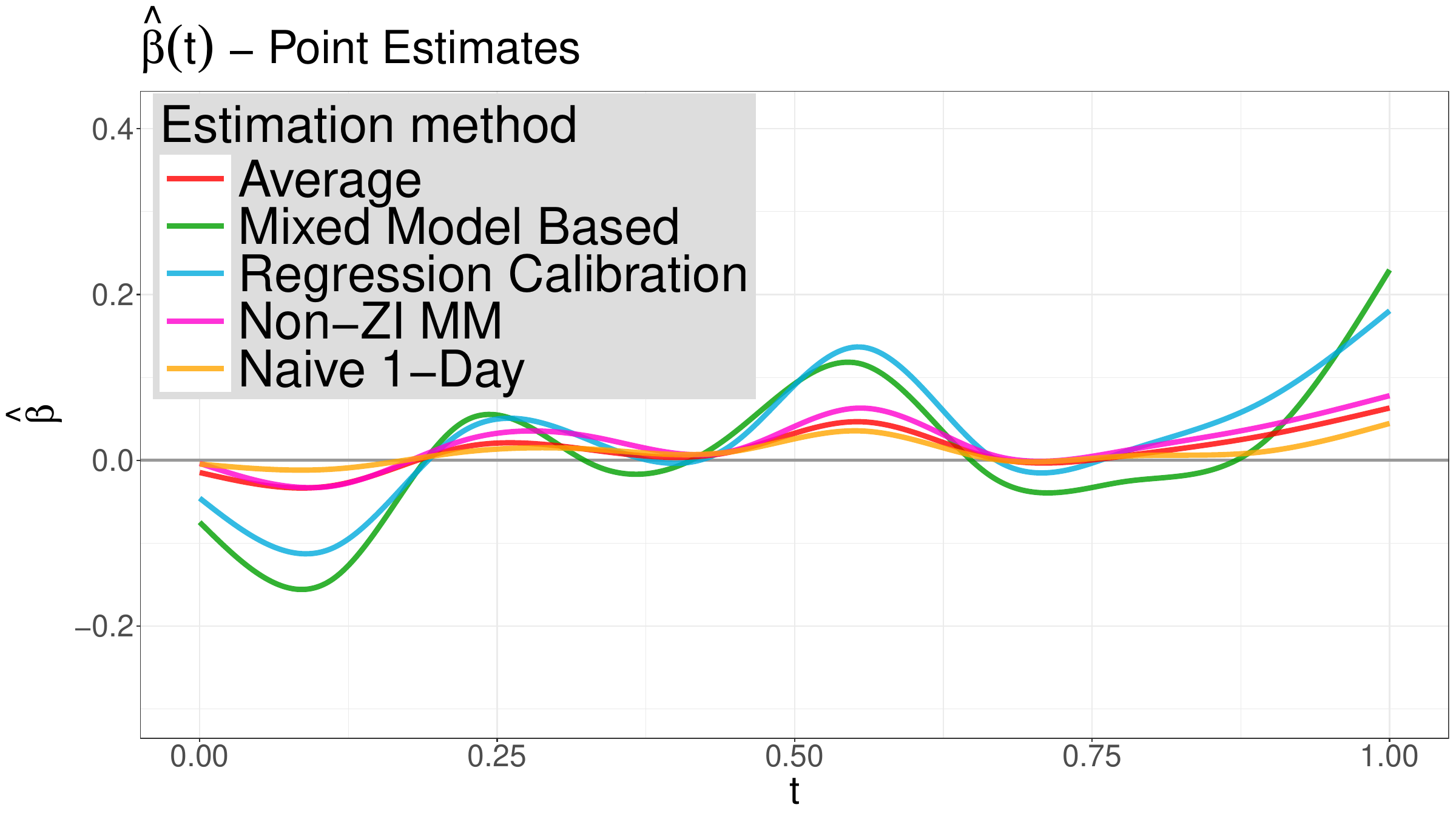}
    \caption{Point estimate from Average, mixed model, and regression calibration method}
    \label{subfig:ra_pe}
\end{subfigure}
\hfill
\begin{subfigure}[t]{0.47\textwidth}
    \centering
    \includegraphics[width=\textwidth]{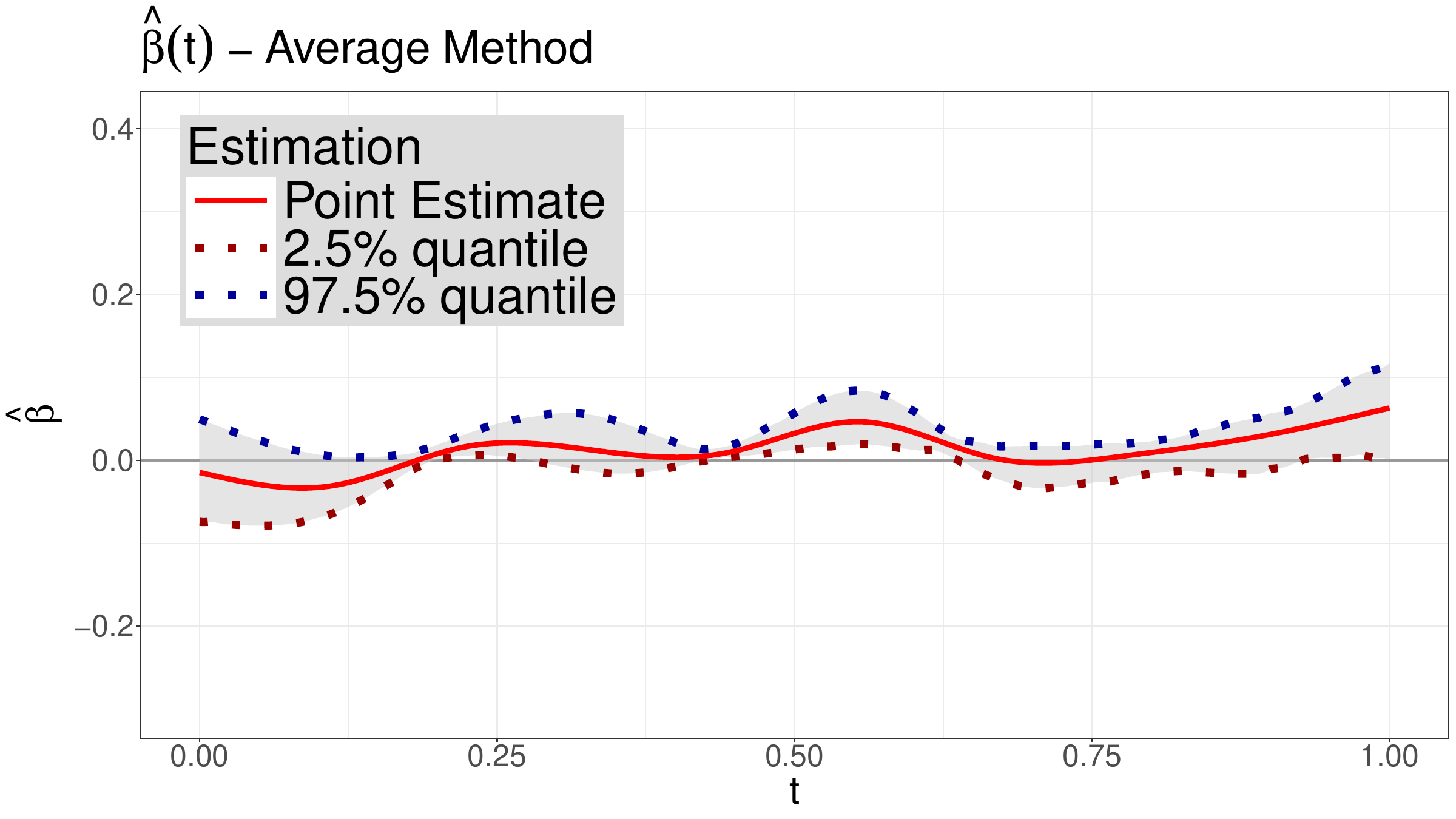}
    \caption{Estimate and confidence interval from average method}
    \label{subfig:ra_avg}
\end{subfigure}

\vspace{1em}

\begin{subfigure}[t]{0.47\textwidth}
    \centering
    \includegraphics[width=\textwidth]{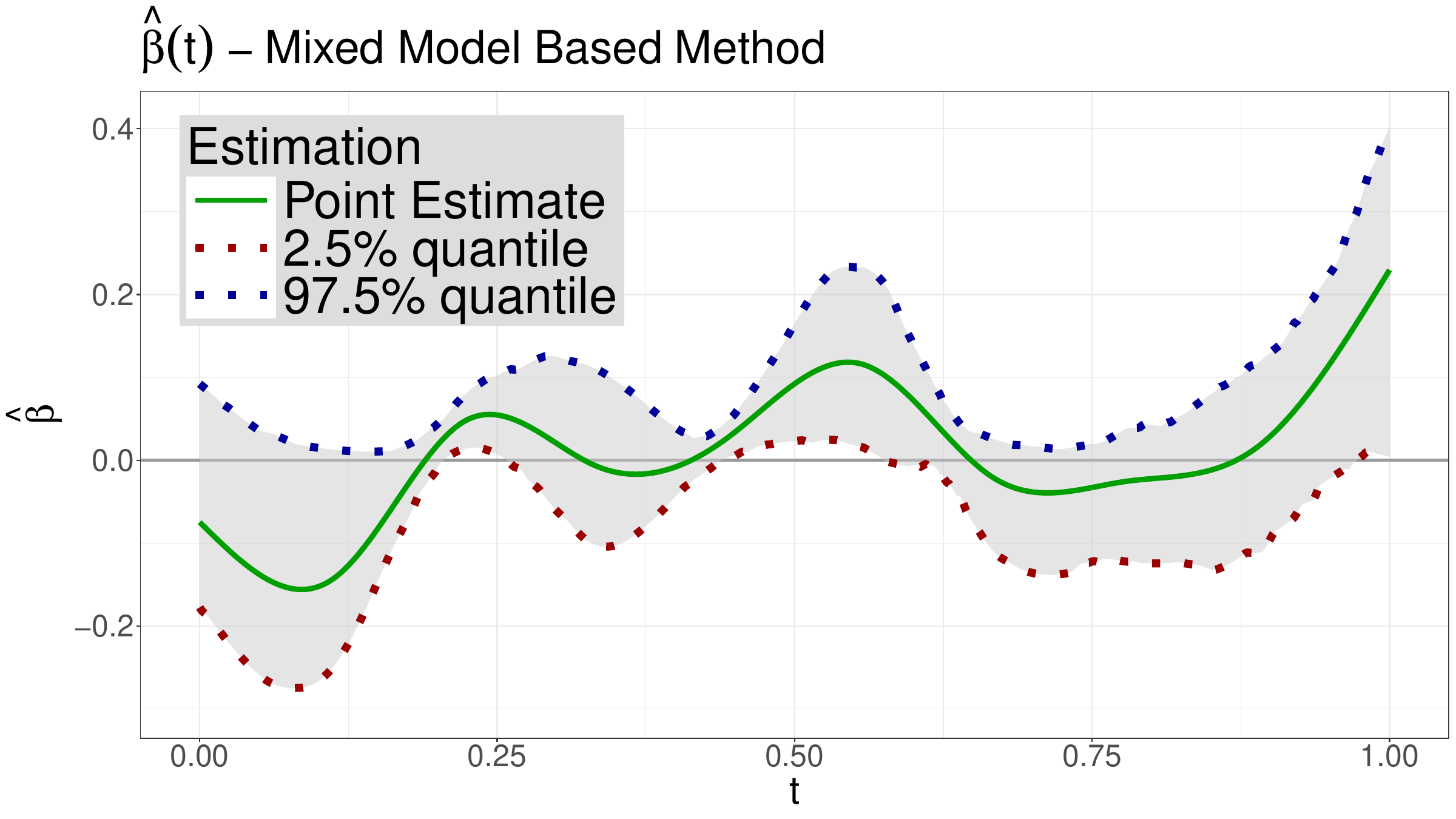}
    \caption{Estimate and confidence interval from mixed model based method}
    \label{subfig:ra_mm}
\end{subfigure}
\hfill
\begin{subfigure}[t]{0.47\textwidth}
    \centering
    \includegraphics[width=\textwidth]{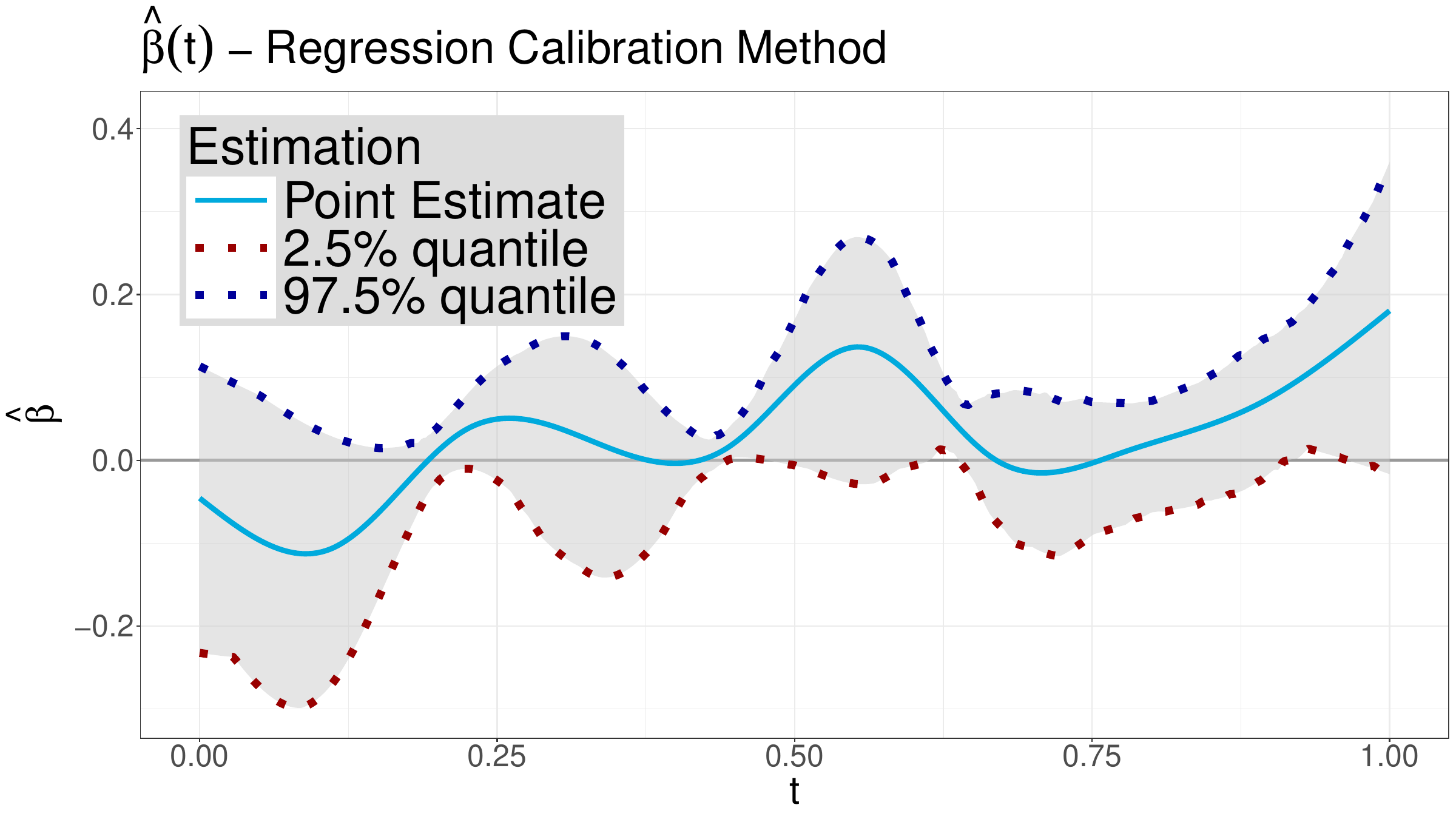}
    \caption{Estimate and confidence interval from regression calibration method}
    \label{subfig:ra_rc}
\end{subfigure}

\vspace{1em}

\begin{subfigure}[t]{0.47\textwidth}
    \centering
    \includegraphics[width=\textwidth]{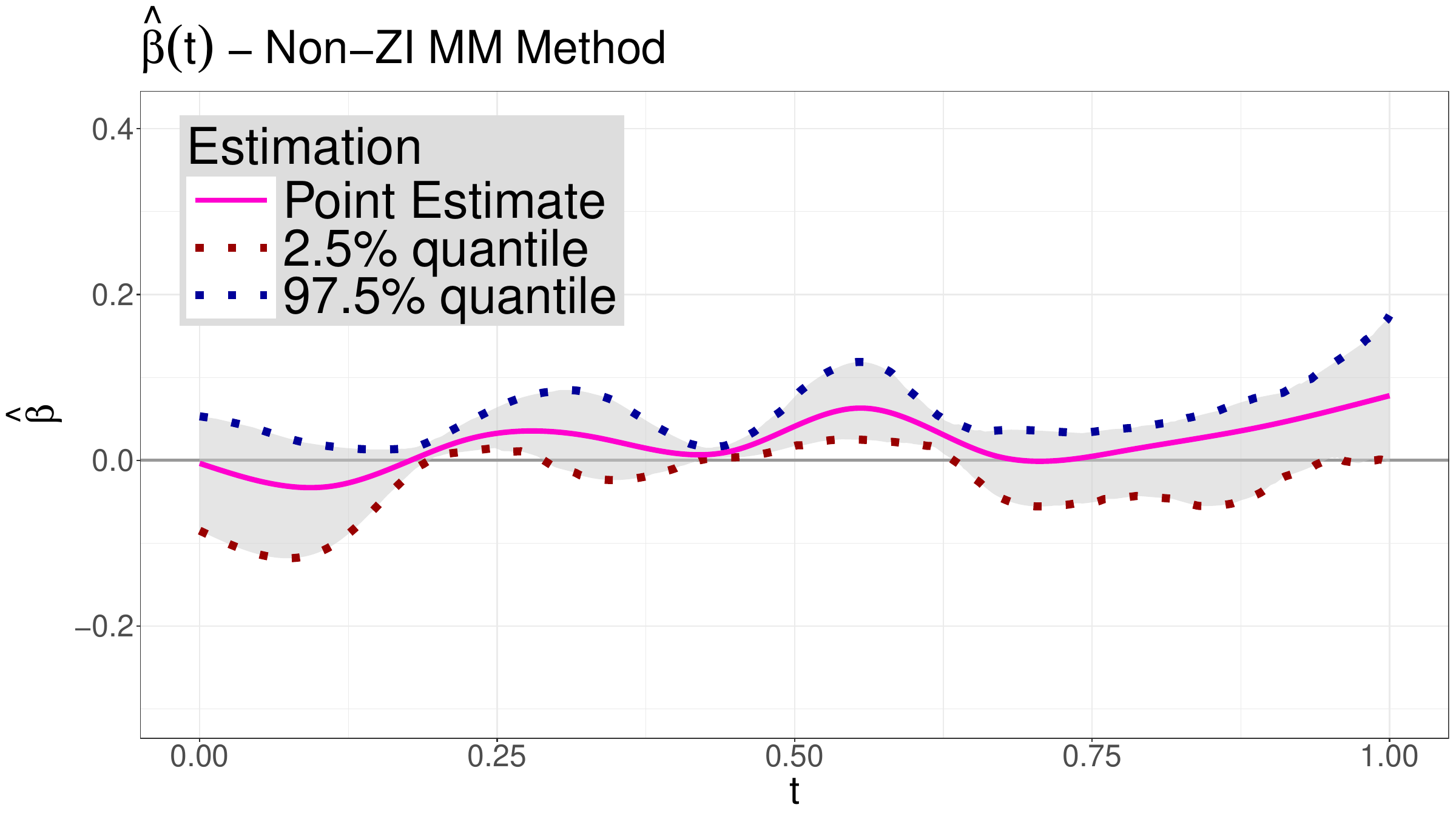}
    \caption{Estimate and confidence interval from no zero inflation mixed model based method}
    \label{subfig:ra_nz}
\end{subfigure}
\hfill
\begin{subfigure}[t]{0.47\textwidth}
    \centering
    \includegraphics[width=\textwidth]{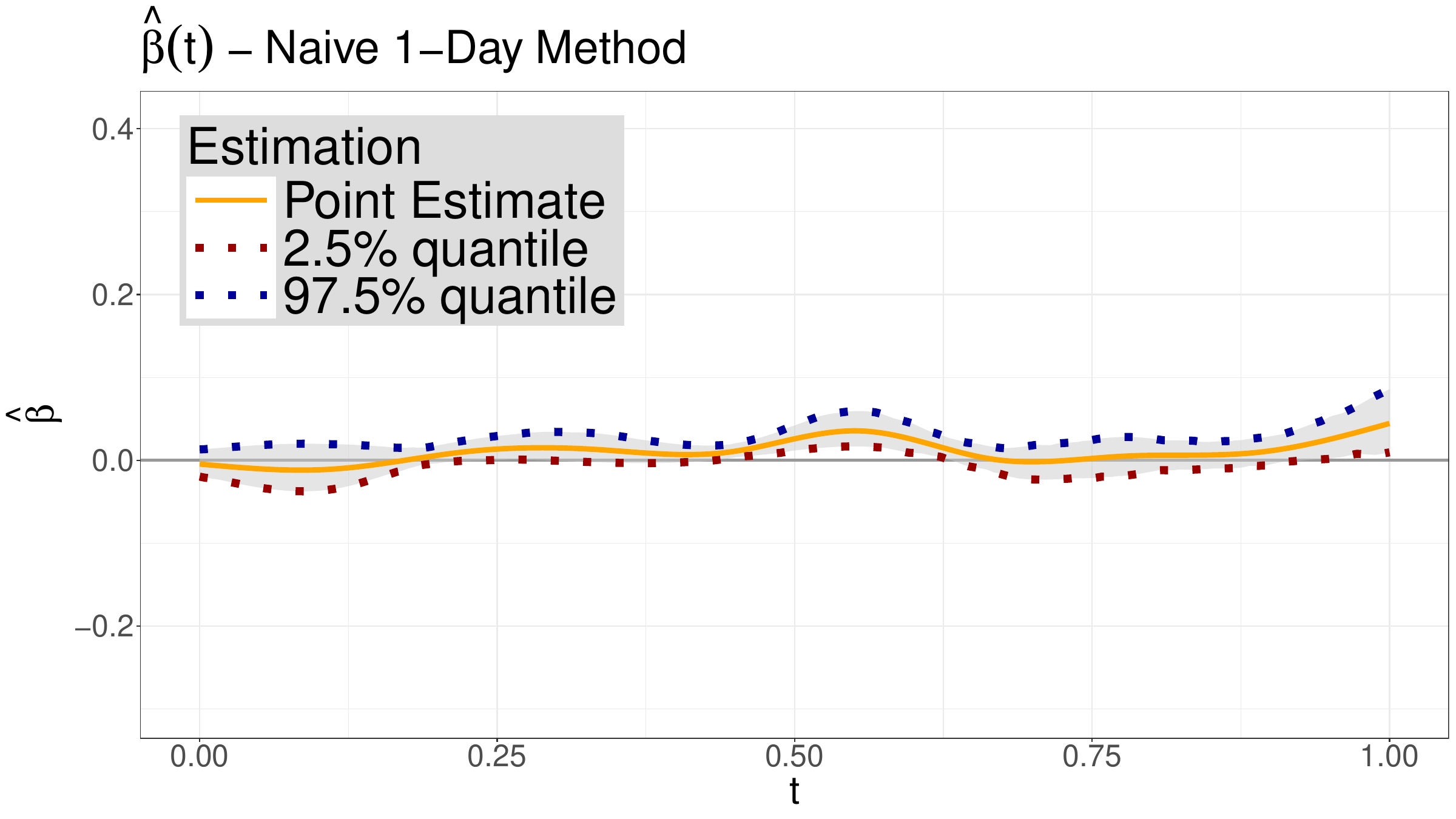}
    \caption{Estimate and confidence interval from one-day method}
    \label{subfig:ra_1d}
\end{subfigure}

\caption{Estimated effect of step counts on \emph{delta\_BMI} adjusted for age, sex, and ethnicity. Panel (\subref{subfig:ra_pe}) shows the point estimate from different methods. Panels (\subref{subfig:ra_avg})–(\subref{subfig:ra_1d}) show the point estimate and confidence interval of $\widehat\beta(t)$ from each method. Panels (\subref{subfig:ra_mm}) and (\subref{subfig:ra_rc}) correspond to our proposed bias correction methods.}
\label{fig:ra.1}
\end{figure}

\begin{figure}[H]
\centering

\begin{subfigure}[t]{0.47\textwidth}
    \centering
    \includegraphics[width=\textwidth]{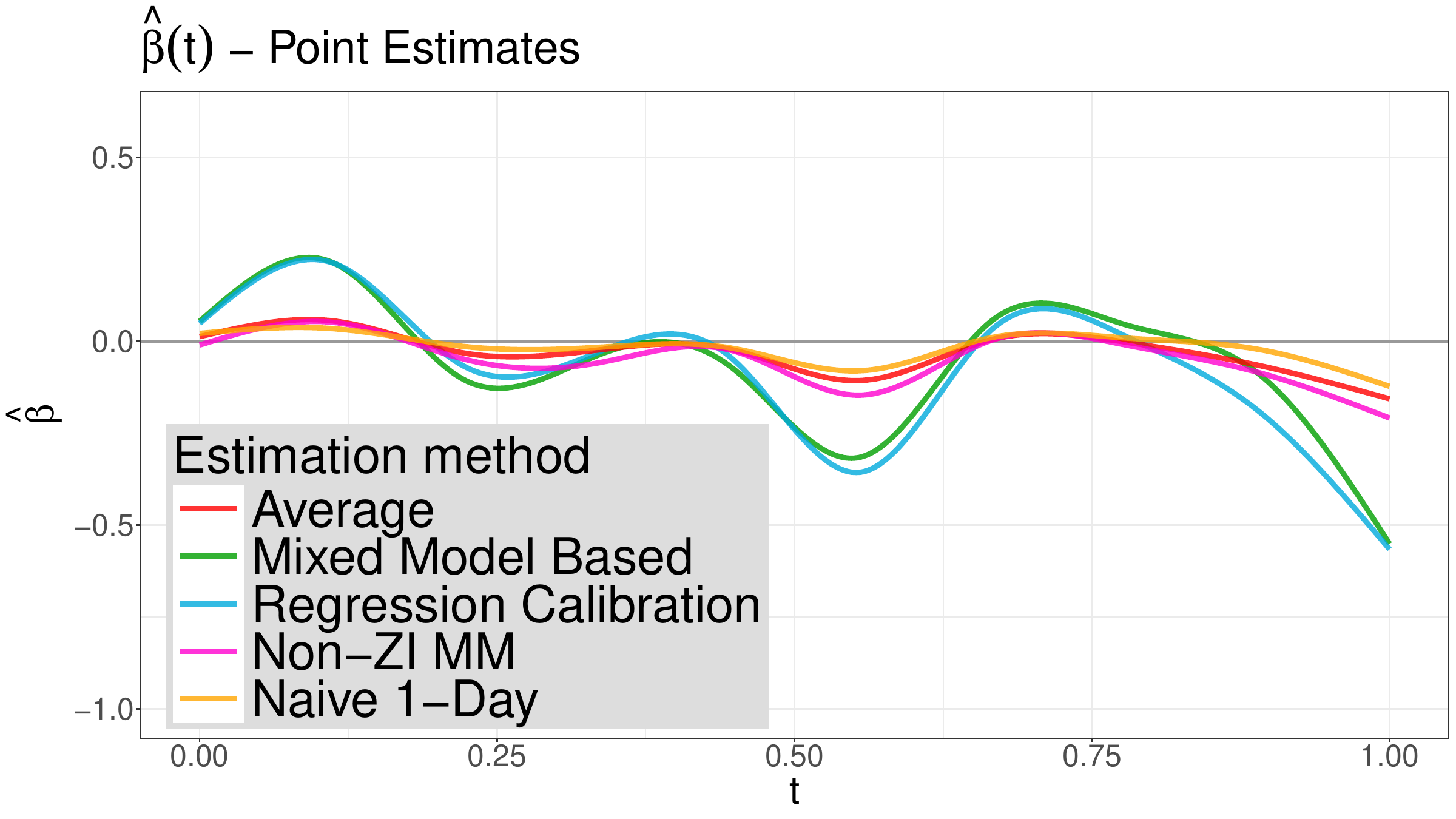}
    \caption{Point estimate from Average, mixed model, and regression calibration method}
    \label{subfig:ra_pe_bY}
\end{subfigure}
\hfill
\begin{subfigure}[t]{0.47\textwidth}
    \centering
    \includegraphics[width=\textwidth]{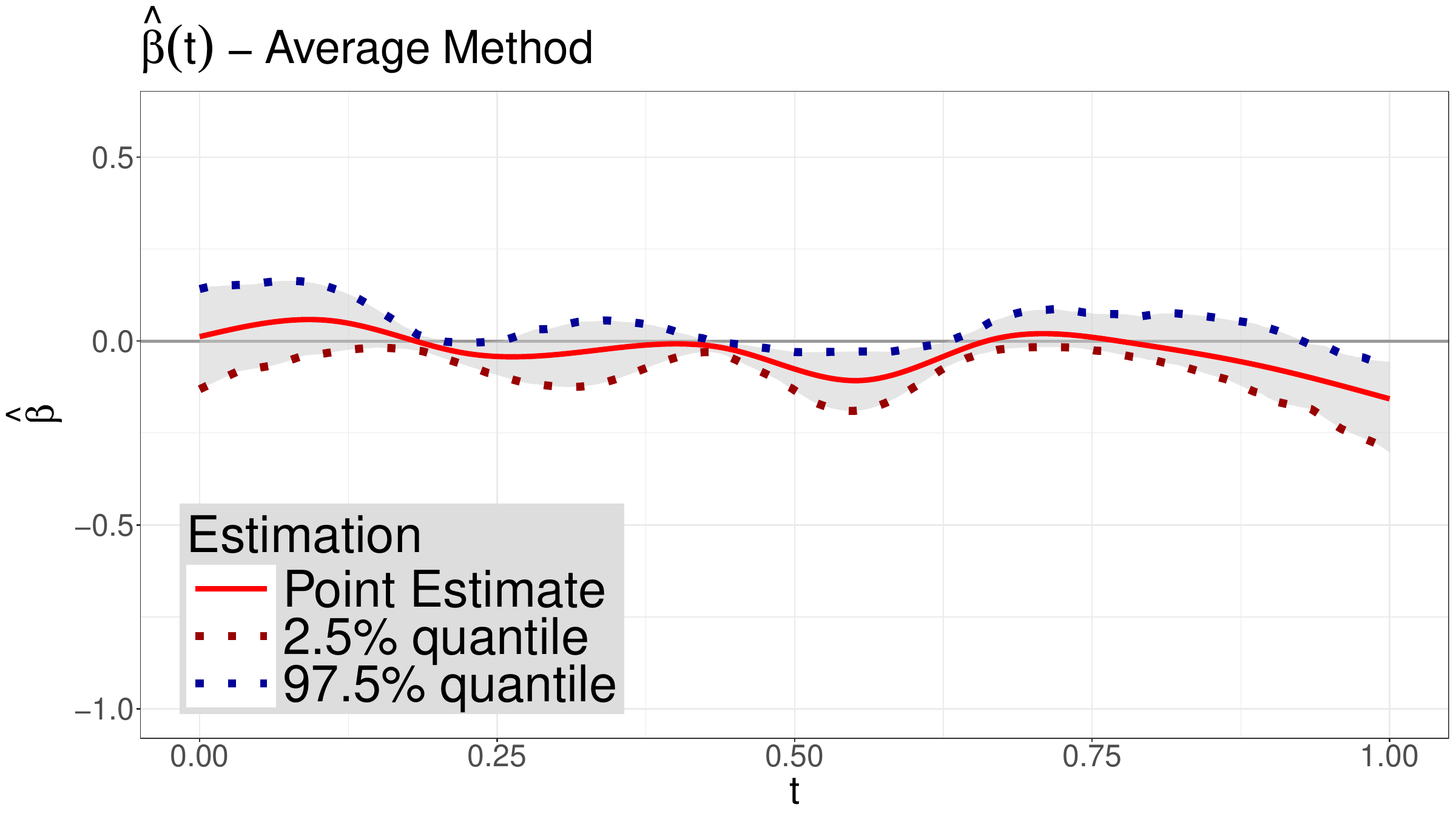}
    \caption{Estimate and confidence interval from average method}
    \label{subfig:ra_avg_bY}
\end{subfigure}

\vspace{1em}

\begin{subfigure}[t]{0.47\textwidth}
    \centering
    \includegraphics[width=\textwidth]{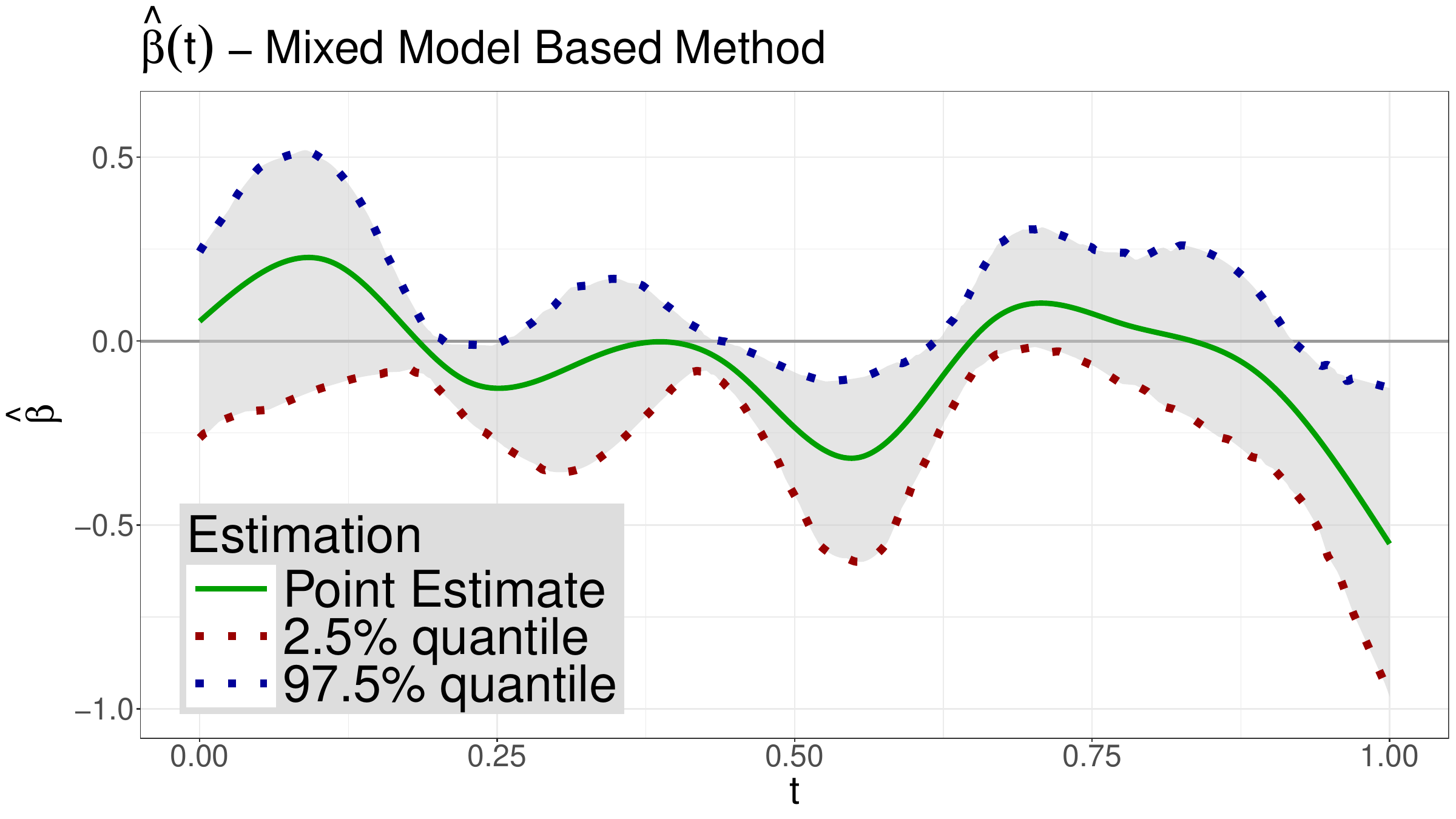}
    \caption{Estimate and confidence interval from mixed model based method}
    \label{subfig:ra_mm_bY}
\end{subfigure}
\hfill
\begin{subfigure}[t]{0.47\textwidth}
    \centering
    \includegraphics[width=\textwidth]{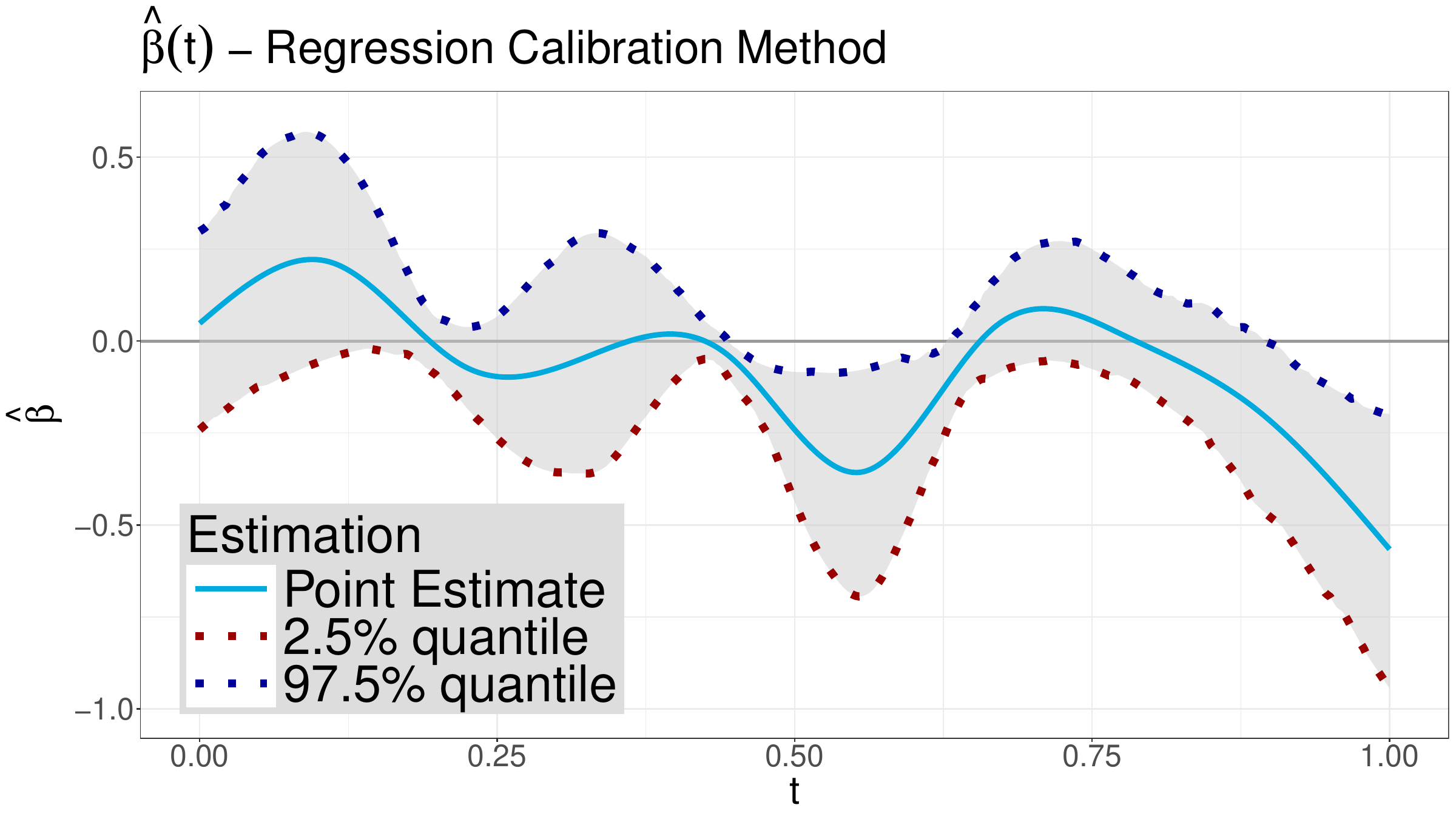}
    \caption{Estimate and confidence interval from regression calibration method}
    \label{subfig:ra_rc_bY}
\end{subfigure}

\vspace{1em}

\begin{subfigure}[t]{0.47\textwidth}
    \centering
    \includegraphics[width=\textwidth]{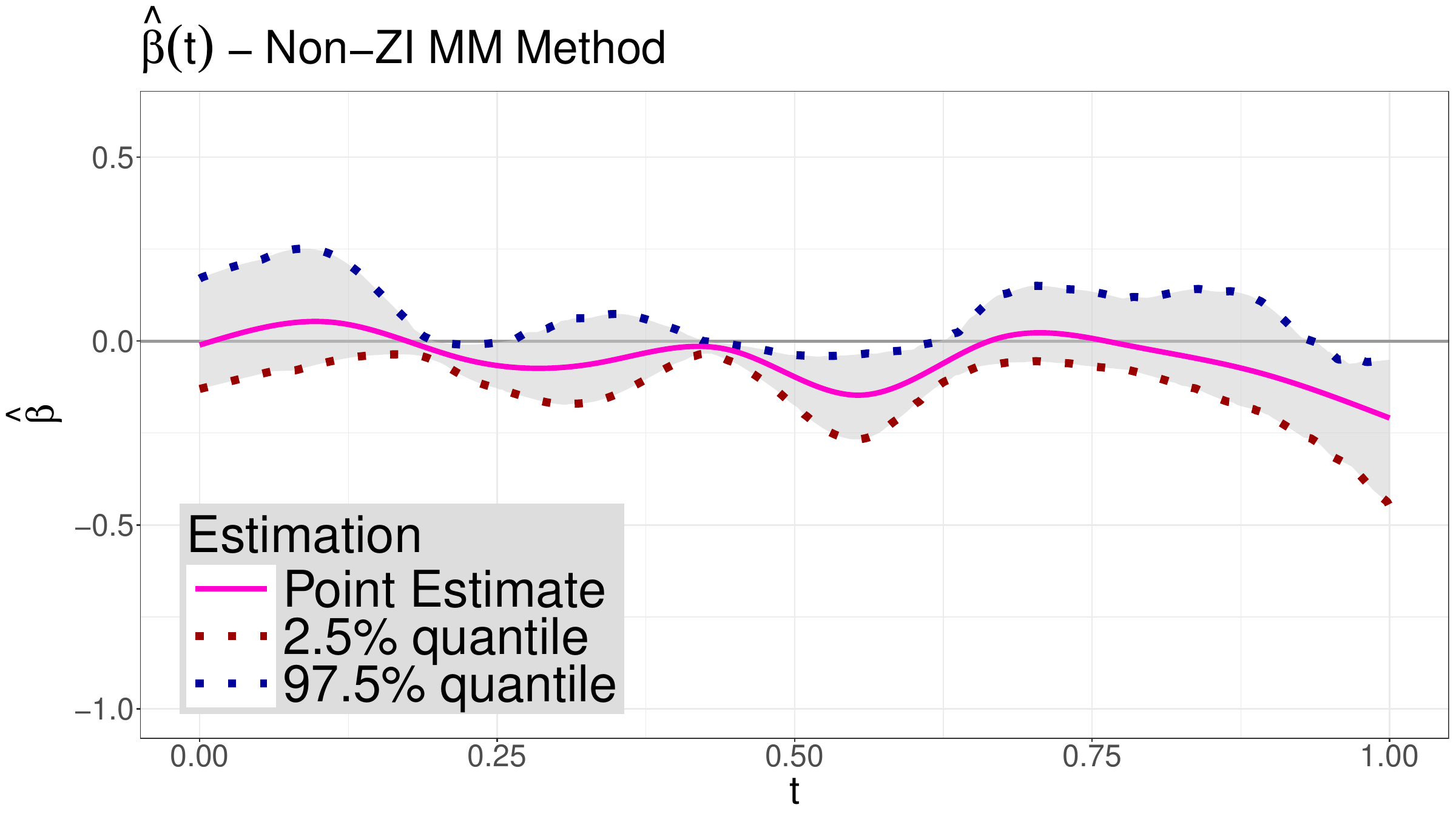}
    \caption{Estimate and confidence interval from no zero inflation mixed model based method}
    \label{subfig:ra_nZ_bY}
\end{subfigure}
\hfill
\begin{subfigure}[t]{0.47\textwidth}
    \centering
    \includegraphics[width=\textwidth]{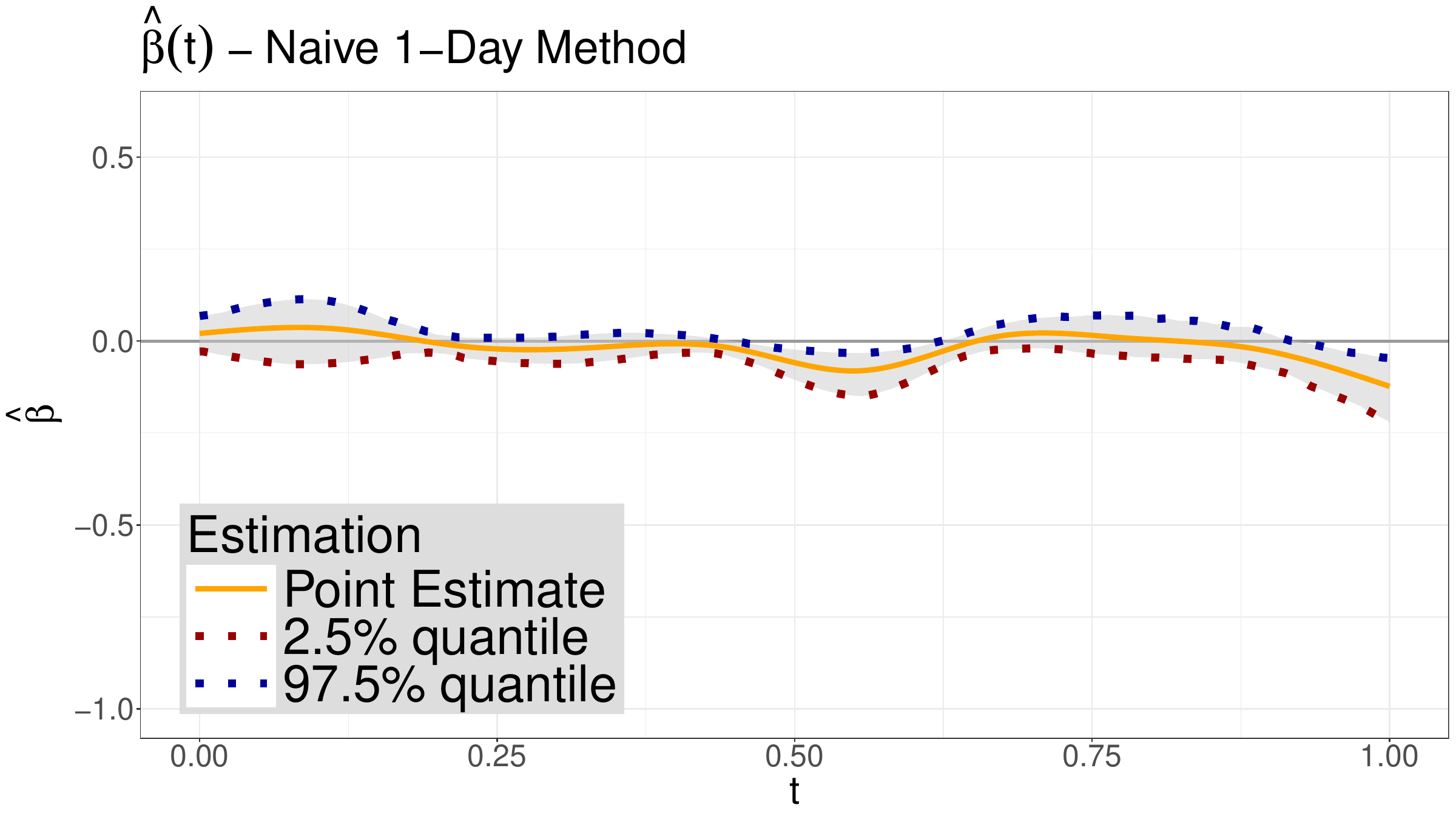}
    \caption{Estimate and confidence interval from one-day method}
    \label{subfig:ra_1d_bY}
\end{subfigure}

\caption{Estimated effect of step counts on \emph{BMI\_reduced} adjusted for age, sex, and ethnicity. Panel (\subref{subfig:ra_pe_bY}) shows the point estimate from different methods. Panels (\subref{subfig:ra_avg_bY})–(\subref{subfig:ra_1d_bY}) show the point estimates and confidence intervals of $\widehat\beta(t)$ from each method. Panels (\subref{subfig:ra_mm_bY}) and (\subref{subfig:ra_rc_bY}) correspond to our proposed bias correction methods.}
\label{fig:ra.2}
\end{figure}

\section{Conclusion and Discussion}\label{sec:C&D}

In this paper, we proposed a novel statistical framework for generalized linear regression with a scalar response and a zero-inflated error-prone functional predictor. The model is applicable to settings where a function-valued predictor is prone to both zero inflation and measurement errors such as step counts monitored by wearable devices. The model accommodates key challenges including excess zeros and measurement error in time-varying predictors through a semi-continuous modeling structure that leverages functional data techniques.

We developed two pointwise bias correction methods for parameter estimation: a mixed effects model-based approach and a regression calibration approach. Their performances were evaluated through extensive simulation studies and compared with several alternative estimation methods that do not address the measurement error or zero inflation issues. Results demonstrate that both proposed approaches substantially reduce estimation bias relative to the alternatives. Between the two, the mixed model-based method achieves lower bias and faster computation, while the regression calibration method exhibits better performance in terms of variance. The overall computational burden is acceptable in practice, and parallelization can be employed to further improve efficiency.
These complementary strengths make the two approaches suitable for different practical scenarios, depending on the desired trade-off between bias, variance, and computational efficiency.

The scalar-on-function regression methods have been implemented in the {\texttt{R}} package {\texttt{MECfda}} \cite{ji2024MECfda}, which also includes support for certain measurement error correction techniques (excluding the one proposed in this paper). The new methods introduced in this paper will be incorporated into future versions of the package.

The proposed estimation procedures are pointwise in nature, offering two main advantages. First, they reduce computational cost relative to methods that jointly model multiple time points. Second, they are more robust to violations of assumptions such as temporal correlation structures. However, a limitation of pointwise methods is their inability to exploit the full joint distribution of the functional predictor, which may affect statistical efficiency. 

We applied our proposed methods to data collected from a cohort of children in Texas, drawn from a school-based childhood obesity intervention study that included health and behavioral measurements. Our analysis uncovered meaningful associations between physical activity (step count) and BMI outcomes, while accounting for measurement error and zero inflation—features commonly present in real-world accelerometer data.

Several avenues exist for extending our framework. In this study, we modeled the marginal distribution of the functional predictor at each time point as a mixture of a Gaussian and a point mass at zero. Future work could explore more flexible distributions beyond the Gaussian to better capture real-world variability. Additionally, developing estimation procedures that account for cross-time dependencies may improve statistical efficiency by leveraging the full joint distribution of the functional predictor. Bayesian approaches may enhance model flexibility, uncertainty quantification, and robustness in sparse or heterogeneous data scenarios.

Overall, our proposed framework provides a statistically rigorous and computationally feasible tool for analyzing complex functional data with excess zeros and measurement error, with broad applicability in public health and behavioral science.

\bibliographystyle{plainnat}
\bibliography{references}

\appendix
\renewcommand{\thesection}{Appendix \Alph{section}}

\section{Notations}\label{sec:notations}

Here we provide the definition of some notations used in Lemma~\ref{th:lemma1} and Theorem~\ref{theorem1}.

\begin{itemize}[itemsep=-3pt,topsep=-8pt]
  \item For an element in $L^2(\Omega)$, $\|\cdot\|_{L^2(\Omega)}$ denotes the $L^2$ norm on $\Omega$, i.e., $\|\beta\|_{L^2(\Omega)} = \left(\int_{\Omega} |\beta(t)|^2\,dt\right)^{1/2}$, abbreviated as \(\lVert\cdot\rVert_{L^{2}}\) when the domain is clear.
  \item For a measureable subset of $\mathbb{R}$, $\|\Omega\|$ denotes the Lebesgue measure of $\Omega$, i.e., $\|\Omega\| = \int_{\Omega} 1\,dt$.
  \item For a finite set \(\mathcal{T}\), \(\lvert\mathcal{T}\rvert\) denotes its cardinality; e.g., if \(\mathcal{T}=\{t_{1},\dots,t_{m}\}\) then \(\lvert\mathcal{T}\rvert=m\).
  \item Let $\beta_K(t) = \sum_{k=1}^{K} b_k \rho_k(t)$ denote the approximation of $\beta$ onto the first $K$ basis functions. 
\end{itemize}

\section{Proofs}\label{sec:proofs}

\subsection{Auxiliary Lemmas}

Before proving the main results we collect two auxiliary facts.

\begin{lemma}\label{th:lemma2}
Let \(X\) be a random variable and, conditional on \(X\), let
\(W_1,\dots,W_J\) be i.i.d.\ copies of a proxy variable \(W\).
Suppose that
\begin{enumerate}[label=\textbf{(E\arabic*)},align=left,itemsep=-3pt,topsep=-8pt]
  \item \(W_j \mid X=x \stackrel{\mathrm{i.i.d.}}{\sim}
        F_{W\mid X}(\,\cdot\mid x)\);
  \item for any $x_1 \neq x_2$, $F_{W\mid X = x_1} \neq F_{W\mid X = x_2}$ \textup{(identifiability)};
  \item for any $j$, $\Var(W_j|X) < \infty$.
\end{enumerate}
Then
\[
  \mathbb{E}\!\bigl[X \,\bigl|\,\widehat{W}\bigr]
  \xrightarrow{\text{a.s.}} X
  \qquad (J\to\infty),
  \quad\text{where}\;
  \widehat{W}:=(W_1,\dots,W_J)^{\top}.
\]
\end{lemma}

This Lemma~\ref{th:lemma2} is also a well known result. 
It shows that the conditional expectation of $X$ given a large number of i.i.d. observations of $W$ (if conditional on $X$) converges to $X$ almost surely, under certain conditions.
Its proof is omitted. 

\begin{lemma}
Let
\[
  h(Z)\;=\;(Z^{\top}Z)^{-1}Zy,
  \qquad Z\in\mathbb{R}^{n\times p},\;y\in\mathbb{R}^{n}.
\]
Suppose $\{Z_1^{(m)}\}_{m\ge 1}$ and $Z_2$ are design matrices with identical
dimensions such that
\[
  Z_1^{(m)}
  \xrightarrow{P}
  Z_2
  \quad
  (\text{element-wise, equivalently } \|Z_1^{(m)}-Z_2\|_F\xrightarrow{P}0).
\]
If $Z_2^{\top}Z_2$ is invertible (i.e.\ $\lambda_{\min}(Z_2^{\top}Z_2)>0$),
then
\[
  h\!\bigl(Z_1^{(m)}\bigr)\xrightarrow{P}h(Z_2).
\]
\label{th:lemma3}
\end{lemma}

\begin{proof}[Proof of Lemma~\ref{th:lemma3}]
For any matrices $A,B$ of the same size,
\[
  \|A^{\top}A-B^{\top}B\|_F
  \le \|A^{\top}(A-B)\|_F + \|(A-B)^{\top}B\|_F
  \le (\|A\|_F+\|B\|_F)\,\|A-B\|_F.
\]
Hence
$\|Z_1^{(m)\top}Z_1^{(m)}-Z_2^{\top}Z_2\|_F\xrightarrow{P}0$.
\medskip
Write $A_m=Z_1^{(m)\top}Z_1^{(m)}$ and $A=Z_2^{\top}Z_2$.
Weyl’s inequality yields
\[
  \bigl|\lambda_{\min}(A_m)-\lambda_{\min}(A)\bigr|
  \le \|A_m-A\|_2
  \le \|A_m-A\|_F\xrightarrow{P}0.
\]
Because $\lambda_{\min}(A)>0$, it follows that
$\lambda_{\min}(A_m)\xrightarrow{P}\lambda_{\min}(A)>0$,
so $A_m$ is invertible with probability tending to one and
$A_m^{-1}\xrightarrow{P}A^{-1}$.
\medskip
The matrix inversion map is continuous on the set of invertible matrices, and
matrix multiplication is jointly continuous. Consequently,
\[
  (Z_1^{(m)\top}Z_1^{(m)})^{-1}
  \xrightarrow{P}
  (Z_2^{\top}Z_2)^{-1},
  \qquad
  Z_1^{(m)}
  \xrightarrow{P}
  Z_2.
\]
Slutsky's theorem therefore gives
\[
  (Z_1^{(m)\top}Z_1^{(m)})^{-1}Z_1^{(m)}
  \xrightarrow{P}
  (Z_2^{\top}Z_2)^{-1}Z_2.
\]
\medskip
If $y$ is deterministic (or independent of the $Z_1^{(m)}$ sequence with finite
second moment), another application of Slutsky yields
\[
  h\!\bigl(Z_1^{(m)}\bigr)
  =(Z_1^{(m)\top}Z_1^{(m)})^{-1}Z_1^{(m)}y
  \xrightarrow{P}
  (Z_2^{\top}Z_2)^{-1}Z_2y
  =h(Z_2).
\]
\end{proof}

\subsection{Proof of Theorem~\ref{theorem1}}

\begin{proof}[Proof of Theorem~\ref{theorem1}]

We decompose the estimation error using the triangle inequality,
\begin{equation}
\|\widehat{\beta}_{K_n} - \beta\|_{L^2(\Omega)} \leq \|\widehat{\beta}_{K_n} - \widetilde{\beta}_{K_n}\|_{L^2(\Omega)} + \|\widetilde{\beta}_{K_n} - \beta\|_{L^2(\Omega)}.
\end{equation}
Based on the Lemma~\ref{th:lemma1}, we already have $\|\widetilde{\beta}_{K_n} - \beta\|_{L^2(\Omega)} \xrightarrow{P} 0$ as $n \to \infty$. 
Therefore, we only need to show 
$\|\widehat{\beta}_{K_n} - \widetilde{\beta}_{K_n}\|_{L^2(\Omega)} \xrightarrow{P} 0$, 
equivalently, $\|\widehat{\bm\theta}_{K_n} - \widetilde{\bm\theta}_{K_n}\| \xrightarrow{P} 0$.


Let $\bm y := (Y_1, \ldots, Y_n)^\top$ be the vector of response
And the least-squares estimator $\widehat{\bm{\theta}}_{K_n}$ and $\widetilde{\bm{\theta}}_{K_n}$ 
can be expressed in closed form of 
$\widehat{\bm\theta}_{K_n} = ({\widehat{\bm X}}_{{K_n}}^\top {\widehat{\bm X}}_{{K_n}})^{-1} {\widehat{\bm X}}_{{K_n}}^\top \bm y$ and
$\widetilde{\bm\theta}_{K_n} = ({{\bm X}}_{{K_n}}^\top {{\bm X}}_{{K_n}})^{-1} {{\bm X}}_{{K_n}}^\top \bm y$. 

Because $\bm{X}_{K_n}$ and $\widehat{\bm{X}}_{K_n}$ have full column rank, 
$\bm{X}_{K_n}^\top \bm{X}_{K_n}$ and $\widehat{\bm{X}}_{K_n}^\top \widehat{\bm{X}}_{K_n}$ are both invertable. 
Based on Lemma~\ref{th:lemma3}, to show $\|\widehat{\beta}_{K_n} - \widetilde{\beta}_{K_n}\|_{L^2(\Omega)} \xrightarrow{P} 0$, 
we only need to show that $|\hat{x}_{ik} - x_{ik}|\xrightarrow{P} 0$,
equivalently $\frac{\|\Omega\|}{|\mathcal{T}_m|}\sum_{t\in\mathcal{T}_m}\widehat{X}_i(t) \rho_k(t) \to \int_{\Omega} X_i(t) \rho_k(t)\,dt$.

We decompose $\left|\hat{x}_{ik} - x_{ik}\right|$ as follows:
\begin{align*}
\left|\hat{x}_{ik} - x_{ik}\right|
\leq
\left|\hat{x}_{ik} - A_m\right| + 
\left|A_m - x_{ik}\right|
\end{align*}
where $A_m = \frac{\|\Omega\|}{|\mathcal{T}_m|}\sum_{t\in\mathcal{T}_m}X_i(t) \rho_k(t)$.

Since $$\hat{x}_{ik} - A_m = \frac{\|\Omega\|}{|\mathcal{T}_m|}\sum_{t\in\mathcal{T}_m}\{X_i(t) - \widehat{X}_{i}(t)\} \rho_k(t),$$
and we have $\widehat{X}_{i}(t) \xrightarrow{a.s.} {X}_{i}(t)$ as $J_i\to\infty$ for any $i$ and $t$ (by Lemma~\ref{th:lemma2}),
the first term $\left|\hat{x}_{ik} - A_m\right|$ can be shown to converge to $0$ as $J_i \to \infty$ for any $i$ and $k$.

The second term $\left|A_m - x_{ik}\right|$ can be proven to converge to $0$ as $\Delta_m \to 0$ by Riemann-Lebesgue Theorem.

Combining the pieces, we have
$$\|\widehat{\beta}_{K_n} - \beta\|_{L^2(\Omega)} \xrightarrow{P} 0$$
as $n \to \infty$, $\inf_i{J_i} \to \infty$, and $m \to \infty$.

\end{proof}

(Undefined notation is summarized in Section~\ref{sec:notations}.)

\section{Simulation Results}\label{sec:more simulation results}

In the main text (Section~\ref{simulation result}), we presented selected results from the simulation study under a single set of simulation settings. In this appendix, we provide additional tables that summarize the results under a wider range of settings. 

Each table isolates the effect of one simulation parameter by varying it while keeping all other settings fixed. 
Within each table, subtable (a) reports the squared bias, and subtable (b) reports the variance of the estimator $\hat\beta(t)$. For each scenario, we report the bias and variance of the estimator $\hat\beta(t)$ to illustrate the impact of the specific setting on estimation performance.

\begin{table}[htbp]
\captionsetup{justification=raggedright}
\caption{Effect of sample size. {\normalsize{Common settings:} 
{conditional distribution of $Y|X(t),Z$:} Gaussian distribution (with identity link); 
{estimation method of $\theta(t)$:} point-wise; 
{proportion of zero values:} $\mathbb{E}\{\Pr(W_{ij}(t)=0)\}=0.335$; 
{Covariance function of $U(t)$:} squared-exponential with $\rho_u=0.2$; 
$\sigma_u = 1$; 
$q_g=0$. 
{Benchmark} = benchmark method; 
{Average} = naïve average method; 
{MM} = mixed model based method proposed in Section~\ref{subsec:predict X}; 
{RC} = regression calibration method proposed in Section~\ref{subsec:predict X}; 
{Non-ZI MM} = non-zero-inflated mixed model based method; 
{1 day} = naive one-day method.}}
\centering

\begin{subtable}{\textwidth}
\centering
\caption{Squared Bias}

\centering
\footnotesize
\begin{tabular}{r|r|r|r|r|r|r|l}
\hline
N & Benchmark & MM & RC & Average & Non-ZI MM & 1 day & basis function\\
\hline
50 & 0.0001213 & 0.003945 & 0.005530 & 0.1497 & 0.2995 & 0.3897 & Bspline\\
100 & 0.0000903 & 0.003446 & 0.005038 & 0.1445 & 0.1689 & 0.3822 & Bspline\\
200 & 0.0001204 & 0.003259 & 0.005246 & 0.1489 & 0.1318 & 0.3840 & Bspline\\
500 & 0.0001262 & 0.002772 & 0.004994 & 0.1468 & 0.1190 & 0.3842 & Bspline\\
1000 & 0.0001110 & 0.002619 & 0.004417 & 0.1494 & 0.1225 & 0.3846 & Bspline\\
2000 & 0.0001118 & 0.002687 & 0.004671 & 0.1491 & 0.1155 & 0.3844 & Bspline\\
5000 & 0.0001117 & 0.002562 & 0.004822 & 0.1485 & 0.1143 & 0.3845 & Bspline\\
50 & 0.1074799 & 0.018802 & 0.011150 & 0.1407 & 0.3793 & 0.3897 & Fourier\\
100 & 0.0128407 & 0.010145 & 0.010521 & 0.1378 & 0.1869 & 0.3822 & Fourier\\
200 & 0.0192103 & 0.005523 & 0.006574 & 0.1417 & 0.1385 & 0.3840 & Fourier\\
500 & 0.0070227 & 0.006152 & 0.008074 & 0.1397 & 0.1215 & 0.3842 & Fourier\\
1000 & 0.0001395 & 0.003398 & 0.005739 & 0.1422 & 0.1261 & 0.3846 & Fourier\\
2000 & 0.0004761 & 0.004042 & 0.006083 & 0.1420 & 0.1187 & 0.3844 & Fourier\\
5000 & 0.0005231 & 0.003902 & 0.006130 & 0.1414 & 0.1163 & 0.3845 & Fourier\\
\hline
\end{tabular}

\end{subtable}

\vspace{1em}

\begin{subtable}{\textwidth}
\centering
\caption{Variance}

\centering
\footnotesize
\begin{tabular}{r|r|r|r|r|r|r|l}
\hline
N & Benchmark & MM & RC & Average & Non-ZI MM & 1 day & basis function\\
\hline
50 & 0.0023843 & 0.135181 & 0.120203 & 0.0586 & 2.6930 & 0.0190 & Bspline\\
100 & 0.0010200 & 0.061922 & 0.049800 & 0.0269 & 0.3297 & 0.0085 & Bspline\\
200 & 0.0005190 & 0.034514 & 0.026972 & 0.0155 & 0.1311 & 0.0047 & Bspline\\
500 & 0.0002088 & 0.014465 & 0.010488 & 0.0069 & 0.0479 & 0.0018 & Bspline\\
1000 & 0.0000995 & 0.008700 & 0.005136 & 0.0046 & 0.0246 & 0.0010 & Bspline\\
2000 & 0.0000490 & 0.005206 & 0.003023 & 0.0033 & 0.0141 & 0.0007 & Bspline\\
5000 & 0.0000199 & 0.004210 & 0.001821 & 0.0026 & 0.0085 & 0.0004 & Bspline\\
50 & 21.7150811 & 1.887667 & 1.827130 & 0.3901 & 7.8791 & 0.0190 & Fourier\\
100 & 8.9110659 & 0.798691 & 0.758550 & 0.1662 & 1.5090 & 0.0085 & Fourier\\
200 & 3.9275922 & 0.419874 & 0.352485 & 0.0775 & 0.5864 & 0.0047 & Fourier\\
500 & 1.5453553 & 0.171680 & 0.135424 & 0.0322 & 0.2068 & 0.0018 & Fourier\\
1000 & 0.7198404 & 0.099063 & 0.071089 & 0.0169 & 0.0987 & 0.0010 & Fourier\\
2000 & 0.3458503 & 0.056823 & 0.037739 & 0.0105 & 0.0552 & 0.0007 & Fourier\\
5000 & 0.1465289 & 0.032680 & 0.018335 & 0.0065 & 0.0283 & 0.0004 & Fourier\\
\hline
\end{tabular}

\end{subtable}

\end{table}

\begin{table}[htbp]
\captionsetup{justification=raggedright}
\caption{Effect of deviation of U. {\normalsize{Common settings:} 
{conditional distribution of $Y|X(t),Z$:} Gaussian distribution (with identity link); 
{estimation method of $\theta(t)$:} point-wise; 
{sample size:} $n = 100$; 
{proportion of zero values:} $\mathbb{E}\{\Pr(W_{ij}(t)=0)\}=0.335$; 
{Covariance function of $U(t)$:} squared-exponential with $\rho_u=0.2$; 
$q_g=0$. 
{Benchmark} = benchmark method; 
{Average} = naïve average method; 
{MM} = mixed model based method proposed in Section~\ref{subsec:predict X}; 
{RC} = regression calibration method proposed in Section~\ref{subsec:predict X}; 
{Non-ZI MM} = non-zero-inflated mixed model based method; 
{1 day} = naive one-day method.}}
\centering

\begin{subtable}{\textwidth}
\centering
\caption{Squared Bias}

\centering
\footnotesize
\begin{tabular}{r|r|r|r|r|r|r|l}
\hline
$\sigma_u$ & Benchmark & MM & RC & Average & Non-ZI MM & 1 day & basis function\\
\hline
1 & 0.0000903 & 0.003446 & 0.005038 & 0.1445 & 0.1689 & 0.3822 & Bspline\\
2 & 0.0000903 & 0.014637 & 0.023435 & 0.1913 & 0.2151 & 0.4180 & Bspline\\
3 & 0.0000903 & 0.033851 & 0.076080 & 0.2534 & 1117.2726 & 0.4461 & Bspline\\
1 & 0.0128407 & 0.010145 & 0.010521 & 0.1378 & 0.1869 & 0.3822 & Fourier\\
2 & 0.0128407 & 0.025381 & 0.028824 & 0.1866 & 0.2678 & 0.4180 & Fourier\\
3 & 0.0128407 & 0.081917 & 0.077106 & 0.2517 & 1.8979 & 0.4461 & Fourier\\
\hline
\end{tabular}

\end{subtable}

\vspace{1em}

\begin{subtable}{\textwidth}
\centering
\caption{Variance}

\centering
\footnotesize
\begin{tabular}{r|r|r|r|r|r|r|l}
\hline
$\sigma_u$ & Benchmark & MM & RC & Average & Non-ZI MM & 1 day & basis function\\
\hline
1 & 0.001020 & 0.061922 & 0.049800 & 0.0269 & 0.3297 & 0.0085 & Bspline\\
2 & 0.001020 & 0.187533 & 0.100191 & 0.0262 & 1.0378 & 0.0068 & Bspline\\
3 & 0.001020 & 0.571430 & 0.114888 & 0.0236 & 544845.5255 & 0.0052 & Bspline\\
1 & 8.911066 & 0.798691 & 0.758550 & 0.1662 & 1.5090 & 0.0085 & Fourier\\
2 & 8.911066 & 1.574898 & 0.926606 & 0.1670 & 3.9490 & 0.0068 & Fourier\\
3 & 8.911066 & 8.294103 & 0.842241 & 0.1484 & 168.2378 & 0.0052 & Fourier\\
\hline
\end{tabular}

\end{subtable}

\end{table}

\begin{table}[htbp]
\captionsetup{justification=raggedright}
\caption{Effect of correlation of u. {\normalsize{Common settings:} 
{conditional distribution of $Y|X(t),Z$:} Gaussian distribution (with identity link); 
{estimation method of $\theta(t)$:} point-wise; 
{sample size:} $n = 100$; 
{proportion of zero values:} $\mathbb{E}\{\Pr(W_{ij}(t)=0)\}=0.335$; 
$\sigma_u = 1$; 
$q_g=0$. 
{Benchmark} = benchmark method; 
{Average} = naïve average method; 
{MM} = mixed model based method proposed in Section~\ref{subsec:predict X}; 
{RC} = regression calibration method proposed in Section~\ref{subsec:predict X}; 
{Non-ZI MM} = non-zero-inflated mixed model based method; 
{1 day} = naive one-day method.}}
\centering

\begin{subtable}{\textwidth}
\centering
\caption{Squared Bias}

\centering
\footnotesize
\begin{tabular}{l|r|r|r|r|r|r|r|l}
\hline
correlation function $U$ & $\rho_u$ & Benchmark & MM & RC & Average & Non-ZI MM & 1 day & basis function\\
\hline
AR1 & 0.20 & 0.0000903 & 0.011834 & 0.008028 & 0.1305 & 0.2135 & 0.3678 & Bspline\\
AR1 & 0.40 & 0.0000903 & 0.009961 & 0.006499 & 0.1323 & 0.2004 & 0.3697 & Bspline\\
AR1 & 0.60 & 0.0000903 & 0.007081 & 0.004566 & 0.1361 & 0.1959 & 0.3734 & Bspline\\
Compound symmetry & 0.20 & 0.0000903 & 0.013329 & 0.009169 & 0.1295 & 0.2217 & 0.3664 & Bspline\\
Compound symmetry & 0.40 & 0.0000903 & 0.013075 & 0.008865 & 0.1302 & 0.2255 & 0.3662 & Bspline\\
Compound symmetry & 0.60 & 0.0000903 & 0.012792 & 0.008586 & 0.1304 & 0.2263 & 0.3670 & Bspline\\
Squared Exponential & 0.10 & 0.0000903 & 0.005598 & 0.004456 & 0.1427 & 0.1840 & 0.3790 & Bspline\\
Squared Exponential & 0.15 & 0.0000903 & 0.003775 & 0.003768 & 0.1448 & 0.1764 & 0.3812 & Bspline\\
Squared Exponential & 0.20 & 0.0000903 & 0.003446 & 0.005038 & 0.1445 & 0.1689 & 0.3822 & Bspline\\
AR1 & 0.20 & 0.0128407 & 0.011961 & 0.008230 & 0.1233 & 0.2384 & 0.3678 & Fourier\\
AR1 & 0.40 & 0.0128407 & 0.009733 & 0.006574 & 0.1241 & 0.2179 & 0.3697 & Fourier\\
AR1 & 0.60 & 0.0128407 & 0.006987 & 0.005004 & 0.1282 & 0.2117 & 0.3734 & Fourier\\
Compound symmetry & 0.20 & 0.0128407 & 0.013377 & 0.009305 & 0.1214 & 0.2515 & 0.3664 & Fourier\\
Compound symmetry & 0.40 & 0.0128407 & 0.013624 & 0.009515 & 0.1227 & 0.2598 & 0.3662 & Fourier\\
Compound symmetry & 0.60 & 0.0128407 & 0.013076 & 0.009213 & 0.1219 & 0.2699 & 0.3670 & Fourier\\
Squared Exponential & 0.10 & 0.0128407 & 0.009087 & 0.007798 & 0.1367 & 0.2050 & 0.3790 & Fourier\\
Squared Exponential & 0.15 & 0.0128407 & 0.008810 & 0.009966 & 0.1377 & 0.1912 & 0.3812 & Fourier\\
Squared Exponential & 0.20 & 0.0128407 & 0.010145 & 0.010521 & 0.1378 & 0.1869 & 0.3822 & Fourier\\
\hline
\end{tabular}

\end{subtable}

\vspace{1em}

\begin{subtable}{\textwidth}
\centering
\caption{Variance}

\centering
\footnotesize
\begin{tabular}{l|r|r|r|r|r|r|r|l}
\hline
correlation function $U$ & $\rho_u$ & Benchmark & MM & RC & Average & Non-ZI MM & 1 day & basis function\\
\hline
AR1 & 0.20 & 0.001020 & 0.021621 & 0.013414 & 0.0252 & 0.3825 & 0.0085 & Bspline\\
AR1 & 0.40 & 0.001020 & 0.024735 & 0.016722 & 0.0252 & 0.3101 & 0.0083 & Bspline\\
AR1 & 0.60 & 0.001020 & 0.032121 & 0.023158 & 0.0252 & 0.2978 & 0.0081 & Bspline\\
Compound symmetry & 0.20 & 0.001020 & 0.022135 & 0.012388 & 0.0259 & 0.3221 & 0.0086 & Bspline\\
Compound symmetry & 0.40 & 0.001020 & 0.024029 & 0.014656 & 0.0258 & 0.3433 & 0.0085 & Bspline\\
Compound symmetry & 0.60 & 0.001020 & 0.027675 & 0.017451 & 0.0261 & 0.3292 & 0.0088 & Bspline\\
Squared Exponential & 0.10 & 0.001020 & 0.036727 & 0.030305 & 0.0255 & 0.3730 & 0.0080 & Bspline\\
Squared Exponential & 0.15 & 0.001020 & 0.052211 & 0.043760 & 0.0262 & 0.3314 & 0.0084 & Bspline\\
Squared Exponential & 0.20 & 0.001020 & 0.061922 & 0.049800 & 0.0269 & 0.3297 & 0.0085 & Bspline\\
AR1 & 0.20 & 8.911066 & 0.095675 & 0.066729 & 0.1245 & 1.2455 & 0.0085 & Fourier\\
AR1 & 0.40 & 8.911066 & 0.122368 & 0.092885 & 0.1246 & 1.2295 & 0.0083 & Fourier\\
AR1 & 0.60 & 8.911066 & 0.174827 & 0.141824 & 0.1297 & 1.2172 & 0.0081 & Fourier\\
Compound symmetry & 0.20 & 8.911066 & 0.114145 & 0.070892 & 0.1298 & 1.5421 & 0.0086 & Fourier\\
Compound symmetry & 0.40 & 8.911066 & 0.134355 & 0.090184 & 0.1370 & 1.6842 & 0.0085 & Fourier\\
Compound symmetry & 0.60 & 8.911066 & 0.172675 & 0.126338 & 0.1428 & 2.5201 & 0.0088 & Fourier\\
Squared Exponential & 0.10 & 8.911066 & 0.364757 & 0.330020 & 0.1482 & 1.3609 & 0.0080 & Fourier\\
Squared Exponential & 0.15 & 8.911066 & 0.694774 & 0.643891 & 0.1638 & 1.4703 & 0.0084 & Fourier\\
Squared Exponential & 0.20 & 8.911066 & 0.798691 & 0.758550 & 0.1662 & 1.5090 & 0.0085 & Fourier\\
\hline
\end{tabular}

\end{subtable}

\end{table}

\begin{table}[htbp]
\captionsetup{justification=raggedright}
\caption{Effect of proportion of zero values in W. {\normalsize{Common settings:} 
{conditional distribution of $Y|X(t),Z$:} Gaussian distribution (with identity link); 
{estimation method of $\theta(t)$:} point-wise; 
{sample size:} $n = 100$; 
{Covariance function of $U(t)$:} squared-exponential with $\rho_u=0.2$; 
$\sigma_u = 1$; 
$q_g=0$. 
{Benchmark} = benchmark method; 
{Average} = naïve average method; 
{MM} = mixed model based method proposed in Section~\ref{subsec:predict X}; 
{RC} = regression calibration method proposed in Section~\ref{subsec:predict X}; 
{Non-ZI MM} = non-zero-inflated mixed model based method; 
{1 day} = naive one-day method.}}
\centering

\begin{subtable}{\textwidth}
\centering
\caption{Squared Bias}

\centering
\footnotesize
\begin{tabular}{r|r|r|r|r|r|r|l}
\hline
E(Pr{W=0}) & Benchmark & MM & RC & Average & Non-ZI MM & 1 day & basis function\\
\hline
0.255 & 0.0000903 & 0.006718 & 0.007163 & 0.1074 & 0.0986 & 0.3451 & Bspline\\
0.294 & 0.0000903 & 0.004593 & 0.005817 & 0.1258 & 0.1323 & 0.3641 & Bspline\\
0.335 & 0.0000903 & 0.003446 & 0.005038 & 0.1445 & 0.1689 & 0.3822 & Bspline\\
0.403 & 0.0000903 & 0.002542 & 0.003985 & 0.1770 & 0.2860 & 0.4056 & Bspline\\
0.255 & 0.0128407 & 0.012092 & 0.011912 & 0.1034 & 0.1074 & 0.3451 & Fourier\\
0.294 & 0.0128407 & 0.011797 & 0.011083 & 0.1192 & 0.1371 & 0.3641 & Fourier\\
0.335 & 0.0128407 & 0.010145 & 0.010521 & 0.1378 & 0.1869 & 0.3822 & Fourier\\
0.403 & 0.0128407 & 0.005073 & 0.008999 & 0.1684 & 0.3403 & 0.4056 & Fourier\\
\hline
\end{tabular}

\end{subtable}

\vspace{1em}

\begin{subtable}{\textwidth}
\centering
\caption{Variance}

\centering
\footnotesize
\begin{tabular}{r|r|r|r|r|r|r|l}
\hline
E(Pr{W=0}) & Benchmark & MM & RC & Average & Non-ZI MM & 1 day & basis function\\
\hline
0.255 & 0.001020 & 0.056255 & 0.050114 & 0.0313 & 0.1827 & 0.0120 & Bspline\\
0.294 & 0.001020 & 0.059960 & 0.050266 & 0.0293 & 0.2373 & 0.0099 & Bspline\\
0.335 & 0.001020 & 0.061922 & 0.049800 & 0.0269 & 0.3297 & 0.0085 & Bspline\\
0.403 & 0.001020 & 0.075954 & 0.051316 & 0.0231 & 0.8235 & 0.0065 & Bspline\\
0.255 & 8.911066 & 1.057866 & 0.993120 & 0.1829 & 0.9043 & 0.0120 & Fourier\\
0.294 & 8.911066 & 0.941326 & 0.878270 & 0.1721 & 1.1069 & 0.0099 & Fourier\\
0.335 & 8.911066 & 0.798691 & 0.758550 & 0.1662 & 1.5090 & 0.0085 & Fourier\\
0.403 & 8.911066 & 0.641815 & 0.604558 & 0.1429 & 3.2846 & 0.0065 & Fourier\\
\hline
\end{tabular}

\end{subtable}

\end{table}

\begin{table}[htbp]
\captionsetup{justification=raggedright}
\caption{Effect of $q_g$ for $G(t)$ in Eq~\eqref{eq:sim G}. {\normalsize{Common settings:} 
{conditional distribution of $Y|X(t),Z$:} Gaussian distribution (with identity link); 
{estimation method of $\theta(t)$:} point-wise; 
{sample size:} $n = 100$; 
{proportion of zero values:} $\mathbb{E}\{\Pr(W_{ij}(t)=0)\}=0.335$; 
{Covariance function of $U(t)$:} squared-exponential with $\rho_u=0.2$; 
$\sigma_u = 1$; 
{Benchmark} = benchmark method; 
{Average} = naïve average method; 
{MM} = mixed model based method proposed in Section~\ref{subsec:predict X}; 
{RC} = regression calibration method proposed in Section~\ref{subsec:predict X}; 
{Non-ZI MM} = non-zero-inflated mixed model based method; 
{1 day} = naive one-day method.}}
\centering

\begin{subtable}{\textwidth}
\centering
\caption{Squared Bias}

\centering
\footnotesize
\begin{tabular}{r|r|r|r|r|r|r|l}
\hline
$q_g$ & Benchmark & MM & RC & Average & Non-ZI MM & 1 day & basis function\\
\hline
0.0 & 0.0000903 & 0.003446 & 0.005038 & 0.1445 & 0.1689 & 0.3822 & Bspline\\
0.2 & 0.0000903 & 0.002033 & 0.004820 & 0.1618 & 0.0827 & 0.3818 & Bspline\\
0.4 & 0.0000903 & 0.004155 & 0.004004 & 0.2052 & 0.0951 & 0.3820 & Bspline\\
0.0 & 0.0128407 & 0.010145 & 0.010521 & 0.1378 & 0.1869 & 0.3822 & Fourier\\
0.2 & 0.0128407 & 0.012550 & 0.016378 & 0.1536 & 0.0764 & 0.3818 & Fourier\\
0.4 & 0.0128407 & 0.003063 & 0.008264 & 0.1970 & 0.0844 & 0.3820 & Fourier\\
\hline
\end{tabular}

\end{subtable}

\vspace{1em}

\begin{subtable}{\textwidth}
\centering
\caption{Variance}

\centering
\footnotesize
\begin{tabular}{r|r|r|r|r|r|r|l}
\hline
$q_g$ & Benchmark & MM & RC & Average & Non-ZI MM & 1 day & basis function\\
\hline
0.0 & 0.001020 & 0.061922 & 0.049800 & 0.0269 & 0.3297 & 0.0085 & Bspline\\
0.2 & 0.001020 & 0.071695 & 0.049681 & 0.0256 & 0.1898 & 0.0082 & Bspline\\
0.4 & 0.001020 & 0.095285 & 0.049205 & 0.0233 & 0.0789 & 0.0081 & Bspline\\
0.0 & 8.911066 & 0.798691 & 0.758550 & 0.1662 & 1.5090 & 0.0085 & Fourier\\
0.2 & 8.911066 & 0.746826 & 0.713728 & 0.1547 & 0.9039 & 0.0082 & Fourier\\
0.4 & 8.911066 & 0.713032 & 0.664500 & 0.1687 & 0.5337 & 0.0081 & Fourier\\
\hline
\end{tabular}

\end{subtable}

\end{table}

\begin{table}[htbp]
\captionsetup{justification=raggedright}
\caption{Effect of sample size. {\normalsize{Common settings:} 
{conditional distribution of $Y|X(t),Z$:} Gaussian distribution (with identity link); 
smoothed point-wise; 
{proportion of zero values:} $\mathbb{E}\{\Pr(W_{ij}(t)=0)\}=0.335$; 
{Covariance function of $U(t)$:} squared-exponential with $\rho_u=0.2$; 
$\sigma_u = 1$; 
$q_g=0$. 
{Benchmark} = benchmark method; 
{Average} = naïve average method; 
{MM} = mixed model based method proposed in Section~\ref{subsec:predict X}; 
{RC} = regression calibration method proposed in Section~\ref{subsec:predict X}; 
{Non-ZI MM} = non-zero-inflated mixed model based method; 
{1 day} = naive one-day method.}}
\centering

\begin{subtable}{\textwidth}
\centering
\caption{Squared Bias}

\centering
\footnotesize
\begin{tabular}{r|r|r|r|r|r|r|l}
\hline
N & Benchmark & MM & RC & Average & Non-ZI MM & 1 day & basis function\\
\hline
50 & 0.0001213 & 0.003914 & 0.005610 & 0.1497 & 0.2995 & 0.3897 & Bspline\\
100 & 0.0000903 & 0.003400 & 0.005045 & 0.1445 & 0.1689 & 0.3822 & Bspline\\
200 & 0.0001204 & 0.003118 & 0.005186 & 0.1489 & 0.1318 & 0.3840 & Bspline\\
500 & 0.0001262 & 0.002646 & 0.004883 & 0.1468 & 0.1190 & 0.3842 & Bspline\\
1000 & 0.0001110 & 0.002508 & 0.004310 & 0.1494 & 0.1225 & 0.3846 & Bspline\\
50 & 0.1074799 & 0.017861 & 0.011053 & 0.1407 & 0.3793 & 0.3897 & Fourier\\
100 & 0.0128407 & 0.009351 & 0.009639 & 0.1378 & 0.1869 & 0.3822 & Fourier\\
200 & 0.0192103 & 0.005651 & 0.006818 & 0.1417 & 0.1385 & 0.3840 & Fourier\\
500 & 0.0070227 & 0.006278 & 0.008148 & 0.1397 & 0.1215 & 0.3842 & Fourier\\
1000 & 0.0001395 & 0.003415 & 0.005745 & 0.1422 & 0.1261 & 0.3846 & Fourier\\
\hline
\end{tabular}

\end{subtable}

\vspace{1em}

\begin{subtable}{\textwidth}
\centering
\caption{Variance}

\centering
\footnotesize
\begin{tabular}{r|r|r|r|r|r|r|l}
\hline
N & Benchmark & MM & RC & Average & Non-ZI MM & 1 day & basis function\\
\hline
50 & 0.0023843 & 0.135726 & 0.119692 & 0.0586 & 2.6930 & 0.0190 & Bspline\\
100 & 0.0010200 & 0.062859 & 0.050002 & 0.0269 & 0.3297 & 0.0085 & Bspline\\
200 & 0.0005190 & 0.034582 & 0.026807 & 0.0155 & 0.1311 & 0.0047 & Bspline\\
500 & 0.0002088 & 0.014532 & 0.010509 & 0.0069 & 0.0479 & 0.0018 & Bspline\\
1000 & 0.0000995 & 0.008773 & 0.005209 & 0.0046 & 0.0246 & 0.0010 & Bspline\\
50 & 21.7150811 & 1.982048 & 1.896337 & 0.3901 & 7.8791 & 0.0190 & Fourier\\
100 & 8.9110659 & 0.808248 & 0.760005 & 0.1662 & 1.5090 & 0.0085 & Fourier\\
200 & 3.9275922 & 0.420307 & 0.355182 & 0.0775 & 0.5864 & 0.0047 & Fourier\\
500 & 1.5453553 & 0.171809 & 0.134104 & 0.0322 & 0.2068 & 0.0018 & Fourier\\
1000 & 0.7198404 & 0.099889 & 0.071308 & 0.0169 & 0.0987 & 0.0010 & Fourier\\
\hline
\end{tabular}

\end{subtable}

\end{table}

\begin{table}[htbp]
\captionsetup{justification=raggedright}
\caption{Effect of deviation of U. {\normalsize{Common settings:} 
{conditional distribution of $Y|X(t),Z$:} Gaussian distribution (with identity link); 
smoothed point-wise; 
{sample size:} $n = 100$; 
{proportion of zero values:} $\mathbb{E}\{\Pr(W_{ij}(t)=0)\}=0.335$; 
{Covariance function of $U(t)$:} squared-exponential with $\rho_u=0.2$; 
$q_g=0$. 
{Benchmark} = benchmark method; 
{Average} = naïve average method; 
{MM} = mixed model based method proposed in Section~\ref{subsec:predict X}; 
{RC} = regression calibration method proposed in Section~\ref{subsec:predict X}; 
{Non-ZI MM} = non-zero-inflated mixed model based method; 
{1 day} = naive one-day method.}}
\centering

\begin{subtable}{\textwidth}
\centering
\caption{Squared Bias}

\centering
\footnotesize
\begin{tabular}{r|r|r|r|r|r|r|l}
\hline
$\sigma_u$ & Benchmark & MM & RC & Average & Non-ZI MM & 1 day & basis function\\
\hline
1 & 0.0000903 & 0.003400 & 0.005045 & 0.1445 & 0.1689 & 0.3822 & Bspline\\
2 & 0.0000903 & 0.014826 & 0.023468 & 0.1913 & 0.2151 & 0.4180 & Bspline\\
3 & 0.0000903 & 0.034820 & 0.076291 & 0.2534 & 1117.2726 & 0.4461 & Bspline\\
1 & 0.0128407 & 0.009351 & 0.009639 & 0.1378 & 0.1869 & 0.3822 & Fourier\\
2 & 0.0128407 & 0.025021 & 0.028312 & 0.1866 & 0.2678 & 0.4180 & Fourier\\
3 & 0.0128407 & 0.097736 & 0.077173 & 0.2517 & 1.8979 & 0.4461 & Fourier\\
\hline
\end{tabular}

\end{subtable}

\vspace{1em}

\begin{subtable}{\textwidth}
\centering
\caption{Variance}

\centering
\footnotesize
\begin{tabular}{r|r|r|r|r|r|r|l}
\hline
$\sigma_u$ & Benchmark & MM & RC & Average & Non-ZI MM & 1 day & basis function\\
\hline
1 & 0.001020 & 0.062859 & 0.050002 & 0.0269 & 0.3297 & 0.0085 & Bspline\\
2 & 0.001020 & 0.189413 & 0.100581 & 0.0262 & 1.0378 & 0.0068 & Bspline\\
3 & 0.001020 & 0.577805 & 0.114981 & 0.0236 & 544845.5255 & 0.0052 & Bspline\\
1 & 8.911066 & 0.808248 & 0.760005 & 0.1662 & 1.5090 & 0.0085 & Fourier\\
2 & 8.911066 & 1.583932 & 0.928256 & 0.1670 & 3.9490 & 0.0068 & Fourier\\
3 & 8.911066 & 9.735902 & 0.842110 & 0.1484 & 168.2378 & 0.0052 & Fourier\\
\hline
\end{tabular}

\end{subtable}

\end{table}

\begin{table}[htbp]
\captionsetup{justification=raggedright}
\caption{Effect of correlation of u. {\normalsize{Common settings:} 
{conditional distribution of $Y|X(t),Z$:} Gaussian distribution (with identity link); 
smoothed point-wise; 
{sample size:} $n = 100$; 
{proportion of zero values:} $\mathbb{E}\{\Pr(W_{ij}(t)=0)\}=0.335$; 
$\sigma_u = 1$; 
$q_g=0$. 
{Benchmark} = benchmark method; 
{Average} = naïve average method; 
{MM} = mixed model based method proposed in Section~\ref{subsec:predict X}; 
{RC} = regression calibration method proposed in Section~\ref{subsec:predict X}; 
{Non-ZI MM} = non-zero-inflated mixed model based method; 
{1 day} = naive one-day method.}}
\centering

\begin{subtable}{\textwidth}
\centering
\caption{Squared Bias}

\centering
\footnotesize
\begin{tabular}{l|r|r|r|r|r|r|r|l}
\hline
correlation function $U$ & $\rho_u$ & Benchmark & MM & RC & Average & Non-ZI MM & 1 day & basis function\\
\hline
AR1 & 0.20 & 0.0000903 & 0.011842 & 0.008007 & 0.1305 & 0.2135 & 0.3678 & Bspline\\
AR1 & 0.40 & 0.0000903 & 0.009996 & 0.006483 & 0.1323 & 0.2004 & 0.3697 & Bspline\\
AR1 & 0.60 & 0.0000903 & 0.007109 & 0.004566 & 0.1361 & 0.1959 & 0.3734 & Bspline\\
Compound symmetry & 0.20 & 0.0000903 & 0.013358 & 0.009150 & 0.1295 & 0.2217 & 0.3664 & Bspline\\
Compound symmetry & 0.40 & 0.0000903 & 0.013094 & 0.008834 & 0.1302 & 0.2255 & 0.3662 & Bspline\\
Compound symmetry & 0.60 & 0.0000903 & 0.012801 & 0.008557 & 0.1304 & 0.2263 & 0.3670 & Bspline\\
Squared Exponential & 0.10 & 0.0000903 & 0.005626 & 0.004464 & 0.1427 & 0.1840 & 0.3790 & Bspline\\
Squared Exponential & 0.15 & 0.0000903 & 0.003774 & 0.003764 & 0.1448 & 0.1764 & 0.3812 & Bspline\\
Squared Exponential & 0.20 & 0.0000903 & 0.003400 & 0.005045 & 0.1445 & 0.1689 & 0.3822 & Bspline\\
AR1 & 0.20 & 0.0128407 & 0.011971 & 0.008180 & 0.1233 & 0.2384 & 0.3678 & Fourier\\
AR1 & 0.40 & 0.0128407 & 0.009788 & 0.006552 & 0.1241 & 0.2179 & 0.3697 & Fourier\\
AR1 & 0.60 & 0.0128407 & 0.007047 & 0.005011 & 0.1282 & 0.2117 & 0.3734 & Fourier\\
Compound symmetry & 0.20 & 0.0128407 & 0.013439 & 0.009299 & 0.1214 & 0.2515 & 0.3664 & Fourier\\
Compound symmetry & 0.40 & 0.0128407 & 0.013689 & 0.009518 & 0.1227 & 0.2598 & 0.3662 & Fourier\\
Compound symmetry & 0.60 & 0.0128407 & 0.013263 & 0.009373 & 0.1219 & 0.2699 & 0.3670 & Fourier\\
Squared Exponential & 0.10 & 0.0128407 & 0.009227 & 0.007960 & 0.1367 & 0.2050 & 0.3790 & Fourier\\
Squared Exponential & 0.15 & 0.0128407 & 0.009174 & 0.010184 & 0.1377 & 0.1912 & 0.3812 & Fourier\\
Squared Exponential & 0.20 & 0.0128407 & 0.009351 & 0.009639 & 0.1378 & 0.1869 & 0.3822 & Fourier\\
\hline
\end{tabular}

\end{subtable}

\vspace{1em}

\begin{subtable}{\textwidth}
\centering
\caption{Variance}

\centering
\footnotesize
\begin{tabular}{l|r|r|r|r|r|r|r|l}
\hline
correlation function $U$ & $\rho_u$ & Benchmark & MM & RC & Average & Non-ZI MM & 1 day & basis function\\
\hline
AR1 & 0.20 & 0.001020 & 0.021910 & 0.013392 & 0.0252 & 0.3825 & 0.0085 & Bspline\\
AR1 & 0.40 & 0.001020 & 0.025064 & 0.016814 & 0.0252 & 0.3101 & 0.0083 & Bspline\\
AR1 & 0.60 & 0.001020 & 0.032598 & 0.023383 & 0.0252 & 0.2978 & 0.0081 & Bspline\\
Compound symmetry & 0.20 & 0.001020 & 0.022411 & 0.012457 & 0.0259 & 0.3221 & 0.0086 & Bspline\\
Compound symmetry & 0.40 & 0.001020 & 0.024365 & 0.014776 & 0.0258 & 0.3433 & 0.0085 & Bspline\\
Compound symmetry & 0.60 & 0.001020 & 0.028158 & 0.017721 & 0.0261 & 0.3292 & 0.0088 & Bspline\\
Squared Exponential & 0.10 & 0.001020 & 0.036899 & 0.030264 & 0.0255 & 0.3730 & 0.0080 & Bspline\\
Squared Exponential & 0.15 & 0.001020 & 0.052564 & 0.043696 & 0.0262 & 0.3314 & 0.0084 & Bspline\\
Squared Exponential & 0.20 & 0.001020 & 0.062859 & 0.050002 & 0.0269 & 0.3297 & 0.0085 & Bspline\\
AR1 & 0.20 & 8.911066 & 0.096569 & 0.066691 & 0.1245 & 1.2455 & 0.0085 & Fourier\\
AR1 & 0.40 & 8.911066 & 0.123998 & 0.093512 & 0.1246 & 1.2295 & 0.0083 & Fourier\\
AR1 & 0.60 & 8.911066 & 0.176484 & 0.142422 & 0.1297 & 1.2172 & 0.0081 & Fourier\\
Compound symmetry & 0.20 & 8.911066 & 0.115153 & 0.070980 & 0.1298 & 1.5421 & 0.0086 & Fourier\\
Compound symmetry & 0.40 & 8.911066 & 0.135990 & 0.090871 & 0.1370 & 1.6842 & 0.0085 & Fourier\\
Compound symmetry & 0.60 & 8.911066 & 0.175110 & 0.126550 & 0.1428 & 2.5201 & 0.0088 & Fourier\\
Squared Exponential & 0.10 & 8.911066 & 0.370782 & 0.334086 & 0.1482 & 1.3609 & 0.0080 & Fourier\\
Squared Exponential & 0.15 & 8.911066 & 0.710734 & 0.652240 & 0.1638 & 1.4703 & 0.0084 & Fourier\\
Squared Exponential & 0.20 & 8.911066 & 0.808248 & 0.760005 & 0.1662 & 1.5090 & 0.0085 & Fourier\\
\hline
\end{tabular}

\end{subtable}

\end{table}

\begin{table}[htbp]
\captionsetup{justification=raggedright}
\caption{Effect of proportion of zero values in W. {\normalsize{Common settings:} 
{conditional distribution of $Y|X(t),Z$:} Gaussian distribution (with identity link); 
smoothed point-wise; 
{sample size:} $n = 100$; 
{Covariance function of $U(t)$:} squared-exponential with $\rho_u=0.2$; 
$\sigma_u = 1$; 
$q_g=0$. 
{Benchmark} = benchmark method; 
{Average} = naïve average method; 
{MM} = mixed model based method proposed in Section~\ref{subsec:predict X}; 
{RC} = regression calibration method proposed in Section~\ref{subsec:predict X}; 
{Non-ZI MM} = non-zero-inflated mixed model based method; 
{1 day} = naive one-day method.}}
\centering

\begin{subtable}{\textwidth}
\centering
\caption{Squared Bias}

\centering
\footnotesize
\begin{tabular}{r|r|r|r|r|r|r|l}
\hline
E(Pr{W=0}) & Benchmark & MM & RC & Average & Non-ZI MM & 1 day & basis function\\
\hline
0.255 & 0.0000903 & 0.006585 & 0.007092 & 0.1074 & 0.0986 & 0.3451 & Bspline\\
0.294 & 0.0000903 & 0.004534 & 0.005764 & 0.1258 & 0.1323 & 0.3641 & Bspline\\
0.335 & 0.0000903 & 0.003400 & 0.005045 & 0.1445 & 0.1689 & 0.3822 & Bspline\\
0.403 & 0.0000903 & 0.002549 & 0.003986 & 0.1770 & 0.2860 & 0.4056 & Bspline\\
0.255 & 0.0128407 & 0.011672 & 0.011569 & 0.1034 & 0.1074 & 0.3451 & Fourier\\
0.294 & 0.0128407 & 0.011333 & 0.010548 & 0.1192 & 0.1371 & 0.3641 & Fourier\\
0.335 & 0.0128407 & 0.009351 & 0.009639 & 0.1378 & 0.1869 & 0.3822 & Fourier\\
0.403 & 0.0128407 & 0.005167 & 0.009180 & 0.1684 & 0.3403 & 0.4056 & Fourier\\
\hline
\end{tabular}

\end{subtable}

\vspace{1em}

\begin{subtable}{\textwidth}
\centering
\caption{Variance}

\centering
\footnotesize
\begin{tabular}{r|r|r|r|r|r|r|l}
\hline
E(Pr{W=0}) & Benchmark & MM & RC & Average & Non-ZI MM & 1 day & basis function\\
\hline
0.255 & 0.001020 & 0.056708 & 0.050282 & 0.0313 & 0.1827 & 0.0120 & Bspline\\
0.294 & 0.001020 & 0.060766 & 0.050660 & 0.0293 & 0.2373 & 0.0099 & Bspline\\
0.335 & 0.001020 & 0.062859 & 0.050002 & 0.0269 & 0.3297 & 0.0085 & Bspline\\
0.403 & 0.001020 & 0.077295 & 0.051715 & 0.0231 & 0.8235 & 0.0065 & Bspline\\
0.255 & 8.911066 & 1.071181 & 0.998176 & 0.1829 & 0.9043 & 0.0120 & Fourier\\
0.294 & 8.911066 & 0.950327 & 0.877272 & 0.1721 & 1.1069 & 0.0099 & Fourier\\
0.335 & 8.911066 & 0.808248 & 0.760005 & 0.1662 & 1.5090 & 0.0085 & Fourier\\
0.403 & 8.911066 & 0.655488 & 0.615357 & 0.1429 & 3.2846 & 0.0065 & Fourier\\
\hline
\end{tabular}

\end{subtable}

\end{table}

\begin{table}[htbp]
\captionsetup{justification=raggedright}
\caption{Effect of $q_g$ for $G(t)$ in Eq~\eqref{eq:sim G}. {\normalsize{Common settings:} 
{conditional distribution of $Y|X(t),Z$:} Gaussian distribution (with identity link); 
smoothed point-wise; 
{sample size:} $n = 100$; 
{proportion of zero values:} $\mathbb{E}\{\Pr(W_{ij}(t)=0)\}=0.335$; 
{Covariance function of $U(t)$:} squared-exponential with $\rho_u=0.2$; 
$\sigma_u = 1$; 
{Benchmark} = benchmark method; 
{Average} = naïve average method; 
{MM} = mixed model based method proposed in Section~\ref{subsec:predict X}; 
{RC} = regression calibration method proposed in Section~\ref{subsec:predict X}; 
{Non-ZI MM} = non-zero-inflated mixed model based method; 
{1 day} = naive one-day method.}}
\centering

\begin{subtable}{\textwidth}
\centering
\caption{Squared Bias}

\centering
\footnotesize
\begin{tabular}{r|r|r|r|r|r|r|l}
\hline
$q_g$ & Benchmark & MM & RC & Average & Non-ZI MM & 1 day & basis function\\
\hline
0.0 & 0.0000903 & 0.003400 & 0.005045 & 0.1445 & 0.1689 & 0.3822 & Bspline\\
0.2 & 0.0000903 & 0.001995 & 0.004874 & 0.1618 & 0.0827 & 0.3818 & Bspline\\
0.4 & 0.0000903 & 0.004308 & 0.004004 & 0.2052 & 0.0951 & 0.3820 & Bspline\\
0.0 & 0.0128407 & 0.009351 & 0.009639 & 0.1378 & 0.1869 & 0.3822 & Fourier\\
0.2 & 0.0128407 & 0.011823 & 0.015466 & 0.1536 & 0.0764 & 0.3818 & Fourier\\
0.4 & 0.0128407 & 0.003327 & 0.008107 & 0.1970 & 0.0844 & 0.3820 & Fourier\\
\hline
\end{tabular}

\end{subtable}

\vspace{1em}

\begin{subtable}{\textwidth}
\centering
\caption{Variance}

\centering
\footnotesize
\begin{tabular}{r|r|r|r|r|r|r|l}
\hline
$q_g$ & Benchmark & MM & RC & Average & Non-ZI MM & 1 day & basis function\\
\hline
0.0 & 0.001020 & 0.062859 & 0.050002 & 0.0269 & 0.3297 & 0.0085 & Bspline\\
0.2 & 0.001020 & 0.072695 & 0.049978 & 0.0256 & 0.1898 & 0.0082 & Bspline\\
0.4 & 0.001020 & 0.096320 & 0.049105 & 0.0233 & 0.0789 & 0.0081 & Bspline\\
0.0 & 8.911066 & 0.808248 & 0.760005 & 0.1662 & 1.5090 & 0.0085 & Fourier\\
0.2 & 8.911066 & 0.765011 & 0.723594 & 0.1547 & 0.9039 & 0.0082 & Fourier\\
0.4 & 8.911066 & 0.727050 & 0.678187 & 0.1687 & 0.5337 & 0.0081 & Fourier\\
\hline
\end{tabular}

\end{subtable}

\end{table}

\begin{table}[htbp]
\captionsetup{justification=raggedright}
\caption{Effect of sample size. {\normalsize{Common settings:} 
Bernoulli distribution (with logit link); 
{estimation method of $\theta(t)$:} point-wise; 
{proportion of zero values:} $\mathbb{E}\{\Pr(W_{ij}(t)=0)\}=0.335$; 
{Covariance function of $U(t)$:} squared-exponential with $\rho_u=0.2$; 
$\sigma_u = 1$; 
$q_g=0$. 
{Benchmark} = benchmark method; 
{Average} = naïve average method; 
{MM} = mixed model based method proposed in Section~\ref{subsec:predict X}; 
{RC} = regression calibration method proposed in Section~\ref{subsec:predict X}; 
{Non-ZI MM} = non-zero-inflated mixed model based method; 
{1 day} = naive one-day method.}}
\centering

\begin{subtable}{\textwidth}
\centering
\caption{Squared Bias}

\centering
\footnotesize
\begin{tabular}{r|r|r|r|r|r|r|l}
\hline
N & Benchmark & MM & RC & Average & Non-ZI MM & 1 day & basis function\\
\hline
50 & 0.1230072 & 0.074700 & 0.064624 & 0.0758 & 0.6532 & 0.3373 & Bspline\\
100 & 0.0150326 & 0.049570 & 0.033648 & 0.1008 & 0.2145 & 0.3878 & Bspline\\
200 & 0.0119906 & 0.012999 & 0.009027 & 0.1610 & 0.1734 & 0.3838 & Bspline\\
500 & 0.0117602 & 0.007565 & 0.006724 & 0.1439 & 0.0901 & 0.3855 & Bspline\\
1000 & 0.0061865 & 0.005625 & 0.009590 & 0.1508 & 0.1286 & 0.3905 & Bspline\\
50 & 8510.7222608 & 276.106580 & 487.760915 & 0.3125 & 171.2551 & 0.3373 & Fourier\\
100 & 107.4245267 & 0.546630 & 0.575267 & 0.1014 & 0.5080 & 0.3878 & Fourier\\
200 & 21.4239549 & 0.148504 & 0.218869 & 0.1558 & 0.2650 & 0.3838 & Fourier\\
500 & 13.6864192 & 0.023105 & 0.051624 & 0.1375 & 0.1093 & 0.3855 & Fourier\\
1000 & 19.6685291 & 0.059636 & 0.053697 & 0.1404 & 0.1304 & 0.3905 & Fourier\\
\hline
\end{tabular}

\end{subtable}

\vspace{1em}

\begin{subtable}{\textwidth}
\centering
\caption{Variance}

\centering
\footnotesize
\begin{tabular}{r|r|r|r|r|r|r|l}
\hline
N & Benchmark & MM & RC & Average & Non-ZI MM & 1 day & basis function\\
\hline
50 & 4.988905e+01 & 4.308465e+01 & 4.164993e+01 & 6.6078 & 226.5763 & 1.0837 & Bspline\\
100 & 1.430998e+01 & 1.236682e+01 & 1.166427e+01 & 1.9058 & 20.1516 & 0.3184 & Bspline\\
200 & 5.874116e+00 & 5.087259e+00 & 5.147803e+00 & 0.8319 & 6.0889 & 0.1413 & Bspline\\
500 & 2.413604e+00 & 1.849339e+00 & 1.890382e+00 & 0.2662 & 1.7644 & 0.0527 & Bspline\\
1000 & 1.060543e+00 & 9.195330e-01 & 9.573540e-01 & 0.1297 & 0.8577 & 0.0248 & Bspline\\
50 & 4.284808e+06 & 1.719690e+05 & 4.035298e+05 & 73.6180 & 86927.1378 & 1.0837 & Fourier\\
100 & 1.339427e+05 & 1.957815e+02 & 1.937196e+02 & 13.2699 & 120.6445 & 0.3184 & Fourier\\
200 & 5.053468e+04 & 7.489478e+01 & 8.055835e+01 & 5.1519 & 34.0757 & 0.1413 & Fourier\\
500 & 1.651193e+04 & 2.220087e+01 & 2.520010e+01 & 1.6742 & 10.6785 & 0.0527 & Fourier\\
1000 & 7.949001e+03 & 1.067807e+01 & 1.254839e+01 & 0.8060 & 4.9957 & 0.0248 & Fourier\\
\hline
\end{tabular}

\end{subtable}

\end{table}

\begin{table}[htbp]
\captionsetup{justification=raggedright}
\caption{Effect of deviation of U. {\normalsize{Common settings:} 
Bernoulli distribution (with logit link); 
{estimation method of $\theta(t)$:} point-wise; 
{sample size:} $n = 100$; 
{proportion of zero values:} $\mathbb{E}\{\Pr(W_{ij}(t)=0)\}=0.335$; 
{Covariance function of $U(t)$:} squared-exponential with $\rho_u=0.2$; 
$q_g=0$. 
{Benchmark} = benchmark method; 
{Average} = naïve average method; 
{MM} = mixed model based method proposed in Section~\ref{subsec:predict X}; 
{RC} = regression calibration method proposed in Section~\ref{subsec:predict X}; 
{Non-ZI MM} = non-zero-inflated mixed model based method; 
{1 day} = naive one-day method.}}
\centering

\begin{subtable}{\textwidth}
\centering
\caption{Squared Bias}

\centering
\footnotesize
\begin{tabular}{r|r|r|r|r|r|r|l}
\hline
$\sigma_u$ & Benchmark & MM & RC & Average & Non-ZI MM & 1 day & basis function\\
\hline
1 & 0.0150326 & 0.049570 & 0.033648 & 0.1008 & 0.2145 & 0.3878 & Bspline\\
2 & 0.0150326 & 0.111989 & 0.051939 & 0.1479 & 0.1762 & 0.4200 & Bspline\\
3 & 0.0150326 & 0.142549 & 0.080534 & 0.2269 & 569.5784 & 0.4411 & Bspline\\
1 & 107.4245267 & 0.546630 & 0.575267 & 0.1014 & 0.5080 & 0.3878 & Fourier\\
2 & 107.4245267 & 0.384125 & 0.257744 & 0.1443 & 0.8877 & 0.4200 & Fourier\\
3 & 107.4245267 & 0.922071 & 0.102786 & 0.2223 & 14.4792 & 0.4411 & Fourier\\
\hline
\end{tabular}

\end{subtable}

\vspace{1em}

\begin{subtable}{\textwidth}
\centering
\caption{Variance}

\centering
\footnotesize
\begin{tabular}{r|r|r|r|r|r|r|l}
\hline
$\sigma_u$ & Benchmark & MM & RC & Average & Non-ZI MM & 1 day & basis function\\
\hline
1 & 14.30998 & 12.36682 & 11.664274 & 1.9058 & 20.1516 & 0.3184 & Bspline\\
2 & 14.30998 & 15.11924 & 8.756273 & 1.5582 & 51.9469 & 0.2542 & Bspline\\
3 & 14.30998 & 48.61541 & 6.456680 & 1.2733 & 289780.6263 & 0.1953 & Bspline\\
1 & 133942.68329 & 195.78154 & 193.719591 & 13.2699 & 120.6445 & 0.3184 & Fourier\\
2 & 133942.68329 & 150.64076 & 87.686845 & 10.9018 & 253.1402 & 0.2542 & Fourier\\
3 & 133942.68329 & 491.90707 & 56.412395 & 8.6190 & 17416.4713 & 0.1953 & Fourier\\
\hline
\end{tabular}

\end{subtable}

\end{table}

\begin{table}[htbp]
\captionsetup{justification=raggedright}
\caption{Effect of correlation of u. {\normalsize{Common settings:} 
Bernoulli distribution (with logit link); 
{estimation method of $\theta(t)$:} point-wise; 
{sample size:} $n = 100$; 
{proportion of zero values:} $\mathbb{E}\{\Pr(W_{ij}(t)=0)\}=0.335$; 
$\sigma_u = 1$; 
$q_g=0$. 
{Benchmark} = benchmark method; 
{Average} = naïve average method; 
{MM} = mixed model based method proposed in Section~\ref{subsec:predict X}; 
{RC} = regression calibration method proposed in Section~\ref{subsec:predict X}; 
{Non-ZI MM} = non-zero-inflated mixed model based method; 
{1 day} = naive one-day method.}}
\centering

\begin{subtable}{\textwidth}
\centering
\caption{Squared Bias}

\centering
\footnotesize
\begin{tabular}{l|r|r|r|r|r|r|r|l}
\hline
correlation function $U$ & $\rho_u$ & Benchmark & MM & RC & Average & Non-ZI MM & 1 day & basis function\\
\hline
AR1 & 0.20 & 0.0150326 & 0.047444 & 0.040433 & 0.0961 & 0.3036 & 0.3709 & Bspline\\
AR1 & 0.40 & 0.0150326 & 0.052652 & 0.042118 & 0.0950 & 0.3248 & 0.3748 & Bspline\\
AR1 & 0.60 & 0.0150326 & 0.040998 & 0.030933 & 0.0969 & 0.3013 & 0.3794 & Bspline\\
Compound symmetry & 0.20 & 0.0150326 & 0.061152 & 0.057405 & 0.0929 & 0.3586 & 0.3701 & Bspline\\
Compound symmetry & 0.40 & 0.0150326 & 0.052440 & 0.043674 & 0.0946 & 0.4973 & 0.3711 & Bspline\\
Compound symmetry & 0.60 & 0.0150326 & 0.056921 & 0.052824 & 0.0941 & 0.3707 & 0.3706 & Bspline\\
Squared Exponential & 0.10 & 0.0150326 & 0.022937 & 0.017656 & 0.1019 & 0.2524 & 0.3819 & Bspline\\
Squared Exponential & 0.15 & 0.0150326 & 0.025354 & 0.020654 & 0.1041 & 0.2693 & 0.3866 & Bspline\\
Squared Exponential & 0.20 & 0.0150326 & 0.049570 & 0.033648 & 0.1008 & 0.2145 & 0.3878 & Bspline\\
AR1 & 0.20 & 107.4245267 & 0.274307 & 0.250685 & 0.0974 & 0.5523 & 0.3709 & Fourier\\
AR1 & 0.40 & 107.4245267 & 0.150216 & 0.119055 & 0.1061 & 0.5633 & 0.3748 & Fourier\\
AR1 & 0.60 & 107.4245267 & 0.327863 & 0.227200 & 0.1079 & 0.5313 & 0.3794 & Fourier\\
Compound symmetry & 0.20 & 107.4245267 & 0.210395 & 0.172365 & 0.0935 & 0.6425 & 0.3701 & Fourier\\
Compound symmetry & 0.40 & 107.4245267 & 0.165175 & 0.182104 & 0.0905 & 0.6127 & 0.3711 & Fourier\\
Compound symmetry & 0.60 & 107.4245267 & 0.147165 & 0.116096 & 0.0916 & 0.6691 & 0.3706 & Fourier\\
Squared Exponential & 0.10 & 107.4245267 & 0.379175 & 0.379208 & 0.1065 & 0.5471 & 0.3819 & Fourier\\
Squared Exponential & 0.15 & 107.4245267 & 0.261578 & 0.389972 & 0.1134 & 0.5281 & 0.3866 & Fourier\\
Squared Exponential & 0.20 & 107.4245267 & 0.546630 & 0.575267 & 0.1014 & 0.5080 & 0.3878 & Fourier\\
\hline
\end{tabular}

\end{subtable}

\vspace{1em}

\begin{subtable}{\textwidth}
\centering
\caption{Variance}

\centering
\footnotesize
\begin{tabular}{l|r|r|r|r|r|r|r|l}
\hline
correlation function $U$ & $\rho_u$ & Benchmark & MM & RC & Average & Non-ZI MM & 1 day & basis function\\
\hline
AR1 & 0.20 & 14.30998 & 9.189825 & 8.914817 & 1.9079 & 19.1846 & 0.3298 & Bspline\\
AR1 & 0.40 & 14.30998 & 9.564324 & 9.233625 & 1.8525 & 18.1525 & 0.3209 & Bspline\\
AR1 & 0.60 & 14.30998 & 9.923213 & 9.475688 & 1.8980 & 18.8789 & 0.3144 & Bspline\\
Compound symmetry & 0.20 & 14.30998 & 10.899403 & 10.620337 & 1.9477 & 18.3008 & 0.3291 & Bspline\\
Compound symmetry & 0.40 & 14.30998 & 10.895815 & 10.663124 & 1.8911 & 38.9483 & 0.3287 & Bspline\\
Compound symmetry & 0.60 & 14.30998 & 11.548294 & 11.015541 & 1.9159 & 18.2469 & 0.3316 & Bspline\\
Squared Exponential & 0.10 & 14.30998 & 8.560971 & 8.125550 & 1.8312 & 18.2983 & 0.3055 & Bspline\\
Squared Exponential & 0.15 & 14.30998 & 11.100982 & 10.645155 & 1.8191 & 18.4807 & 0.3106 & Bspline\\
Squared Exponential & 0.20 & 14.30998 & 12.366820 & 11.664274 & 1.9058 & 20.1516 & 0.3184 & Bspline\\
AR1 & 0.20 & 133942.68329 & 60.189778 & 58.157740 & 11.0185 & 94.0496 & 0.3298 & Fourier\\
AR1 & 0.40 & 133942.68329 & 62.434690 & 59.305086 & 10.8672 & 101.7753 & 0.3209 & Fourier\\
AR1 & 0.60 & 133942.68329 & 73.484362 & 70.415540 & 11.7778 & 119.6064 & 0.3144 & Fourier\\
Compound symmetry & 0.20 & 133942.68329 & 78.125086 & 75.306457 & 11.5590 & 95.3529 & 0.3291 & Fourier\\
Compound symmetry & 0.40 & 133942.68329 & 87.018630 & 84.761023 & 11.8844 & 103.7438 & 0.3287 & Fourier\\
Compound symmetry & 0.60 & 133942.68329 & 104.651268 & 102.823516 & 12.6293 & 108.1200 & 0.3316 & Fourier\\
Squared Exponential & 0.10 & 133942.68329 & 119.483667 & 120.043084 & 12.1252 & 114.2930 & 0.3055 & Fourier\\
Squared Exponential & 0.15 & 133942.68329 & 175.507967 & 174.291604 & 13.0771 & 115.7622 & 0.3106 & Fourier\\
Squared Exponential & 0.20 & 133942.68329 & 195.781545 & 193.719591 & 13.2699 & 120.6445 & 0.3184 & Fourier\\
\hline
\end{tabular}

\end{subtable}

\end{table}

\begin{table}[htbp]
\captionsetup{justification=raggedright}
\caption{Effect of proportion of zero values in W. {\normalsize{Common settings:} 
Bernoulli distribution (with logit link); 
{estimation method of $\theta(t)$:} point-wise; 
{sample size:} $n = 100$; 
{Covariance function of $U(t)$:} squared-exponential with $\rho_u=0.2$; 
$\sigma_u = 1$; 
$q_g=0$. 
{Benchmark} = benchmark method; 
{Average} = naïve average method; 
{MM} = mixed model based method proposed in Section~\ref{subsec:predict X}; 
{RC} = regression calibration method proposed in Section~\ref{subsec:predict X}; 
{Non-ZI MM} = non-zero-inflated mixed model based method; 
{1 day} = naive one-day method.}}
\centering

\begin{subtable}{\textwidth}
\centering
\caption{Squared Bias}

\centering
\footnotesize
\begin{tabular}{r|r|r|r|r|r|r|l}
\hline
E(Pr{W=0}) & Benchmark & MM & RC & Average & Non-ZI MM & 1 day & basis function\\
\hline
0.255 & 0.0150326 & 0.050779 & 0.039638 & 0.0645 & 0.1668 & 0.3503 & Bspline\\
0.294 & 0.0150326 & 0.040168 & 0.028547 & 0.0772 & 0.1833 & 0.3699 & Bspline\\
0.335 & 0.0150326 & 0.049570 & 0.033648 & 0.1008 & 0.2145 & 0.3878 & Bspline\\
0.403 & 0.0150326 & 0.060177 & 0.036469 & 0.1329 & 0.2759 & 0.4077 & Bspline\\
0.255 & 107.4245267 & 1.280489 & 0.903613 & 0.0807 & 0.3221 & 0.3503 & Fourier\\
0.294 & 107.4245267 & 0.442614 & 0.523684 & 0.0904 & 0.4478 & 0.3699 & Fourier\\
0.335 & 107.4245267 & 0.546630 & 0.575267 & 0.1014 & 0.5080 & 0.3878 & Fourier\\
0.403 & 107.4245267 & 0.230456 & 0.635797 & 0.1276 & 1.0988 & 0.4077 & Fourier\\
\hline
\end{tabular}

\end{subtable}

\vspace{1em}

\begin{subtable}{\textwidth}
\centering
\caption{Variance}

\centering
\footnotesize
\begin{tabular}{r|r|r|r|r|r|r|l}
\hline
E(Pr{W=0}) & Benchmark & MM & RC & Average & Non-ZI MM & 1 day & basis function\\
\hline
0.255 & 14.30998 & 13.33442 & 12.76470 & 2.5310 & 13.1404 & 0.4871 & Bspline\\
0.294 & 14.30998 & 12.83989 & 12.33306 & 2.2693 & 16.4019 & 0.3980 & Bspline\\
0.335 & 14.30998 & 12.36682 & 11.66427 & 1.9058 & 20.1516 & 0.3184 & Bspline\\
0.403 & 14.30998 & 10.72348 & 10.80394 & 1.5109 & 53.4614 & 0.2436 & Bspline\\
0.255 & 133942.68329 & 279.38631 & 261.94404 & 19.6025 & 95.7604 & 0.4871 & Fourier\\
0.294 & 133942.68329 & 238.04848 & 223.39025 & 15.8953 & 105.0335 & 0.3980 & Fourier\\
0.335 & 133942.68329 & 195.78154 & 193.71959 & 13.2699 & 120.6445 & 0.3184 & Fourier\\
0.403 & 133942.68329 & 132.47330 & 156.03496 & 9.8360 & 191.8466 & 0.2436 & Fourier\\
\hline
\end{tabular}

\end{subtable}

\end{table}

\begin{table}[htbp]
\captionsetup{justification=raggedright}
\caption{Effect of $q_g$ for $G(t)$ in Eq~\eqref{eq:sim G}. {\normalsize{Common settings:} 
Bernoulli distribution (with logit link); 
{estimation method of $\theta(t)$:} point-wise; 
{sample size:} $n = 100$; 
{proportion of zero values:} $\mathbb{E}\{\Pr(W_{ij}(t)=0)\}=0.335$; 
{Covariance function of $U(t)$:} squared-exponential with $\rho_u=0.2$; 
$\sigma_u = 1$; 
{Benchmark} = benchmark method; 
{Average} = naïve average method; 
{MM} = mixed model based method proposed in Section~\ref{subsec:predict X}; 
{RC} = regression calibration method proposed in Section~\ref{subsec:predict X}; 
{Non-ZI MM} = non-zero-inflated mixed model based method; 
{1 day} = naive one-day method.}}
\centering

\begin{subtable}{\textwidth}
\centering
\caption{Squared Bias}

\centering
\footnotesize
\begin{tabular}{r|r|r|r|r|r|r|l}
\hline
$q_g$ & Benchmark & MM & RC & Average & Non-ZI MM & 1 day & basis function\\
\hline
0.0 & 0.0150326 & 0.049570 & 0.033648 & 0.1008 & 0.2145 & 0.3878 & Bspline\\
0.2 & 0.0150326 & 0.044423 & 0.033209 & 0.1330 & 0.1476 & 0.3887 & Bspline\\
0.4 & 0.0150326 & 0.022480 & 0.031240 & 0.1829 & 0.1060 & 0.3876 & Bspline\\
0.0 & 107.4245267 & 0.546630 & 0.575267 & 0.1014 & 0.5080 & 0.3878 & Fourier\\
0.2 & 107.4245267 & 0.344596 & 0.517744 & 0.1545 & 0.4463 & 0.3887 & Fourier\\
0.4 & 107.4245267 & 0.215212 & 0.468406 & 0.1729 & 0.1486 & 0.3876 & Fourier\\
\hline
\end{tabular}

\end{subtable}

\vspace{1em}

\begin{subtable}{\textwidth}
\centering
\caption{Variance}

\centering
\footnotesize
\begin{tabular}{r|r|r|r|r|r|r|l}
\hline
$q_g$ & Benchmark & MM & RC & Average & Non-ZI MM & 1 day & basis function\\
\hline
0.0 & 14.30998 & 12.366820 & 11.66427 & 1.9058 & 20.1516 & 0.3184 & Bspline\\
0.2 & 14.30998 & 11.833070 & 11.60735 & 1.6796 & 10.8239 & 0.3191 & Bspline\\
0.4 & 14.30998 & 8.546491 & 11.00679 & 1.2983 & 4.2749 & 0.3277 & Bspline\\
0.0 & 133942.68329 & 195.781545 & 193.71959 & 13.2699 & 120.6445 & 0.3184 & Fourier\\
0.2 & 133942.68329 & 182.521904 & 189.75403 & 13.0345 & 76.2574 & 0.3191 & Fourier\\
0.4 & 133942.68329 & 117.085774 & 163.64099 & 10.4905 & 33.8754 & 0.3277 & Fourier\\
\hline
\end{tabular}

\end{subtable}

\end{table}

\begin{table}[htbp]
\captionsetup{justification=raggedright}
\caption{Effect of sample size. {\normalsize{Common settings:} 
Bernoulli distribution (with logit link); 
smoothed point-wise; 
{proportion of zero values:} $\mathbb{E}\{\Pr(W_{ij}(t)=0)\}=0.335$; 
{Covariance function of $U(t)$:} squared-exponential with $\rho_u=0.2$; 
$\sigma_u = 1$; 
$q_g=0$. 
{Benchmark} = benchmark method; 
{Average} = naïve average method; 
{MM} = mixed model based method proposed in Section~\ref{subsec:predict X}; 
{RC} = regression calibration method proposed in Section~\ref{subsec:predict X}; 
{Non-ZI MM} = non-zero-inflated mixed model based method; 
{1 day} = naive one-day method.}}
\centering

\begin{subtable}{\textwidth}
\centering
\caption{Squared Bias}

\centering
\footnotesize
\begin{tabular}{r|r|r|r|r|r|r|l}
\hline
N & Benchmark & MM & RC & Average & Non-ZI MM & 1 day & basis function\\
\hline
50 & 0.1230072 & 0.073128 & 0.063324 & 0.0758 & 0.6532 & 0.3373 & Bspline\\
100 & 0.0150326 & 0.050386 & 0.035714 & 0.1008 & 0.2145 & 0.3878 & Bspline\\
200 & 0.0119906 & 0.012468 & 0.008708 & 0.1610 & 0.1734 & 0.3838 & Bspline\\
500 & 0.0117602 & 0.007298 & 0.006896 & 0.1439 & 0.0901 & 0.3855 & Bspline\\
1000 & 0.0061865 & 0.005265 & 0.009536 & 0.1508 & 0.1286 & 0.3905 & Bspline\\
50 & 8510.7222608 & 11307.977149 & 3.445035 & 0.3125 & 171.2551 & 0.3373 & Fourier\\
100 & 107.4245267 & 0.508382 & 0.433005 & 0.1014 & 0.5080 & 0.3878 & Fourier\\
200 & 21.4239549 & 0.151802 & 0.197691 & 0.1558 & 0.2650 & 0.3838 & Fourier\\
500 & 13.6864192 & 0.021474 & 0.053488 & 0.1375 & 0.1093 & 0.3855 & Fourier\\
1000 & 19.6685291 & 0.057265 & 0.051438 & 0.1404 & 0.1304 & 0.3905 & Fourier\\
\hline
\end{tabular}

\end{subtable}

\vspace{1em}

\begin{subtable}{\textwidth}
\centering
\caption{Variance}

\centering
\footnotesize
\begin{tabular}{r|r|r|r|r|r|r|l}
\hline
N & Benchmark & MM & RC & Average & Non-ZI MM & 1 day & basis function\\
\hline
50 & 4.988905e+01 & 4.299165e+01 & 41.845801 & 6.6078 & 226.5763 & 1.0837 & Bspline\\
100 & 1.430998e+01 & 1.230884e+01 & 11.617732 & 1.9058 & 20.1516 & 0.3184 & Bspline\\
200 & 5.874116e+00 & 5.085867e+00 & 5.139273 & 0.8319 & 6.0889 & 0.1413 & Bspline\\
500 & 2.413604e+00 & 1.836943e+00 & 1.880630 & 0.2662 & 1.7644 & 0.0527 & Bspline\\
1000 & 1.060543e+00 & 9.194920e-01 & 0.956270 & 0.1297 & 0.8577 & 0.0248 & Bspline\\
50 & 4.284808e+06 & 5.636157e+06 & 1428.345493 & 73.6180 & 86927.1378 & 1.0837 & Fourier\\
100 & 1.339427e+05 & 1.990476e+02 & 195.187175 & 13.2699 & 120.6445 & 0.3184 & Fourier\\
200 & 5.053468e+04 & 7.541321e+01 & 80.556140 & 5.1519 & 34.0757 & 0.1413 & Fourier\\
500 & 1.651193e+04 & 2.225222e+01 & 25.129643 & 1.6742 & 10.6785 & 0.0527 & Fourier\\
1000 & 7.949001e+03 & 1.057006e+01 & 12.464185 & 0.8060 & 4.9957 & 0.0248 & Fourier\\
\hline
\end{tabular}

\end{subtable}

\end{table}

\begin{table}[htbp]
\captionsetup{justification=raggedright}
\caption{Effect of deviation of U. {\normalsize{Common settings:} 
Bernoulli distribution (with logit link); 
smoothed point-wise; 
{sample size:} $n = 100$; 
{proportion of zero values:} $\mathbb{E}\{\Pr(W_{ij}(t)=0)\}=0.335$; 
{Covariance function of $U(t)$:} squared-exponential with $\rho_u=0.2$; 
$q_g=0$. 
{Benchmark} = benchmark method; 
{Average} = naïve average method; 
{MM} = mixed model based method proposed in Section~\ref{subsec:predict X}; 
{RC} = regression calibration method proposed in Section~\ref{subsec:predict X}; 
{Non-ZI MM} = non-zero-inflated mixed model based method; 
{1 day} = naive one-day method.}}
\centering

\begin{subtable}{\textwidth}
\centering
\caption{Squared Bias}

\centering
\footnotesize
\begin{tabular}{r|r|r|r|r|r|r|l}
\hline
$\sigma_u$ & Benchmark & MM & RC & Average & Non-ZI MM & 1 day & basis function\\
\hline
1 & 0.0150326 & 0.050386 & 0.035714 & 0.1008 & 0.2145 & 0.3878 & Bspline\\
2 & 0.0150326 & 0.114038 & 0.052613 & 0.1479 & 0.1762 & 0.4200 & Bspline\\
3 & 0.0150326 & 0.149900 & 0.081176 & 0.2269 & 569.5784 & 0.4411 & Bspline\\
1 & 107.4245267 & 0.508382 & 0.433005 & 0.1014 & 0.5080 & 0.3878 & Fourier\\
2 & 107.4245267 & 0.384697 & 0.239458 & 0.1443 & 0.8877 & 0.4200 & Fourier\\
3 & 107.4245267 & 0.953764 & 0.102823 & 0.2223 & 14.4792 & 0.4411 & Fourier\\
\hline
\end{tabular}

\end{subtable}

\vspace{1em}

\begin{subtable}{\textwidth}
\centering
\caption{Variance}

\centering
\footnotesize
\begin{tabular}{r|r|r|r|r|r|r|l}
\hline
$\sigma_u$ & Benchmark & MM & RC & Average & Non-ZI MM & 1 day & basis function\\
\hline
1 & 14.30998 & 12.30884 & 11.617732 & 1.9058 & 20.1516 & 0.3184 & Bspline\\
2 & 14.30998 & 15.06918 & 8.722044 & 1.5582 & 51.9469 & 0.2542 & Bspline\\
3 & 14.30998 & 45.31736 & 6.450979 & 1.2733 & 289780.6263 & 0.1953 & Bspline\\
1 & 133942.68329 & 199.04763 & 195.187175 & 13.2699 & 120.6445 & 0.3184 & Fourier\\
2 & 133942.68329 & 152.37204 & 87.743007 & 10.9018 & 253.1402 & 0.2542 & Fourier\\
3 & 133942.68329 & 515.57976 & 56.601526 & 8.6190 & 17416.4713 & 0.1953 & Fourier\\
\hline
\end{tabular}

\end{subtable}

\end{table}

\begin{table}[htbp]
\captionsetup{justification=raggedright}
\caption{Effect of correlation of u. {\normalsize{Common settings:} 
Bernoulli distribution (with logit link); 
smoothed point-wise; 
{sample size:} $n = 100$; 
{proportion of zero values:} $\mathbb{E}\{\Pr(W_{ij}(t)=0)\}=0.335$; 
$\sigma_u = 1$; 
$q_g=0$. 
{Benchmark} = benchmark method; 
{Average} = naïve average method; 
{MM} = mixed model based method proposed in Section~\ref{subsec:predict X}; 
{RC} = regression calibration method proposed in Section~\ref{subsec:predict X}; 
{Non-ZI MM} = non-zero-inflated mixed model based method; 
{1 day} = naive one-day method.}}
\centering

\begin{subtable}{\textwidth}
\centering
\caption{Squared Bias}

\centering
\footnotesize
\begin{tabular}{l|r|r|r|r|r|r|r|l}
\hline
correlation function $U$ & $\rho_u$ & Benchmark & MM & RC & Average & Non-ZI MM & 1 day & basis function\\
\hline
AR1 & 0.20 & 0.0150326 & 0.047030 & 0.039978 & 0.0961 & 0.3036 & 0.3709 & Bspline\\
AR1 & 0.40 & 0.0150326 & 0.053365 & 0.043068 & 0.0950 & 0.3248 & 0.3748 & Bspline\\
AR1 & 0.60 & 0.0150326 & 0.041493 & 0.031321 & 0.0969 & 0.3013 & 0.3794 & Bspline\\
Compound symmetry & 0.20 & 0.0150326 & 0.061318 & 0.058211 & 0.0929 & 0.3586 & 0.3701 & Bspline\\
Compound symmetry & 0.40 & 0.0150326 & 0.052949 & 0.044449 & 0.0946 & 0.4973 & 0.3711 & Bspline\\
Compound symmetry & 0.60 & 0.0150326 & 0.056266 & 0.052737 & 0.0941 & 0.3707 & 0.3706 & Bspline\\
Squared Exponential & 0.10 & 0.0150326 & 0.022744 & 0.017715 & 0.1019 & 0.2524 & 0.3819 & Bspline\\
Squared Exponential & 0.15 & 0.0150326 & 0.025846 & 0.019917 & 0.1041 & 0.2693 & 0.3866 & Bspline\\
Squared Exponential & 0.20 & 0.0150326 & 0.050386 & 0.035714 & 0.1008 & 0.2145 & 0.3878 & Bspline\\
AR1 & 0.20 & 107.4245267 & 0.283537 & 0.259335 & 0.0974 & 0.5523 & 0.3709 & Fourier\\
AR1 & 0.40 & 107.4245267 & 0.153780 & 0.121301 & 0.1061 & 0.5633 & 0.3748 & Fourier\\
AR1 & 0.60 & 107.4245267 & 0.348051 & 0.231356 & 0.1079 & 0.5313 & 0.3794 & Fourier\\
Compound symmetry & 0.20 & 107.4245267 & 0.222522 & 0.181032 & 0.0935 & 0.6425 & 0.3701 & Fourier\\
Compound symmetry & 0.40 & 107.4245267 & 0.162984 & 0.185270 & 0.0905 & 0.6127 & 0.3711 & Fourier\\
Compound symmetry & 0.60 & 107.4245267 & 0.144948 & 0.115386 & 0.0916 & 0.6691 & 0.3706 & Fourier\\
Squared Exponential & 0.10 & 107.4245267 & 0.389426 & 0.413785 & 0.1065 & 0.5471 & 0.3819 & Fourier\\
Squared Exponential & 0.15 & 107.4245267 & 0.240888 & 0.304061 & 0.1134 & 0.5281 & 0.3866 & Fourier\\
Squared Exponential & 0.20 & 107.4245267 & 0.508382 & 0.433005 & 0.1014 & 0.5080 & 0.3878 & Fourier\\
\hline
\end{tabular}

\end{subtable}

\vspace{1em}

\begin{subtable}{\textwidth}
\centering
\caption{Variance}

\centering
\footnotesize
\begin{tabular}{l|r|r|r|r|r|r|r|l}
\hline
correlation function $U$ & $\rho_u$ & Benchmark & MM & RC & Average & Non-ZI MM & 1 day & basis function\\
\hline
AR1 & 0.20 & 14.30998 & 9.189542 & 8.901574 & 1.9079 & 19.1846 & 0.3298 & Bspline\\
AR1 & 0.40 & 14.30998 & 9.560459 & 9.215958 & 1.8525 & 18.1525 & 0.3209 & Bspline\\
AR1 & 0.60 & 14.30998 & 9.900726 & 9.480119 & 1.8980 & 18.8789 & 0.3144 & Bspline\\
Compound symmetry & 0.20 & 14.30998 & 10.840769 & 10.624313 & 1.9477 & 18.3008 & 0.3291 & Bspline\\
Compound symmetry & 0.40 & 14.30998 & 10.940169 & 10.729000 & 1.8911 & 38.9483 & 0.3287 & Bspline\\
Compound symmetry & 0.60 & 14.30998 & 11.512976 & 11.052571 & 1.9159 & 18.2469 & 0.3316 & Bspline\\
Squared Exponential & 0.10 & 14.30998 & 8.566455 & 8.126454 & 1.8312 & 18.2983 & 0.3055 & Bspline\\
Squared Exponential & 0.15 & 14.30998 & 11.061876 & 10.623938 & 1.8191 & 18.4807 & 0.3106 & Bspline\\
Squared Exponential & 0.20 & 14.30998 & 12.308843 & 11.617732 & 1.9058 & 20.1516 & 0.3184 & Bspline\\
AR1 & 0.20 & 133942.68329 & 60.612574 & 58.359882 & 11.0185 & 94.0496 & 0.3298 & Fourier\\
AR1 & 0.40 & 133942.68329 & 62.163259 & 58.847516 & 10.8672 & 101.7753 & 0.3209 & Fourier\\
AR1 & 0.60 & 133942.68329 & 73.174854 & 70.377144 & 11.7778 & 119.6064 & 0.3144 & Fourier\\
Compound symmetry & 0.20 & 133942.68329 & 78.419842 & 75.447020 & 11.5590 & 95.3529 & 0.3291 & Fourier\\
Compound symmetry & 0.40 & 133942.68329 & 87.418195 & 85.217615 & 11.8844 & 103.7438 & 0.3287 & Fourier\\
Compound symmetry & 0.60 & 133942.68329 & 105.645114 & 103.465558 & 12.6293 & 108.1200 & 0.3316 & Fourier\\
Squared Exponential & 0.10 & 133942.68329 & 120.781684 & 121.525296 & 12.1252 & 114.2930 & 0.3055 & Fourier\\
Squared Exponential & 0.15 & 133942.68329 & 178.018432 & 175.560460 & 13.0771 & 115.7622 & 0.3106 & Fourier\\
Squared Exponential & 0.20 & 133942.68329 & 199.047625 & 195.187175 & 13.2699 & 120.6445 & 0.3184 & Fourier\\
\hline
\end{tabular}

\end{subtable}

\end{table}

\begin{table}[htbp]
\captionsetup{justification=raggedright}
\caption{Effect of proportion of zero values in W. {\normalsize{Common settings:} 
Bernoulli distribution (with logit link); 
smoothed point-wise; 
{sample size:} $n = 100$; 
{Covariance function of $U(t)$:} squared-exponential with $\rho_u=0.2$; 
$\sigma_u = 1$; 
$q_g=0$. 
{Benchmark} = benchmark method; 
{Average} = naïve average method; 
{MM} = mixed model based method proposed in Section~\ref{subsec:predict X}; 
{RC} = regression calibration method proposed in Section~\ref{subsec:predict X}; 
{Non-ZI MM} = non-zero-inflated mixed model based method; 
{1 day} = naive one-day method.}}
\centering

\begin{subtable}{\textwidth}
\centering
\caption{Squared Bias}

\centering
\footnotesize
\begin{tabular}{r|r|r|r|r|r|r|l}
\hline
E(Pr{W=0}) & Benchmark & MM & RC & Average & Non-ZI MM & 1 day & basis function\\
\hline
0.255 & 0.0150326 & 0.051167 & 0.040753 & 0.0645 & 0.1668 & 0.3503 & Bspline\\
0.294 & 0.0150326 & 0.041507 & 0.030263 & 0.0772 & 0.1833 & 0.3699 & Bspline\\
0.335 & 0.0150326 & 0.050386 & 0.035714 & 0.1008 & 0.2145 & 0.3878 & Bspline\\
0.403 & 0.0150326 & 0.060706 & 0.038519 & 0.1329 & 0.2759 & 0.4077 & Bspline\\
0.255 & 107.4245267 & 1.280440 & 0.935615 & 0.0807 & 0.3221 & 0.3503 & Fourier\\
0.294 & 107.4245267 & 0.399134 & 0.396307 & 0.0904 & 0.4478 & 0.3699 & Fourier\\
0.335 & 107.4245267 & 0.508382 & 0.433005 & 0.1014 & 0.5080 & 0.3878 & Fourier\\
0.403 & 107.4245267 & 0.248041 & 0.628551 & 0.1276 & 1.0988 & 0.4077 & Fourier\\
\hline
\end{tabular}

\end{subtable}

\vspace{1em}

\begin{subtable}{\textwidth}
\centering
\caption{Variance}

\centering
\footnotesize
\begin{tabular}{r|r|r|r|r|r|r|l}
\hline
E(Pr{W=0}) & Benchmark & MM & RC & Average & Non-ZI MM & 1 day & basis function\\
\hline
0.255 & 14.30998 & 13.32331 & 12.73832 & 2.5310 & 13.1404 & 0.4871 & Bspline\\
0.294 & 14.30998 & 12.81114 & 12.29866 & 2.2693 & 16.4019 & 0.3980 & Bspline\\
0.335 & 14.30998 & 12.30884 & 11.61773 & 1.9058 & 20.1516 & 0.3184 & Bspline\\
0.403 & 14.30998 & 10.66842 & 10.74905 & 1.5109 & 53.4614 & 0.2436 & Bspline\\
0.255 & 133942.68329 & 283.62098 & 265.61055 & 19.6025 & 95.7604 & 0.4871 & Fourier\\
0.294 & 133942.68329 & 239.78025 & 224.45952 & 15.8953 & 105.0335 & 0.3980 & Fourier\\
0.335 & 133942.68329 & 199.04763 & 195.18717 & 13.2699 & 120.6445 & 0.3184 & Fourier\\
0.403 & 133942.68329 & 134.04625 & 155.68044 & 9.8360 & 191.8466 & 0.2436 & Fourier\\
\hline
\end{tabular}

\end{subtable}

\end{table}

\begin{table}[htbp]
\captionsetup{justification=raggedright}
\caption{Effect of $q_g$ for $G(t)$ in Eq~\eqref{eq:sim G}. {\normalsize{Common settings:} 
Bernoulli distribution (with logit link); 
smoothed point-wise; 
{sample size:} $n = 100$; 
{proportion of zero values:} $\mathbb{E}\{\Pr(W_{ij}(t)=0)\}=0.335$; 
{Covariance function of $U(t)$:} squared-exponential with $\rho_u=0.2$; 
$\sigma_u = 1$; 
{Benchmark} = benchmark method; 
{Average} = naïve average method; 
{MM} = mixed model based method proposed in Section~\ref{subsec:predict X}; 
{RC} = regression calibration method proposed in Section~\ref{subsec:predict X}; 
{Non-ZI MM} = non-zero-inflated mixed model based method; 
{1 day} = naive one-day method.}}
\centering

\begin{subtable}{\textwidth}
\centering
\caption{Squared Bias}

\centering
\footnotesize
\begin{tabular}{r|r|r|r|r|r|r|l}
\hline
$q_g$ & Benchmark & MM & RC & Average & Non-ZI MM & 1 day & basis function\\
\hline
0.0 & 0.0150326 & 0.050386 & 0.035714 & 0.1008 & 0.2145 & 0.3878 & Bspline\\
0.2 & 0.0150326 & 0.045836 & 0.035390 & 0.1330 & 0.1476 & 0.3887 & Bspline\\
0.4 & 0.0150326 & 0.023270 & 0.032864 & 0.1829 & 0.1060 & 0.3876 & Bspline\\
0.0 & 107.4245267 & 0.508382 & 0.433005 & 0.1014 & 0.5080 & 0.3878 & Fourier\\
0.2 & 107.4245267 & 0.359648 & 0.440845 & 0.1545 & 0.4463 & 0.3887 & Fourier\\
0.4 & 107.4245267 & 0.206262 & 0.453843 & 0.1729 & 0.1486 & 0.3876 & Fourier\\
\hline
\end{tabular}

\end{subtable}

\vspace{1em}

\begin{subtable}{\textwidth}
\centering
\caption{Variance}

\centering
\footnotesize
\begin{tabular}{r|r|r|r|r|r|r|l}
\hline
$q_g$ & Benchmark & MM & RC & Average & Non-ZI MM & 1 day & basis function\\
\hline
0.0 & 14.30998 & 12.308843 & 11.61773 & 1.9058 & 20.1516 & 0.3184 & Bspline\\
0.2 & 14.30998 & 11.778162 & 11.59655 & 1.6796 & 10.8239 & 0.3191 & Bspline\\
0.4 & 14.30998 & 8.521128 & 10.96912 & 1.2983 & 4.2749 & 0.3277 & Bspline\\
0.0 & 133942.68329 & 199.047625 & 195.18717 & 13.2699 & 120.6445 & 0.3184 & Fourier\\
0.2 & 133942.68329 & 188.240261 & 191.84586 & 13.0345 & 76.2574 & 0.3191 & Fourier\\
0.4 & 133942.68329 & 119.301798 & 166.51689 & 10.4905 & 33.8754 & 0.3277 & Fourier\\
\hline
\end{tabular}

\end{subtable}

\end{table}

\end{CJK}
\end{document}